\documentclass[11pt]{article}

\newcommand{\magicSAT}{15\,051\ }
\newcommand{\magicNEARSAT}{843\ }
\usepackage[letterpaper,margin=1in]{geometry}

\usepackage[T1]{fontenc}
\usepackage{lmodern}
\usepackage{amsfonts}
\usepackage{graphicx}
\usepackage{epstopdf}
\usepackage{enumitem}
\usepackage{algorithmic}
\usepackage{stmaryrd}

\usepackage
{
        amssymb,
        amsmath,
        amsthm,
        times,
        nicefrac,
        xcolor
}

\usepackage{thmtools}
\usepackage{tikz}
\usepackage{scalerel}
\usepackage{bm}
\usetikzlibrary{arrows.meta}

\newcommand{\circRad}{0.90ex}   
\newcommand{\circLine}{0.08ex}  

\newcommand{\upHalf}{0.64ex}    
\newcommand{\upHead}{0.45ex}    
\newcommand{\upBase}{0.32ex}    

\DeclareRobustCommand{\oUp}{%
  \mathbin{%
    \scalerel*{%
      \tikz[baseline=-0.6ex, line width=\circLine]{
        \draw (0,0) circle (\circRad);
        \draw (0,-\upHalf) -- (0,\upHalf);
        \fill (0,\upHalf)
              -- (-\upBase,\upHalf-\upHead)
              -- (\upBase,\upHalf-\upHead)
              -- cycle;
      }%
    }{\oplus}%
  }%
}

\definecolor{darkred}{RGB}{180,0,0}
\definecolor{darkblue}{RGB}{0,0,200}

\usepackage[unicode,colorlinks=true,citecolor=black,filecolor=black,linkcolor=black,urlcolor=black, pdfpagelabels, plainpages=false]{hyperref}
\usepackage{bookmark}

\newcommand {\calI}    {{\cal{I}}}
\newcommand {\udsim}    {\sim_{ud}}
\newcommand {\Sn}    {\mathbb{S}_n}
\newcommand {\Sk}    {\mathbb{S}_k}
\newcommand{\eqdef}{\mathrel{\stackrel{\text{\tiny def}}{=}}}

\DeclareMathOperator {\Sat} {Sat}
\DeclareMathOperator {\udsign} {ud-sign}
\DeclareMathOperator{\image}{Im}

\newcommand {\brc}   [1] {\left(#1\right)}
\newcommand {\Exp}       {\mathbb{E}}

\newcommand {\Prob}  [1] {\Pr \brc{#1 }}

\newcommand{\NIdeal}{{\mathbb R}N^{\tau}}

\newcommand {\bbN}    {\mathbb{N}}

\DeclareMathOperator {\NFirst} {NFirst}
\DeclareMathOperator {\NLast}  {NLast}

\newcommand {\transpose}{\intercal}

\DeclareMathOperator {\conv}{conv}
\DeclareMathOperator {\patt}{patt}
\DeclareMathOperator {\partv}{part}

\DeclareMathOperator {\CSP}{CSP}
\DeclareMathOperator {\PCSP}{MaxPCSP}

\newtheorem{theorem}{Theorem}[section]
\newtheorem{lemma}[theorem]{Lemma}

\newtheorem{claim}[theorem]{Claim}
\newtheorem{corollary}[theorem]{Corollary}

\definecolor{nblue}{RGB}{173, 216, 230}
\theoremstyle{definition}
\declaretheorem[sibling=theorem, shaded={bgcolor=nblue!10}]{definition}
\declaretheorem[sibling=theorem, shaded={bgcolor=nblue!10}]{example}

\newtheorem{observation}[theorem]{Observation}
\newtheorem{fact}[theorem]{Fact}

\newtheorem{remark}[theorem]{Remark}

\title{Approximation algorithms for satisfiable and nearly satisfiable ordering CSPs}
\author{Yury Makarychev \\ Toyota Technological Institute at Chicago}
\date{}
\begin{document}
\maketitle

\begin{abstract}
We study approximation algorithms for satisfiable and nearly satisfiable instances of ordering constraint satisfaction problems (ordering CSPs). Ordering CSPs arise naturally in ranking and scheduling, yet their approximability remains poorly understood beyond a few isolated cases. Apart from tractable cases in which satisfiable instances can be solved exactly in polynomial time, prior nontrivial guarantees in the satisfiable regime were known only for Betweenness, while prior algorithms in the nearly satisfiable regime applied only to bounded-arity precedence CSPs -- CSPs whose constraints are conjunctions of clauses of the form $x_i < x_j$.

We introduce a general framework for designing approximation algorithms for ordering CSPs. The framework relaxes an input instance to an auxiliary ordering CSP, solves the relaxation, and then applies a randomized transformation to obtain an ordering for the original instance. This reduces the search for approximation algorithms to an optimization problem over randomized transformations.

Our main technical contribution is to show that the power of this framework is captured by a structured class of transformations, which we call strong IDU transformations: every transformation used in the framework can be replaced by a strong IDU transformation without weakening the resulting approximation guarantee. We then classify strong IDU transformations and show that optimizing over them reduces to an explicit optimization problem whose dimension depends only on the maximum predicate arity $k$ and the desired precision $\delta > 0$. As a consequence, for any finite ordering constraint language, we can compute a strong IDU transformation whose guarantee is within $\delta$ of the best guarantee achievable by the framework, in time depending only on $k$ and $\delta$.

The framework applies broadly and yields nontrivial approximation guarantees for a wide class of ordering predicates; the following arity-$4$ results illustrate its scope. Among NP-hard ordering CSPs defined by a single predicate of arity~4, we show that at least \magicSAT{} CSPs admit nontrivial approximation in the satisfiable regime. Moreover, among NP-hard and polynomial-time solvable single-predicate ordering CSPs of arity~4 -- excluding bounded-arity precedence CSPs -- at least \magicNEARSAT{} predicates admit nontrivial approximation in the nearly satisfiable regime. Specifically, given a $(1-\varepsilon)$-satisfiable instance, our algorithm produces an ordering that satisfies at least an $\alpha - O(\varepsilon \log n \log \log n)$ fraction of the constraints, where $\alpha > \alpha_{\mathrm{random}}$ depends only on the predicate.
\end{abstract}

\newpage

\tableofcontents
\newpage
\section{Introduction}
\subsection{Ordering CSPs}
In this paper, we develop a general framework for designing approximation algorithms for satisfiable and nearly satisfiable instances of ordering constraint satisfaction problems (ordering CSPs).
There is an extensive literature on the approximability of CSPs over finite domains, ranging from the classification of polynomially tractable CSPs~\cite{Bulatov,Zhuk,Schaefer}, to Raghavendra's general framework for approximating arbitrary maximization CSPs~\cite{Raghavendra}, to nontrivial approximation algorithms for NP-hard CSPs (see, e.g.,~\cite{GW,LLZ,ABZ05,CMM2009near,MM12,DBHPZ23,BNZ}), including algorithms for satisfiable and nearly satisfiable instances (see, e.g.,~\cite{GW,Z98,Trevisan,ACMM,CMM1,CMM2,bhangale2021optimal,BKM25}). In addition, there are results on approximation-resistant predicates~\cite{AM09,KTW14,Chan16}. See the survey~\cite{MMsurvey} for further details and references.  
However, there are only a few positive results on the approximability of ordering CSPs. Recall the definition of an ordering CSP.

\begin{definition}[Ordering CSP]
An \emph{ordering predicate} (or \emph{relation}) $\varphi(x_1, \dots, x_k)$ of arity $k$ is a disjunction of clauses of the form 
\[
x_{i_1} < x_{i_2} < \dots < x_{i_t}.
\]
where $i_1, \dots, i_t\in [k]$ are distinct and $t\geq 2$.
A \emph{constraint language} is a collection of such ordering predicates (different predicates in the language may have different arities).  

For a constraint language $\Pi$, we define the corresponding \emph{ordering CSP} $\CSP(\Pi)$ as follows.  
An instance of the problem consists of variables $x_1, \dots, x_n$ and a set of constraints over them, where each constraint is of the form 
\[
\varphi(x_{i_1}, \dots, x_{i_k}) \quad \text{for some } \varphi \in \Pi.
\]
A \emph{solution} to the instance is a permutation $\pi$ on $[n]$.  
A constraint $\varphi$ is said to be \emph{satisfied} by $\pi$ if the assignment $x_i = \pi(i)$ satisfies $\varphi$.  
The objective is to maximize the number of satisfied constraints.  

An instance is said to be \emph{$\alpha$-satisfiable} if there exists a solution satisfying at least an $\alpha$-fraction of the constraints.  
An instance is \emph{completely satisfiable} if it is $1$-satisfiable, that is, there exists a permutation satisfying all constraints.
\end{definition}

A key reason for the scarcity of positive results is the hardness result of Guruswami, Håstad, Manokaran, Raghavendra, and Charikar~\cite{GHMRC}, who showed that, under the Unique Games Conjecture, ordering CSPs are approximation resistant in the following sense: for every ordering CSP and every constant $\varepsilon > 0$, no polynomial-time algorithm can, given an instance that is at least $(1-\varepsilon)$-satisfiable, find a solution satisfying more than an $\alpha_{\mathrm{random}} + \varepsilon$ fraction of constraints, where $\alpha_{\mathrm{random}}$ denotes the expected fraction of constraints satisfied by a random ordering.
Thus, nontrivial approximation is ruled out for every fixed constant $\varepsilon>0$. However, this still leaves open completely satisfiable instances and, more generally, the asymptotic regime in which $\varepsilon$ tends to $0$ as $n\to \infty$. 

At the same time, a remarkable result by Bodirsky and Kára provides a classification of \emph{temporal} CSPs -- a generalization of ordering CSPs, precisely characterizing which temporal CSPs are polynomially tractable and which are NP-hard~\cite{BK10}. In particular, it yields a complete classification of polynomially tractable \emph{ordering} CSPs. Given an ordering CSP $\CSP(\Pi)$, the result of Bodirsky--Kára determines whether we can solve \emph{completely satisfiable} instances of $\CSP(\Pi)$ \emph{exactly} in polynomial time (assuming $\mathrm{P} \neq \mathrm{NP}$).

Against this background, two basic questions remain wide open:
\begin{itemize}
    \item \textbf{Nontrivial Approximation for Completely Satisfiable CSPs.} If an ordering CSP $P$ is not polynomially tractable, can we obtain a nontrivial approximation -- that is, one achieving a value strictly above $\alpha_{\mathrm{random}}$ -- for \emph{completely satisfiable} instances of $P$?
    \item \textbf{Nontrivial Approximation for Nearly Satisfiable CSPs.} For a given ordering CSP $P$ (whether or not polynomially tractable), can we obtain a nontrivial approximation for $(1 - \varepsilon)$-satisfiable instances of $P$, where $\varepsilon = 1/\mathrm{polylog}(n)$?
\end{itemize}

There are two known results addressing these questions.
First, there is a $\nicefrac{1}{2}$-approximation algorithm for completely satisfiable instances of \emph{Betweenness} (an ordering CSP of arity~3) due to Chor and Sudan~\cite{CS98}; a simpler combinatorial $\nicefrac{1}{2}$-approximation was later given in our previous work~\cite{Mak} (a random ordering achieves $\alpha_{\mathrm{random}} = \nicefrac{1}{3}$).
Second, there is an $O(\log n \log \log n)$-approximation algorithm for the \emph{Minimum Feedback Arc Set} problem: Seymour~\cite{Seymour} established an upper bound on the LP integrality gap, and Even, Naor, Schieber, and Sudan~\cite{ENSS98} used this to obtain a polynomial-time algorithm that explicitly constructs an ordering.
In CSP terms, this algorithm applies to nearly satisfiable instances of the ordering CSP with the less-than predicate $x_1 < x_2$: given a $(1-\varepsilon)$-satisfiable instance, it finds a solution satisfying at least a $1 - O(\varepsilon \log n \log \log n)$ fraction of the constraints.

\begin{definition}\label{def:precedence-CSP}
We say that a predicate is a \emph{precedence predicate} if it is a conjunction of clauses of the form $x_i < x_j$.
An ordering CSP $\CSP(\Pi)$ is a \emph{bounded-arity precedence CSP} if it has a finite constraint language and every predicate in $\Pi$ is a precedence predicate.
\end{definition}

The algorithm for the Minimum Feedback Arc Set problem extends to bounded-arity precedence CSPs: in $(1-\varepsilon)$-satisfiable instances, it satisfies at least a $1 - O(\varepsilon \log n \log \log n)$ fraction of constraints. The constant in the big-$O$ notation depends on the maximum arity of the predicates in the constraint language.

We also note that there are approximation results for ordering CSPs in several restricted settings, including bounded-occurrence ordering CSPs~\cite{GZ12,MakBounded} and approximation above the average~\cite{CMM07advantage,GIMY, MMZ15}.

\subsection{Overview of our framework}
The principal contribution of this work is a new framework for approximating ordering CSPs, together with structural results that make the framework algorithmically usable. We develop this framework to address the two questions above and \emph{conjecture that optimal approximation algorithms for both problems fall within it}. While we apply the framework here only to these questions, it can also be used in other settings, for example, to design approximation algorithms for bounded-occurrence ordering CSPs.

We apply our techniques to ordering CSPs of arity~4 defined by a single predicate. 
Among NP-hard single-predicate arity-4 ordering CSPs, we prove that at least \textit{\magicSAT{} non-isomorphic predicates} admit a nontrivial approximation for completely satisfiable instances. In the nearly satisfiable regime, at least \textit{\magicNEARSAT{} such predicates} admit nontrivial approximation (excluding bounded-arity precedence CSPs, see Definition~\ref{def:precedence-CSP}).

We now outline the main ideas underlying our framework and introduce the necessary definitions.

\begin{definition}
We say that a predicate $\varphi'$ is a \emph{relaxation} of a predicate $\varphi$ of the same arity if every assignment satisfying $\varphi$ also satisfies $\varphi'$, or, equivalently, if $\varphi(x_1, \dots, x_k)$ implies $\varphi'(x_1, \dots, x_k)$. 
Given a constraint language $\Pi$, fix a relaxation $\varphi'$ for every predicate $\varphi \in \Pi$,
and let $\Pi' = \{\varphi' : \varphi \in \Pi\}$ denote the resulting constraint language, which we call the \textit{relaxation} of $\Pi$. 

Consider the corresponding maximization promise constraint satisfaction problem $\PCSP(\Pi, \Pi')$ defined as follows: given an instance $\cal I$ of $\CSP(\Pi)$, replace every constraint $\varphi$ with its relaxation $\varphi'$ in $\Pi'$. Now, the goal is to find a solution maximizing the number of satisfied constraints in the resulting instance $\mathcal{I}'$ of $\CSP(\Pi')$; we will refer to $\mathcal{I}'$ as the relaxation of $\mathcal{I}$.
(We note that the correspondence $\varphi \mapsto \varphi'$ is implicit in our notation.)
\end{definition}

Observe that if the original instance is $\alpha$-satisfiable, then so is its relaxation. Importantly, if $\CSP(\Pi')$ is polynomially tractable (that is, completely satisfiable instances of  $\CSP(\Pi')$ can be solved exactly in polynomial time),\footnote{For ordering CSPs, there is no distinction between the polynomial-time tractability of the decision and search variants: if there is a polynomial-time algorithm for one, then there is also a polynomial-time algorithm for the other.} then a relaxation of every satisfiable instance of $\CSP(\Pi)$ can also be solved exactly in polynomial time.

We now introduce the notion of a transformation that, together with its strengthening, plays a key role in our framework. 
\begin{definition}
A \emph{weak IDU transformation} is a collection of \emph{random} permutations $\{\sigma_i\}_{i\in\bbN}$, with one random permutation $\sigma_n$ defined for every $n$. We will refer to this IDU transformation simply as $\bm{\sigma}$.
\end{definition}
This definition is only a starting point. We will soon impose additional properties and arrive at a much more structured class of transformations, called \emph{strong IDU transformations}.
We are now ready to outline our general approach.

\begin{enumerate}
    \item First, we identify an appropriate relaxation $\varphi'$ for every predicate $\varphi \in \Pi$, obtaining a relaxed constraint language $\Pi'$. While we will discuss this step in more detail later (see Sections~\ref{sec:satisfiable-csps} and~\ref{sec:nearly-satisfiable-csps}), a natural choice is to take predicates $\varphi'$ from a constraint language defining a polynomially tractable ordering CSP (as mentioned above, such languages are classified explicitly by Bodirsky and Kára).
    \item Given an instance of $\CSP(\Pi)$, consider its relaxation in $\CSP(\Pi')$. Solve this instance exactly or approximately using an appropriate algorithm for $\CSP(\Pi')$, and obtain a solution $\pi$.
    \item Apply a random IDU transformation to $\pi$ and output $\sigma_n \pi$ (the composition of permutations $\sigma_n$ and $\pi$). That is, each variable $x_i$ is assigned to position $\sigma_n(\pi(i))$, where $n$ is the number of variables in the instance.
\end{enumerate}
\begin{example}\label{example:algorithm}
We provide a simple example to illustrate this algorithm. Consider the constraint language with a single predicate $\varphi$ of arity 4 and the corresponding $\CSP(\varphi)$:
\begin{align*}
\varphi(x_1,x_2, x_3, x_4) &= 
(x_1< x_2 < x_3 < x_4) \vee (x_1 < x_3 < x_2 < x_4)\vee{}\\
&\phantom{{}={}} (x_2 < x_1 < x_4 < x_3) \vee (x_2 < x_4 < x_1 < x_3).
\end{align*}
Since $x_2 < x_4$ holds in all clauses, $\varphi'(x_1, x_2, x_3, x_4) = (x_2 < x_4)$ is a valid relaxation of $\varphi$. Note that $\varphi'$ is simply the less-than constraint with two dummy variables. Thus, solving satisfiable instances of $\CSP(\varphi')$ reduces to topological sorting. However, the obtained solution $\pi$ itself may not provide any approximation for the original instance. Therefore, we apply the following transformation: we randomly and independently assign each $x_i$ the value of either $\pi(i)$ or $\pi(i) + n$, with probability $\nicefrac{1}{2}$ for each choice. This yields an injective map $\pi':\{1,\dots,n\}\to \{1,\dots, 2n\}$ (see Figure~\ref{fig:transformation-example}). The final assignment $\pi''$ is obtained by replacing the values of $\pi'$ with their relative ranks in $\{\pi'(1),\dots,\pi'(n)\}$. Using the framework developed in this paper, it is easy to show that this algorithm satisfies an $\alpha = \nicefrac{1}{4}$ fraction of all constraints in expectation, while a random ordering satisfies only an $\alpha_{\mathrm{random}} = \nicefrac{1}{6}$ fraction of all constraints.
\end{example}
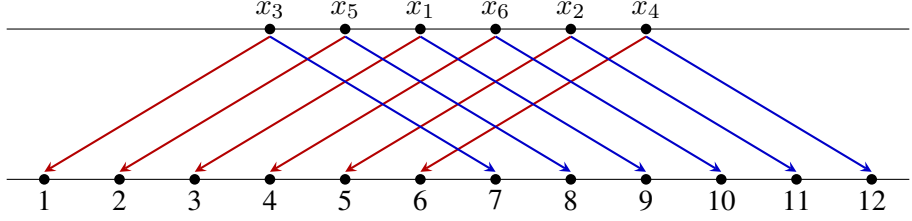
\begin{figure}
\centering
\begin{tikzpicture}[>=stealth,thick]
  \draw[thin] (-6,2) -- (6,2);
  \foreach \x/\lbl in {-2.5/$x_3$,-1.5/$x_5$,-0.5/$x_1$,0.5/$x_6$,1.5/$x_2$,2.5/$x_4$} {
    \fill (\x,2) circle (2pt);
    \node[above] at (\x,2) {\lbl};
  }

  \draw[thin] (-6,0) -- (6,0);
  \foreach \x/\lbl in {-5.5/1,-4.5/2,-3.5/3,-2.5/4,-1.5/5,-0.5/6,0.5/7,1.5/8,2.5/9,3.5/10,4.5/11,5.5/12} {
    \fill (\x,0) circle (2pt);
    \node[below] at (\x,0) {\lbl};
  }

  \draw[->,darkred]  (-2.5,1.9) -- (-5.5,0.1);
  \draw[->,darkred]  (-1.5,1.9) -- (-4.5,0.1);
  \draw[->,darkred]  (-0.5,1.9) -- (-3.5,0.1);
  \draw[->,darkred]  (0.5,1.9) -- (-2.5,0.1);
  \draw[->,darkred]  (1.5,1.9) -- (-1.5,0.1);
  \draw[->,darkred]  (2.5,1.9) -- (-0.5,0.1);

  \draw[->,darkblue] (-2.5,1.9) -- (0.5,0.1);
  \draw[->,darkblue] (-1.5,1.9) -- (1.5,0.1);
  \draw[->,darkblue] (-0.5,1.9) -- (2.5,0.1);
  \draw[->,darkblue] (0.5,1.9) -- (3.5,0.1);
  \draw[->,darkblue] (1.5,1.9) -- (4.5,0.1);
  \draw[->,darkblue] (2.5,1.9) -- (5.5,0.1);
\end{tikzpicture}
\caption{The figure shows the transformation from Example~\ref{example:algorithm}. Each $x_i$ is randomly mapped to either $\pi(i)$ or $\pi(i)+n$.}
\label{fig:transformation-example}
\end{figure}

\paragraph{From weak to strong IDU transformations}
The most nontrivial step in our framework is choosing an appropriate IDU transformation $\bm{\sigma}$. Assume that our algorithm for $\CSP(\Pi')$ satisfies an $s$-fraction of all constraints in Step~2. We first derive a formula for the fraction of constraints in the original $\CSP(\Pi)$ instance that are satisfied by our algorithm. We show that this fraction is at least $p(\Pi' \to \Pi, \bm{\sigma}) \, s$, where $p(\Pi' \to \Pi, \bm{\sigma})$ is characterized in terms of certain \emph{restricted pattern densities} of permutations in $\bm{\sigma}$.

However, the set of all weak IDU transformations is enormous. In general, weak IDU transformations do not admit finite -- let alone succinct -- descriptions, so there is no straightforward way to optimize over them in order to find an optimal or even nearly optimal transformation. To address this issue, we introduce a much more structured class of \emph{strong} IDU transformations. For this class, the formula for $p(\Pi' \to \Pi, \bm{\sigma})$ takes a simpler form, expressed in terms of standard (not restricted) \emph{pattern densities}. This naturally raises the question of whether strong IDU transformations are as powerful as general ones for our purposes. We answer this affirmatively by showing that for every weak IDU transformation $\bm{\sigma}$, there exists a strong IDU transformation $\bm{\sigma}'$ such that
\[
p(\Pi' \to \Pi, \bm{\sigma}') \ge p(\Pi' \to \Pi, \bm{\sigma})
\]
simultaneously for all $\Pi$ and $\Pi'$. Hence, for the purposes of our framework, there is no loss of generality in \emph{restricting attention to strong IDU transformations}.

We then classify all strong IDU transformations, giving a concrete description of the transformations that matter for our framework.
We show that each strong IDU transformation can be represented as a random permuton constructed from the basic permutons $I$ (identity), $D$ (decreasing), and $U$ (uniform) via a simple operation called \emph{up-combination}, which we introduce in this paper. Moreover, every strong IDU transformation can be approximated arbitrarily well by finite up-combinations of $I$, $D$, and $U$ or even by combinations of $I$ and $D$ alone. Using these notions, we develop an algebraic characterization of $p(\Pi' \to \Pi, \bm{\sigma})$, thereby reducing the problem of maximizing $p(\Pi' \to \Pi, \bm{\sigma})$ to that of maximizing certain multivariate polynomials. As a corollary, we show how to find a nearly optimal strong IDU transformation $\bm{\hat\sigma}$ such that
\[
p(\Pi' \to \Pi, \bm{\hat\sigma}) \ge p(\Pi' \to \Pi, \bm{\sigma}) - \varepsilon
\]
for every weak IDU transformation $\bm{\sigma}$, in time $f(k, 1/\varepsilon)$, where $k$ is the maximum arity of predicates in $\Pi$.

We provide a technical overview of IDU transformations in Section~\ref{sec:tech-IDU}. 
The remaining ingredient is the choice of the relaxation $\Pi'$.

\paragraph{Relaxation for completely satisfiable instances}
By the Bodirsky--Kára classification, an \textit{ordering} constraint satisfaction problem $\CSP(\Pi)$ is polynomially tractable if and only if all predicates in $\Pi$ are shuffle-closed or dual-shuffle-closed (we provide the necessary background in Section~\ref{sec:satisfiable-csps}; these details are not essential for the current discussion). Accordingly, there are two \emph{canonical tractable relaxations} of $\CSP(\Pi)$: one with $\Pi' = \Pi_{L}$, obtained by replacing each predicate $\varphi$ with its $L$-relaxation $\varphi_{L}$, namely the closure of $\varphi$ under the shuffle-operation, and the other with $\Pi' = \Pi_{R}$, obtained by replacing each predicate with its $R$-relaxation $\varphi_{R}$, namely the closure of $\varphi$ under the dual-shuffle-operation. Both $\Pi_{L}$ and $\Pi_{R}$ can be efficiently computed from $\Pi$. As the corresponding CSPs -- $\CSP(\Pi_{L})$ and $\CSP(\Pi_{R})$ -- can be solved exactly in polynomial time, the only remaining step is to choose an appropriate IDU transformation.

\paragraph{Relaxation for nearly satisfiable instances}
The only known algorithm for nearly satisfiable instances of any ordering CSP is the algorithm for the Minimum Feedback Arc Set problem~\cite{ENSS98,Seymour}. Using this algorithm, we can satisfy a $1 - O(\varepsilon \log n \log \log n)$ fraction of constraints in $(1-\varepsilon)$-satisfiable instances of bounded-arity precedence CSPs (see Definition~\ref{def:precedence-CSP}). Accordingly, given a $(1-\varepsilon)$-satisfiable instance of $\CSP(\Pi)$, we replace each predicate $\varphi \in \Pi$ with a precedence-constraint relaxation $\varphi'$; among all possible choices, we select the one with the minimal set of satisfying assignments (by inclusion). It is easy to see that such a predicate $\varphi'$ exists and that every other precedence-constraint relaxation of $\varphi$ is also a relaxation of $\varphi'$. We thus obtain an instance of $\CSP(\Pi')$, solve it using the approximation algorithm for the Minimum Feedback Arc Set problem, and obtain a solution $\pi$ satisfying at least a $1 - O(\varepsilon \log n \log \log n)$ fraction of constraints. Finally, we apply an appropriate IDU transformation to $\pi$ and output the resulting ordering.

We formally present our results for ordering CSPs, which rely on our framework, in Section~\ref{sec:results-csp}.

\paragraph{Related work} 
The most important and relevant results for ordering CSPs are the aforementioned ones: the classification of ordering CSPs by Bodirsky and Kára~\cite{BK10}, the hardness of approximation result for ordering CSPs by Guruswami, Håstad, Manokaran, Raghavendra, and Charikar~\cite{GHMRC}, and the $\nicefrac{1}{2}$-approximation algorithm by Chor and Sudan~\cite{CS98}.

In terms of the approach, our framework is an extension of the algorithm for Betweenness proposed by us in~\cite{Mak}. In~\cite{Mak}, we proposed to relax Betweenness constraints, solve the resulting CSP, and then run a postprocessing step that transforms the solution for the relaxation into one for Betweenness. Importantly, the latter transformation step in~\cite{Mak} was ad hoc and can only in retrospect be interpreted as equivalent to applying a strong IDU transformation (namely, $\frac{1}{2}I \oUp \frac{1}{2}D$).

The objective of our framework is analogous to that of Raghavendra for not-necessarily-satisfiable instances of classical CSPs~\cite{Raghavendra} and to the ongoing project by Bhangale, Khot, and Minzer for completely satisfiable instances (see~\cite{BKM25} and references therein). However, in comparison to these results, we focus on algorithms rather than hardness. Our negative results only show that certain CSPs cannot be solved using our framework and do not rule out the existence of other polynomial-time algorithms.

On the positive side, our algorithm for computing a near-optimal approximation factor within our framework is practical: for most CSPs of arity~4, its Python implementation is able to find a good rounding scheme in seconds to minutes on a modern laptop.

Finally, we note that the idea of replacing constraints with new constraints (their relaxations in our terminology) also appears in the algorithm for satisfiable instances of Max $k$-LIN in non-Abelian groups by Bhangale and Khot~\cite{bhangale2021optimal} (the authors attribute the algorithm to folklore): their algorithm replaces the $k$-LIN equations over the original non-Abelian group with corresponding constraints in the Abelian quotient group $G/[G,G]$ (where $[G,G]$ is the commutator subgroup of $G$), solves the resulting instance exactly using Gaussian elimination, and finally ``randomly lifts'' the solution to $G$.

\section{Technical overview of IDU transformations}
\label{sec:tech-IDU}
Our technical focus is on IDU transformations.  
We first discuss the properties of IDU transformations that are important for our applications.

\begin{definition}
For a predicate $\varphi$ of arity $k$, let $\Sat(\varphi)$ denote the set of all 
permutations on $[k]$ that satisfy~$\varphi$.  
\end{definition}

Note that $\varphi'$ is a relaxation of $\varphi$ if and only if $\Sat(\varphi) \subseteq \Sat(\varphi')$.

\begin{definition}\label{def:pattern}
Given a permutation $\pi$ on $[n]$ and a subset of indices $J = \{j_1,\dots, j_k\} \subseteq[n]$ with $j_1 < j_2 < \dots < j_k$,  
we define the \emph{pattern} of $\pi$ at positions $J$ as the permutation $\tau$ on $[k]$ that is order-isomorphic to the sequence $\pi(j_1), \pi(j_2), \dots, \pi(j_k)$.  
In other words, for every $a,b \in [k]$, we have $\tau(a) < \tau(b)$ if and only if $\pi(j_a) < \pi(j_b)$.  
We denote it by $\patt(\pi, J)$.  

Further, given a sequence of points \((x_i, y_i)\) with all \(x_i\) and all \(y_i\) distinct, the \emph{pattern} of the sequence is the permutation \(\tau \in \Sk\) that is order-isomorphic to the sequence of \(y\)-coordinates obtained after sorting the points by increasing \(x\)-coordinates. This definition also applies to sequences of distinct numbers: the pattern of \((a_1, \dots, a_k)\) is defined as the pattern of the points \((1, a_1), (2, a_2), \dots, (k, a_k)\).
\end{definition}

Consider a constraint $\varphi(x_{j_{\beta(1)}}, \dots, x_{j_{\beta(k)}})$,
where $j_1 < \dots < j_k$ and  $\beta$ is a permutation on $[k]$.  
Let $J = \{j_1,\dots, j_k\}$.
The constraint is satisfied by a permutation $\pi$  if and only if the pattern of the sequence $(\pi(j_{\beta(1)}), \dots, \pi(j_{\beta(k)}))$ is in $\Sat(\varphi)$. This can be written as $\patt(\pi, J)\,\beta \in \Sat(\varphi)$ or, equivalently, 
\begin{equation}\label{eq:condition-satisfied}
\patt(\pi, J) \in \Sat(\varphi) \beta^{-1} = \{\rho\beta^{-1}: \rho\in \Sat(\varphi)\}.
\end{equation}

Going back to our algorithm, assume we use a weak IDU transformation $\{\sigma_n\}$ to solve an instance $\calI$ of $\CSP(\Pi)$.  
Suppose we have found a solution $\pi$ for the relaxation $\calI'$ of $\calI$.  
Consider a constraint $\varphi(x_{j_{\beta(1)}}, \dots, x_{j_{\beta(k)}})$.  
Assume that its relaxation $\varphi'(x_{j_{\beta(1)}}, \dots, x_{j_{\beta(k)}})$ is satisfied by $\pi$.
Our goal is to analyze the probability that  $\varphi(x_{j_{\beta(1)}}, \dots, x_{j_{\beta(k)}})$ is satisfied by the permuted solution $\sigma_n \pi$.
Note that
\[
\patt(\sigma_n\pi, J) = \patt(\sigma_n, J') \, \patt(\pi, J),
\]
where $J' = \pi(J) = \{\pi(j_1), \dots, \pi(j_k)\}$. Accordingly,
\begin{align*}
    \Pr(\varphi(x_{j_{\beta(1)}}, \dots, x_{j_{\beta(k)}}) \text{ is satisfied}) 
    &= \Pr(\patt(\sigma_n\pi, J) \in \Sat(\varphi)\beta^{-1}) \\
    &= \Pr(\patt(\sigma_n, J') \cdot \patt(\pi, J) \in \Sat(\varphi)\beta^{-1})\\
    &= \Pr(\patt(\sigma_n, J') \in \Sat(\varphi) \cdot (\patt(\pi, J)\beta)^{-1})\\
    &= \Pr(\patt(\sigma_n, J') \in \Sat(\varphi) \cdot \tau_0^{-1}),
\end{align*}
where $\tau_0 = \patt(\pi, J)\beta$.
Since $\varphi'$ is satisfied by $\pi$, by~\eqref{eq:condition-satisfied}, we have $\tau_0 \in \Sat(\varphi')$. We get,
\[
\Pr(\varphi(x_{j_{\beta(1)}}, \dots, x_{j_{\beta(k)}}) \text{ is satisfied}) 
    \ge \min_{\tau\in \Sat(\varphi')} 
       \Pr\!\big(\patt(\sigma_n, J') \in \Sat(\varphi)\tau^{-1}\big).
\]

\begin{definition}    
Define the \emph{restricted pattern density} $d_J(\tau,\sigma_n)$ as the probability that $\patt(\sigma_n, J) = \tau$ (over the randomness of $\sigma_n$), and let the (standard) pattern density $d(\tau, \sigma_n)$ be its expectation over a uniformly random $k$-subset $J \subseteq [n]$ (i.e., $J$ is sampled uniformly without replacement from $\{1,\dots,n\}$). We will often refer to both quantities simply as pattern densities.
\end{definition}

Then the probability that the constraint $\varphi(x_{j_{\beta(1)}}, \dots, x_{j_{\beta(k)}})$ is satisfied is at least
\[
p(\varphi' \to \varphi, \bm{\sigma}) \eqdef 
\inf_{n\geq k}\ \min_{\substack{J\subseteq[n]\\ |J| = k}}\ 
\min_{\tau\in \Sat(\varphi')} \sum_{\rho\in \Sat(\varphi)\tau^{-1}} d_J(\rho, \sigma_n).
\]
Let $p(\Pi' \to \Pi, \bm{\sigma}) = \min_{\varphi\in \Pi} p(\varphi' \to \varphi, \bm{\sigma})$.  
Since the bound above holds for every satisfied constraint $\varphi'(x_{j_{\beta(1)}}, \dots, x_{j_{\beta(k)}})$, the expected number of satisfied constraints in the final solution $\sigma_n \pi$ is at least $p(\Pi' \to \Pi, \bm{\sigma})$ times the number of satisfied constraints in~$\pi$.  
The above argument yields the following performance guarantee for our framework.

\begin{theorem}\label{thm:idu-framework}
Consider $\CSP(\Pi)$ and a relaxation $\Pi'$ of $\Pi$.  
Assume that there is a polynomial-time algorithm that, given a $c$-satisfiable instance of $\CSP(\Pi)$, finds a solution satisfying at least an $s$ fraction of constraints in the corresponding relaxed instance of $\CSP(\Pi')$.  
Further, assume that there exists a weak IDU transformation $\{\sigma_n\}$ that can be sampled in polynomial time.  
Then there is a polynomial-time algorithm that, given a $c$-satisfiable instance of $\CSP(\Pi)$, finds a solution that satisfies, in expectation, at least a $p(\Pi' \to \Pi, \bm{\sigma}) \cdot s$ fraction of constraints.
\end{theorem}

In our applications of this theorem, we will have $c = 1$ and $s = 1$ for completely satisfiable instances and $c = 1 - \varepsilon$ and $s = 1 - O(\varepsilon \log n \log \log n)$ for nearly satisfiable instances (as discussed above).  
Theorem~\ref{thm:idu-framework} highlights the importance of the restricted pattern densities $d_J(\rho, \sigma_n)$, which fully describe the performance of the weak IDU transformation $\bm{\sigma}$ in our framework. However, analyzing these quantities for an arbitrary weak IDU transformation appears to be quite challenging.
It is considerably easier to handle them within a more structured family of transformations -- namely, those for which the values of $d_J(\rho, \sigma_n)$ do not depend on the choice of $J$. This observation motivates the following definition.

\begin{definition}
A strong \emph{IDU transformation} is a weak IDU transformation $\{\sigma_n\}$ that satisfies the following property: for every $k$ and permutation $\rho \in \Sk$, the value $d_J(\rho, \sigma_n)$ is the same for all $n \ge k$ and all $J \subseteq [n]$ of size $k$.  
In particular, $d_J(\rho, \sigma_n) = d(\rho,\sigma_n)$.  
Since this quantity does not depend on $n$ (for $n \ge k$), we denote it simply by $d(\rho, \bm{\sigma})$.
\end{definition}

For a strong IDU transformation, we obtain a simpler expression for $p(\varphi' \to \varphi, \bm{\sigma})$:
\[
p(\varphi' \to \varphi, \bm{\sigma}) = \min_{\tau\in \Sat(\varphi')} \sum_{\rho \in \Sat(\varphi)\tau^{-1}} d(\rho, \bm{\sigma}).
\]
However, it is not immediately clear how rich the class of strong IDU transformations is, or whether they are as powerful as weak IDU transformations for our purposes. 

The main structural results about IDU transformations are as follows.
First, we connect weak IDU transformations with \emph{random permutons}, the limit objects for random permutations (see Section~\ref{sec:permutons} for the definition of a permuton). Each random permuton defines a weak IDU transformation, and although the converse is not true, we show that such transformations are as powerful for our purposes as arbitrary weak IDU transformations.

We then show that for every weak IDU transformation $\bm{\sigma}$, there exists a strong IDU transformation $\hat{\bm{\sigma}}$ such that for any predicate $\varphi$ and any relaxation $\varphi'$ of $\varphi$:
\[
p(\varphi' \to \varphi, \hat{\bm{\sigma}}) \geq p(\varphi' \to \varphi, \bm{\sigma}).
\]

We next present a concrete classification of strong IDU transformations.
We begin by noting that the three basic permutons -- $I$ (increasing), $D$ (decreasing), and $U$ (uniform or Lebesgue) -- each define a strong IDU transformation (see Figure~\ref{fig:permutons}).
Specifically, the transformation $I$ yields the identity permutation on $[n]$ (thereby leaving every ordering CSP solution unchanged).
The transformation $D$ produces the decreasing permutation on $[n]$ (thus reversing every solution).
Finally, $U$ outputs a permutation drawn uniformly at random from $\Sn$ (the transformed solution is uniformly random).

We then introduce a new composition operation for random permutons and IDU transformations called the \emph{up-combination}.\footnote{We note that strong IDU transformations also form a semigroup under the standard permutation composition operation, but this observation does not appear to be useful in our setting.}  
Loosely speaking, the up-combination of two weak IDU transformations $\alpha$ and $\beta$ with weights $p$ and $1-p$ is obtained as follows:  
\begin{enumerate}
    \item each number from $1$ to $n$ is independently assigned to group $A$ with probability $p$ and to group $B$ with probability $1-p$;  
    \item we apply $\alpha$ to order the elements in $A$ and $\beta$ to order those in $B$;  
    \item finally, we place all elements of $A$ before all elements of $B$.
\end{enumerate}
We denote the resulting transformation by $p \alpha \oUp (1-p) \beta$.  
The up-combination operation differs significantly from the standard permutation compositions $\oplus$ and $\ominus$ commonly used in the literature.%
\footnote{One notable difference is that the operations $\oplus$ and $\ominus$ apply also to permutations (not just permutons or distributions over permutations) and produce a single permutation, whereas the up-combination is defined only for permutons or distributions over permutations. While one could extend it to pairs of permutations of \emph{equal length}, the result would still be \emph{a distribution over permutations}.}
Its key property is that the up-combination of any number of strong IDU transformations is again a strong IDU transformation (see Claim~\ref{claim:up-combination-is-strong-IDU}).  
Thus, starting from the basic transformations $I$, $D$, and $U$, we can construct a wide variety of new ones.

More surprisingly, every strong IDU transformation is a finite or infinite (in a suitable sense) up-combination of $I$, $D$, and $U$ (see Theorem~\ref{thm:main-strongIDU-classification}).  
From a practical standpoint, this implies that for every fixed arity $k$ and desired precision $\delta>0$, every strong IDU transformation can be approximated within~$\delta$ by a finite up-combination of $I$, $D$, and $U$ or even just $I$ and $D$ (see Theorems~\ref{thm:approx-by-IDU} and~\ref{thm:approx-by-ID}).  
A random permutation from such a combination can be sampled in time $f(k, \delta)\,n$, that is, time linear in $n$ when $k$ and $\delta$ are fixed.  
Hence, the additional computational overhead of solving $\CSP(\Pi)$ relative to its relaxation $\CSP(\Pi')$ is linear in $n$ (for fixed $\delta>0$ and finite $\Pi$).  
Moreover, the algorithm from Theorem~\ref{thm:idu-framework} can be derandomized using the method of conditional expectations.

Relying on the classification above, we study the \emph{profiles} of strong IDU transformations, defined as the vectors $\{d(\rho, \bm{\sigma})\}_{\rho\in \Sk}$ (with one coordinate for each $\rho \in \Sk$).  
Such profiles can also be defined for arbitrary deterministic and random permutons $\mu$.  
There is active research on permutation pattern densities (that is, the profiles of deterministic and random permutons), though the full landscape of possible densities remains far from understood.  
While every strong IDU transformation yields a valid permuton profile, the converse is not true.  
This naturally leads to the question of which random permuton profiles can be achieved by strong IDU transformations.

\begin{definition}[see~\cite{Niven68}]\label{def:ud-signature}
The \emph{up--down signature} of a permutation $\rho\in \Sk$ is a sequence $s$ of $k-1$ letters $u$ (for ``up'') and $d$ (for ``down''), defined as follows:  
$s_i = u$ if $\rho(i) < \rho(i+1)$ and $s_i = d$ otherwise.  
We denote the up--down signature of $\rho$ by $\udsign(\rho)$.  
We write $\rho_1 \udsim \rho_2$ if $\rho_1$ and $\rho_2$ have the same signature (and hence the same length).
\end{definition}

The following theorem characterizes the profiles of strong IDU transformations and provides a criterion for when a weak IDU transformation is strong.
\begin{theorem}
I. Consider a strong IDU transformation $\bm{\sigma}$. For every $k$ and every two permutations $\rho_1, \rho_2\in \Sk$ with $\rho_1^{-1}\udsim \rho_2^{-1}$, we have $d(\rho_1, \bm{\sigma}) = d(\rho_2, \bm{\sigma})$.  

II. Conversely, consider a weak IDU transformation $\bm{\sigma}$.  Assume that for all $n$ and $k \leq n$, the pattern density $d(\rho, \sigma_n)$ depends only on the up--down signature of $\rho^{-1}$ for $\rho \in \Sk$. Then $\bm{\sigma}$ is a strong IDU transformation.
\end{theorem}
This result is stated as Theorems~\ref{thm:strong-IDU-inv-sig-invariant} and~\ref{thm:ud-invariance-implies-strong} in Section~\ref{sec:up--down}.
In particular, these theorems show that for every $k$, the profile of a strong IDU transformation has only $2^{k-1}-1$ degrees of freedom.  

At this point, the problem of finding an optimal strong IDU transformation can be stated purely in the language of (standard) pattern densities $d(\cdot, R)$:
\begin{quote}
Find a random permuton $R$ that maximizes
\[
\min_{\varphi \in \Pi} p(\varphi' \to \varphi, R) =\min_{\varphi \in \Pi}\min_{\tau \in \Sat(\varphi')} \sum_{\rho \in \Sat(\varphi)} d(\rho \tau^{-1}, R).
\]
\noindent
subject to $d(\rho_1, R) = d(\rho_2, R)$ for all $\rho_1$ and $\rho_2$ with $\rho_1^{-1}\udsim \rho_2^{-1}$.
\end{quote}

We also develop two complementary tools for bounding and analyzing $p(\varphi' \to \varphi, \bm{\sigma})$.  
First, we show that profiles can be expressed using \emph{quasisymmetric polynomials} (see Section~\ref{sec:qsym} for details).  
This algebraic viewpoint provides a concise description of the feasible space of profiles and yields the following result (stated as Theorem~\ref{thm:compute-nearly-optimal-IDU} in Section~\ref{sec:computation}).

\begin{theorem}
There is an algorithm that, given a constraint language $\Pi$, its relaxation $\Pi'$, and a parameter $\delta > 0$, finds a strong IDU transformation $\bm{\sigma}$ with  
$p(\Pi' \to \Pi, \bm{\sigma}) \ge p(\Pi' \to \Pi, \bm{\sigma}^*) - \delta$,  
where $\bm{\sigma}^*$ is the optimal strong IDU transformation.  
We assume $\Pi$ is finite and has maximum arity $k$.  
The running time of the algorithm depends only on $k$ and $\delta$ (and not on $n$).
\end{theorem}

Second, we explore analytic techniques for proving lower bounds on $p(\varphi' \to \varphi, \bm{\sigma})$.  
To this end, we adapt Razborov's flag algebra framework for pattern densities to the setting of strong IDU transformations.

Finally, we illustrate the framework on ordering CSPs of arity~3 and~4 defined by a single predicate.
We prove that the only ordering CSPs of arity~3 that admit nontrivial approximation for completely satisfiable instances, using our approach, are $\CSP(\mathrm{Betweenness})$ and $\CSP(\{\mathrm{Betweenness}, <\})$ (see Section~\ref{sec:arity-3-analysis}).
Using numerical optimization, we further show that at least \magicSAT{} non-isomorphic ordering CSPs of arity~4, each defined by a single predicate, admit nontrivial approximation for completely satisfiable instances (this number includes only NP-hard ordering CSPs). Moreover, at least \magicNEARSAT{} ordering CSPs defined by a single predicate admit nontrivial approximation in the nearly satisfiable regime (this number excludes bounded-arity precedence CSPs).

\section{Applying the framework to satisfiable and nearly satisfiable ordering CSPs}
\label{sec:results-csp}
\subsection{Results for single-predicate ordering CSPs of arity~3 and~4}
We first briefly describe our results for single-predicate ordering CSPs of arity~3 and~4.

Let us say that two predicates are \emph{isomorphic} if one can be obtained from the other by renaming the variables or by taking the dual (i.e., replacing $<$ with $>$ in all formulas). For example, the following predicates are isomorphic:
$\varphi_1(x_1, x_2, x_3) = (x_1 < x_2) \wedge (x_1 < x_3)$,
$\varphi_2(x_1, x_2, x_3) = (x_2 < x_1) \wedge (x_2 < x_3)$, and
$\varphi_3(x_1, x_2, x_3) = (x_3 > x_1) \wedge (x_3 > x_2)$. 
We call each of the $L$-, $R$-, and $\varepsilon$-relaxations of a predicate $\varphi$ \emph{nontrivial} if it is neither equal to $\varphi$ nor to the constant-true predicate.

Our results are shown in Figure~\ref{fig:summary-3-4-results}.
The main takeaway is that the arity~3 landscape is essentially exhausted, whereas arity~4 already contains many positive examples in both regimes.

Excluding the two constant predicates, there are \(11\) non-isomorphic predicates of arity~3. Among them, \(4\) are polynomial-time tractable, and \(3\) of these are precedence predicates (see Definition~\ref{def:precedence-CSP}). Consequently, these \(3\) predicates admit a \(1 - O(\varepsilon \log n \log \log n)\) approximation in the nearly satisfiable regime. Of the remaining \(7\) predicates, only \(3\) have a nontrivial $L$- or $R$-relaxation. Among these \(3\), only one -- Betweenness -- admits a nontrivial approximation guarantee using our framework. Excluding the \(3\) precedence predicates, none of the remaining \(8\) predicates has a nontrivial $\varepsilon$-relaxation.

Excluding the two constant predicates, there are \(355{\,}046\) non-isomorphic predicates of arity \(4\). Among them, \(29\) are polynomial-time tractable, and \(11\) of these are precedence predicates. Consequently, these \(11\) predicates admit a \(1 - O(\varepsilon \log n \log\log n)\) approximation in the nearly satisfiable regime. Among the remaining \(355{\,}017\) predicates, \(39{\,}299\) have a nontrivial $L$- or $R$-relaxation. At least \(\magicSAT\) of them admit a nontrivial approximation guarantee in the satisfiable regime.
Also, \(993\) predicates have a nontrivial $\varepsilon$-relaxation. Among them, at least \(\magicNEARSAT\) admit a nontrivial approximation guarantee in the nearly satisfiable regime. (Note that all of these \(\magicNEARSAT\) predicates also admit a nontrivial approximation in the satisfiable regime.)

\paragraph{Dataset and verification code}
A database of predicates with their approximation factors, together with descriptions of all relaxations, up-down signatures, and IDU combinations used in the paper, as well as code for verifying the correctness of these results, will be made publicly available at~\cite{ordering-csp-repo}.

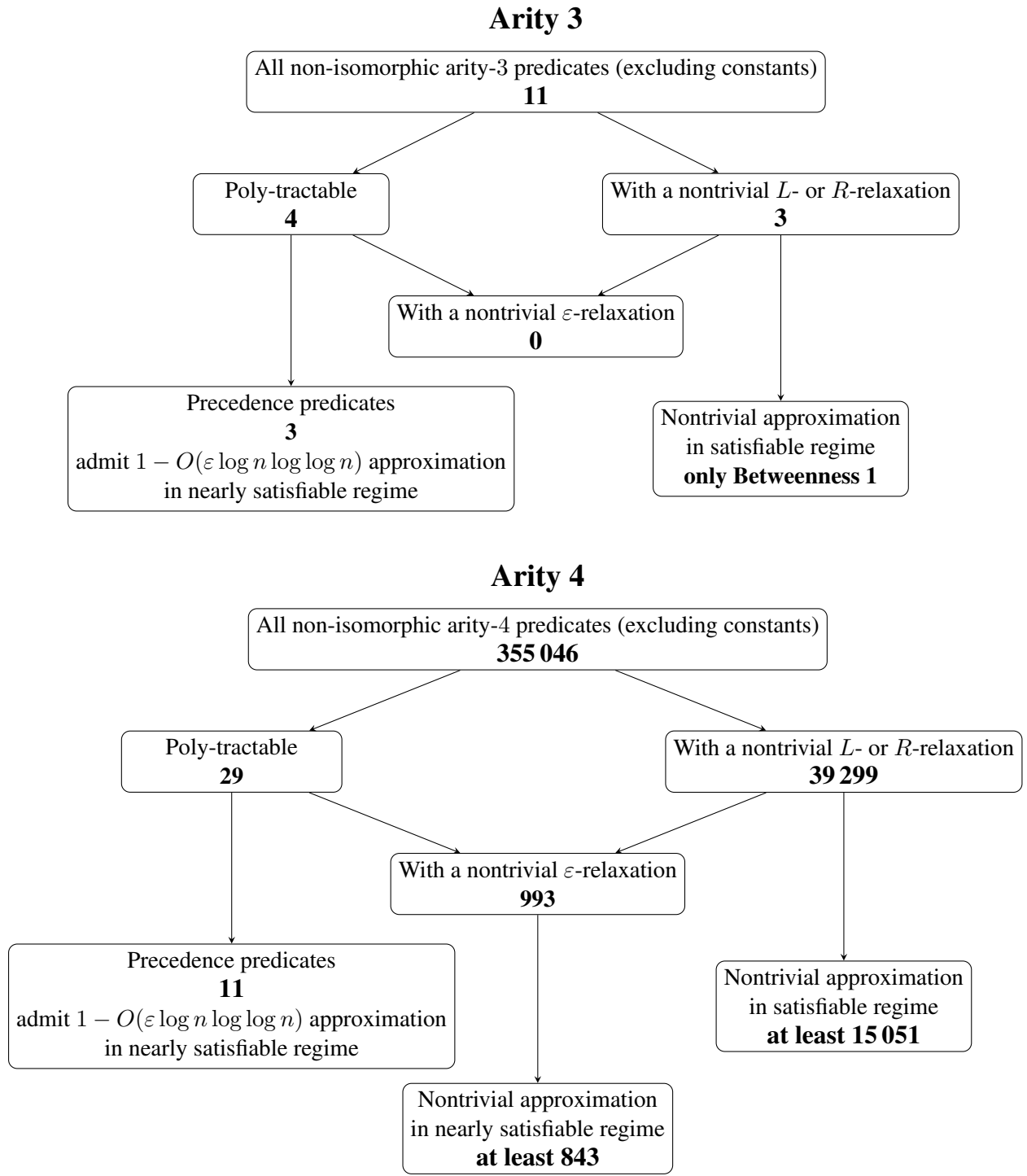
\begin{figure}[p]
\centering
\begin{tikzpicture}[
    x=1cm,y=1cm,
    box/.style={draw, rounded corners, align=center, minimum width=3.8cm, minimum height=0.9cm},
    smallbox/.style={draw, rounded corners, align=center, minimum width=3.2cm, minimum height=0.8cm},
    note/.style={align=center},
    >=stealth
]
\node[font=\bfseries] at (0,1) {\Large Arity 3};
\node[box] (tot3) at (0,0)
{All non-isomorphic arity-$3$ predicates (excluding constants)\\\textbf{\large 11}};

\node[smallbox] (pt3) at (-4,-2)
{Poly-tractable\\\textbf{\large 4}};

\node[smallbox] (sat3) at (4,-2)
{With a nontrivial $L$- or $R$-relaxation\\\textbf{\large 3}};

\node[smallbox] (near3) at (0,-4)
{With a nontrivial $\varepsilon$-relaxation\\\textbf{\large 0}};

\draw[->] (tot3) -- (pt3);
\draw[->] (tot3) -- (sat3);

\draw[<-] (near3) -- (pt3);
\draw[<-] (near3) -- (sat3);

\node[smallbox] (ineq3) at (-4,-6)
{Precedence predicates\\\textbf{3}\\
admit $1-O(\varepsilon\log n\log\log n)$ approximation\\ in nearly satisfiable regime};

\node[smallbox] (app3) at (4,-6)
{Nontrivial approximation\\in satisfiable regime\\\textbf{only Betweenness 1}};

\draw[->] (pt3) -- (ineq3);
\draw[->] (sat3) -- (app3);
\end{tikzpicture}

\vspace*{7mm}

\begin{tikzpicture}[
    x=1cm,y=1cm,
    box/.style={draw, rounded corners, align=center, minimum width=4.2cm, minimum height=0.9cm},
    smallbox/.style={draw, rounded corners, align=center, minimum width=3.6cm, minimum height=0.8cm},
    tinybox/.style={draw, rounded corners, align=center, minimum width=3.1cm, minimum height=0.8cm},
    note/.style={align=center},
    >=stealth
]

\node[box] (tot4) at (0,0)
{All non-isomorphic arity-$4$ predicates (excluding constants)\\\textbf{\large 355{\,}046}};

\node[smallbox] (pt4) at (-5,-2)
{Poly-tractable\\\textbf{29}};

\node[smallbox] (sat4) at (5,-2)
{With a nontrivial $L$- or $R$-relaxation\\\textbf{\large 39{\,}299}};

\node[smallbox] (near4) at (0,-4)
{With a nontrivial $\varepsilon$-relaxation\\\textbf{993}};

\draw[->] (tot4) -- (pt4);
\draw[->] (tot4) -- (sat4);

\draw[<-] (near4) -- (pt4);
\draw[<-] (near4) -- (sat4);

\node[tinybox] (ineq4) at (-5,-6)
{Precedence predicates\\\textbf{\large 11}\\
admit $1-O(\varepsilon\log n\log\log n)$ approximation\\ in nearly satisfiable regime};

\node[tinybox] (app4sat) at (5,-6)
{Nontrivial approximation\\in satisfiable regime\\\textbf{\large at least \magicSAT{}}};

\draw[->] (pt4) -- (ineq4);
\draw[->] (sat4) -- (app4sat);

\node[tinybox] (app4near) at (0,-8)
{Nontrivial approximation\\in nearly satisfiable regime\\\textbf{\large at least 843}};

\draw[->] (near4) -- (app4near);
\node[font=\bfseries] at (0,1) {\Large Arity 4};

\end{tikzpicture}

\caption{Summary of the classification results for single-predicate ordering CSPs of arity~3 and~4. The diagram shows the number of predicates of different types. The boxes labeled ``With a nontrivial $L$- or $R$-relaxation'' count predicates for which at least one of the canonical tractable relaxations $\varphi_L$ and $\varphi_R$ is nontrivial. The boxes labeled ``With a nontrivial $\varepsilon$-relaxation'' are defined analogously.}
\label{fig:summary-3-4-results}
\end{figure}

\subsection{Completely satisfiable instances}\label{sec:satisfiable-csps}
We briefly review the classification of polynomially tractable temporal CSPs due to Bodirsky and Kára~\cite{BK10}. The key difference between ordering and temporal CSPs is that, in temporal CSPs, distinct variables may take the same value; accordingly, each predicate is a Boolean formula involving clauses of the form $x_i < x_j$ and $x_i = x_j$. In contrast, in an ordering CSP all variables are required to take distinct values.

Bodirsky and Kára showed that a temporal CSP is polynomial-time tractable if and only if its temporal constraint language is preserved by at least one of the following operations: $ll$, $min$, $mi$, $mx$, their duals, or a constant operation~\cite[Theorem 50]{BK10}; otherwise, the CSP is NP-hard.

To apply this result to ordering CSPs, we reinterpret every ordering CSP as a temporal CSP. First, we allow different variables to take the same value. Second, we extend each predicate by declaring that $\varphi(x_1,\dots,x_k)$ is unsatisfied whenever $x_i = x_j$ for some $i \neq j$. In this way,~\cite{BK10} yields a criterion for the polynomial-time tractability of ordering CSPs, which we now make more explicit.

We define two classes of ordering CSPs, which we call Not-First Ordering CSPs and Not-Last Ordering CSPs. In the context of temporal CSPs and scheduling, these classes have been studied under the name AND/OR precedence constraints.  
\begin{definition}[Not-First and Not-Last Ordering CSPs]\label{def:not-first}
Define the predicate $\NFirst_k(x_1,\dots, x_k)$ as $x_1 > \min(x_2,\dots, x_k)$, or equivalently $(x_1 > x_2) \vee \dots \vee (x_1 > x_k)$.  
It is satisfied if and only if $x_1$ is not the first among $x_1,\dots, x_k$.  
Similarly, define the predicate $\NLast_k(x_1,\dots, x_k)$ as $x_1 < \max(x_2,\dots, x_k)$, or equivalently $(x_1 < x_2) \vee \dots \vee (x_1 < x_k)$.

We denote the set of predicates $\varphi$ that are conjunctions of $\NFirst$ predicates (possibly on different subsets of variables and of different arities) by $\Pi_{\NFirst}$:
$$\varphi = \bigwedge_{i=1}^m \NFirst_{k_i}(x_{j_{i_1}}, \dots, x_{j_{ik_i}}).$$
Similarly, we denote by $\Pi_{\NLast}$ the set of predicates that are conjunctions of $\NLast$ predicates.

We say that an ordering CSP $\CSP(\Pi)$ is a \emph{Not-First Ordering CSP} if $\Pi \subseteq \Pi_{\NFirst}$. Similarly, we say that $\CSP(\Pi)$ is a \emph{Not-Last Ordering CSP} if $\Pi \subseteq \Pi_{\NLast}$. When the arity of $\NFirst$ and $\NLast$ is clear from context, we omit the subscript $k$.
\end{definition}
Not-First and Not-Last ordering CSPs are known to be polynomially tractable.
It follows from the Bodirsky--Kára classification that $\CSP(\Pi)$ is polynomially tractable if and only if it is either a Not-First or a Not-Last ordering CSP.  
Namely, the classification of $ll$-closed temporal CSPs from~\cite{BKfast10} immediately implies that every $ll$-closed ordering CSP is a Not-First ordering CSP, and every dual-$ll$-closed ordering CSP is a Not-Last ordering CSP.

All temporal CSPs preserved by $min$, $mi$, and $mx$ are shuffle-closed. We show in Appendix~\ref{sec:shuffle-closed} that every shuffle-closed ordering CSP is a Not-First ordering CSP. Similarly, all ordering CSPs preserved by dual-$min$, dual-$mi$, and dual-$mx$ are Not-Last ordering CSPs. No ordering predicate of arity $k \geq 2$ is preserved by constant operations.

For every ordering predicate $\varphi$, we define two canonical tractable relaxations, the $L$-relaxation $\varphi_{L}$ and the $R$-relaxation $\varphi_{R}$, as follows.
$\varphi_{L}$ is a relaxation of $\varphi$ in $\Pi_{\NFirst}$ with the minimal (by inclusion) set of satisfying assignments $\Sat(\varphi_{L})$. In other words, it is the conjunction of all predicates $\NFirst_t(x_{i_1},\dots, x_{i_t})$ that are implied by $\varphi(x_1,\dots, x_k)$ (where $i_1,\dots,i_t \in [k]$).  
Similarly, $\varphi_{R}$ is a minimal relaxation of $\varphi$ in $\Pi_{\NLast}$: it is the conjunction of all predicates $\NLast_t(x_{i_1},\dots, x_{i_t})$ that are implied by $\varphi(x_1,\dots, x_k)$.

By applying Theorem~\ref{thm:idu-framework} to $\Pi' = \Pi_{L}$ and $\Pi' = \Pi_{R}$, we obtain the following result.

\begin{theorem}\label{thm:idu-framework-satisfiable}
Consider an ordering CSP $\CSP(\Pi)$. Let $\{\sigma_n\}$ be a weak IDU transformation that can be sampled in polynomial time.  
There is a polynomial-time algorithm that, given a satisfiable instance of $\CSP(\Pi)$, finds a solution satisfying at least a 
$$\max(p(\Pi_{L} \to \Pi, \bm{\sigma}),p(\Pi_{R} \to \Pi, \bm{\sigma}))$$ fraction of constraints in expectation. 
\end{theorem}

Further, by Theorems~\ref{thm:main-weak-to-strong} and~\ref{thm:compute-nearly-optimal-IDU}, for every arity $k$ and $\delta > 0$, we can find a nearly optimal strong IDU transformation $R$ such that 
$$p(\Pi' \to \Pi, R) \geq p(\Pi' \to \Pi, \bm{\sigma}) - \delta.$$
This requires time depending only on $k$ and $\delta$ (and not on $n$).  
Furthermore, in this case, we can sample from this IDU transformation in time $f(k, 1/\delta)\,n$. Additionally, the algorithm from Theorem~\ref{thm:idu-framework-satisfiable} can be derandomized using the method of conditional expectations.

Using Theorem~\ref{thm:compute-nearly-optimal-IDU} together with numerical optimization (via standard optimization libraries), we obtain the following result.
All computations are verified via exact arithmetic over the rational numbers.

\begin{theorem}
There are at least \magicSAT{} non-isomorphic ordering CSPs of arity~4, each defined by a single NP-hard predicate, that admit nontrivial approximation for completely satisfiable instances. In other words, these predicates are not \emph{approximation resistant} in the completely satisfiable regime.
\end{theorem}

\subsection{Nearly satisfiable instances}\label{sec:nearly-satisfiable-csps}
Next, we apply our approach to nearly satisfiable instances of ordering CSPs. Suppose we are given a $(1-\varepsilon)$-satisfiable instance, where $\varepsilon = O(1/\operatorname{polylog}(n))$. Previously, the only known algorithm in this regime was an approximation algorithm for the \emph{Minimum Feedback Arc Set} problem: when applied to $\CSP(<)$, it finds a solution satisfying at least a $1 - O(\varepsilon \log n \log \log n)$ fraction of constraints~\cite{ENSS98,Seymour}.
The same guarantee extends to bounded-arity precedence CSPs (see Definition~\ref{def:precedence-CSP}). No other guarantees were known for nearly satisfiable ordering CSPs.

We use our framework to obtain approximation results for nearly satisfiable instances of other ordering CSPs. Specifically, we replace each predicate $\varphi$ with the conjunction of all clauses of the form $x_i < x_j$ implied by $\varphi$. This yields the $\varepsilon$-relaxation $\varphi_{\varepsilon}$ of $\varphi$, in which every constraint is a conjunction of such clauses. We then apply the approximation algorithm for Minimum Feedback Arc Set to the relaxed instance, and use an appropriate IDU transformation to obtain a solution for the original instance.

By applying Theorem~\ref{thm:idu-framework} to this setting, we obtain the following result.
\begin{theorem}\label{thm:idu-framework-nearly-satisfiable}
Consider an ordering CSP $\CSP(\Pi)$. Let $\{\sigma_n\}$ be a weak IDU transformation that can be sampled in polynomial time.  
There is a polynomial-time algorithm that, given a $(1-\varepsilon)$-satisfiable instance of $\CSP(\Pi)$, finds a solution satisfying at least a 
$$p(\Pi_{\varepsilon} \to \Pi, \bm{\sigma}) - O(\varepsilon \log n \log \log n)$$ fraction of constraints in expectation. 
\end{theorem}

Using numerical optimization together with verification via exact arithmetic over $\mathbb Q$, we obtain the following theorem.
\begin{theorem}
There are at least \magicNEARSAT{} non-isomorphic ordering CSPs of arity~4, each defined by a single predicate, that admit nontrivial approximation in the nearly satisfiable regime; this number excludes bounded-arity precedence CSPs. Given a $(1-\varepsilon)$-satisfiable instance, the algorithm returns a solution satisfying at least an $(\alpha - O(\varepsilon \log n \log \log n))$ fraction of constraints, where $\alpha > \alpha_{\mathrm{random}}$ is the approximation factor achieved by the algorithm.
\end{theorem}

\section{Preliminaries}\label{sec:prelim}
\subsection{Permutons}\label{sec:permutons}
In this section, we review the definitions and basic properties of deterministic and random permutons.
Permutons are limit objects for permutations analogous to graphons for graphs.
Deterministic permutons were introduced by Hoppen, Kohayakawa, Moreira, R{\'a}th, and Sampaio~\cite{HKMRS}. Random permutons were later defined by Bassino, Bouvel, F{\'e}ray, Gerin, Maazoun, and Pierrot~\cite{BBFGMP}. 

\begin{definition}
A (deterministic) \emph{permuton} $\mu$ is a Borel probability measure on the unit square $[0,1]^2$ with uniform marginals.
That is, $\mu(A \times [0,1]) = \mu([0,1] \times A) = \lambda(A)$ for every Borel set $A \subseteq [0,1]$, where $\lambda$ is the Lebesgue measure on $[0,1]$.

A random permuton $R$ is a probability distribution over the space of permutons.
\end{definition}
The space of permutons is equipped with the weak topology: a sequence of permutons $\mu_i$ converges to $\mu$ if for every bounded continuous function $f:[0,1]^2 \to \mathbb{R}$, we have $\int_{[0,1]^2} f\,d\mu_i \to \int_{[0,1]^2} f\,d\mu$. This topology is induced by the box distance, defined as
\[
d_{\square}(\mu_1, \mu_2) = \sup_{[a_1,a_2],[b_1,b_2] \subseteq [0,1]} \left| \mu_1([a_1,a_2] \times [b_1, b_2]) - \mu_2([a_1,a_2] \times [b_1, b_2]) \right|.
\]
The space of random permutons is also equipped with the weak topology: a sequence of random permutons $R_n$ converges to $R$ if for every bounded continuous function $F$ on the space of permutons,
$\mathbb{E}_{\mu\sim R_n}[F(\mu)] \to \mathbb{E}_{\mu\sim R}[F(\mu)]$.

Permutons generalize permutations: every permutation $\pi \in \Sn$ has an associated permuton $\mu_{\pi}$. 
The permuton $\mu_{\pi}$ is supported on $n$ disjoint squares of size $1/n \times 1/n$, where square number $i$ is $[(i-1)/n,\, i/n] \times [(\pi(i)-1)/n,\, \pi(i)/n]$. 
The measure $\mu_{\pi}$ is uniform on its support; that is, its probability density equals $n$ on the support and $0$ elsewhere. 
Similarly, for a distribution $D_n$ of permutations of length $n$, we can define a corresponding random permuton $R_n$ in which each $\mu_{\sigma}$ occurs with the same probability as $\sigma$ in $D_n$.

\begin{figure}
\centering
\begin{tikzpicture}[x=3cm,y=3cm]
  \draw (0,0) rectangle (1,1);

  \draw[very thick] (0,0) -- (1,0.5);   
  \draw[very thick] (0,0.5) -- (1,1);   

  \filldraw[darkred] (0.22,0.11) circle (0.02);
  \filldraw[darkred] (0.74,0.37) circle (0.02);

  \filldraw[darkred] (0.40,0.7) circle (0.02);
  \filldraw[darkred] (0.48,0.74) circle (0.02);
  \filldraw[darkred] (0.86,0.93) circle (0.02);
\end{tikzpicture}
\hspace{3cm}
\begin{tikzpicture}[x=3cm,y=3cm]
  \draw (0,0) rectangle (1,1);

  \draw[step=0.2,gray!40,shift={(0.2,0)}] (0,0) -- (0,1);
  \draw[step=0.2,gray!40,shift={(0,0.2)}] (0,0) -- (1,0);
  \foreach \i in {0.2,0.4,0.6,0.8}{
    \draw[gray!40] (\i,0) -- (\i,1);
    \draw[gray!40] (0,\i) -- (1,\i);
  }

  \fill[gray!35] (0/5,2/5) rectangle (1/5,3/5);
  \fill[gray!35] (1/5,1/5) rectangle (2/5,2/5);
  \fill[gray!35] (2/5,4/5) rectangle (3/5,5/5);
  \fill[gray!35] (3/5,0/5) rectangle (4/5,1/5);
  \fill[gray!35] (4/5,3/5) rectangle (5/5,4/5);

  \foreach \xa/\ya/\xb/\yb in {0/2/1/3, 1/1/2/2, 2/4/3/5, 3/0/4/1, 4/3/5/4} {
    \draw[black!75, line width=0.6pt] (\xa/5,\ya/5) rectangle (\xb/5,\yb/5);
  }
\end{tikzpicture}
\caption{The left panel shows five points sampled from a permuton (whose measure is supported on two segments), defining the pattern $(\mathsf{1\,3\,4\,2\,5})$.  
The right panel displays the permuton \(\mu_\pi\) associated with the permutation \(\pi=(\mathsf{3\,2\,5\,1\,4})\in {\mathbb{S}}_5\), whose support consists of five disjoint squares with side length \(1/5\).}
\end{figure}
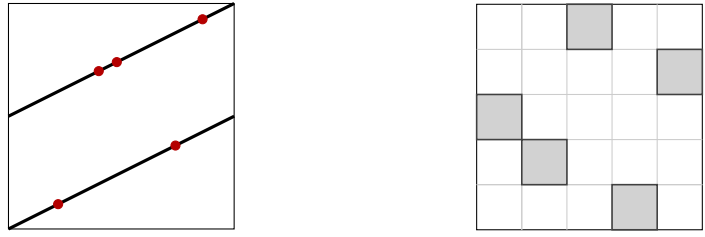

Given a deterministic permuton $\mu$ and a parameter $n$, we can sample a permutation (or pattern) $\sigma_n$ of length $n$ as follows. Choose $n$ random points $(X_1, Y_1), \dots, (X_n, Y_n)$ independently according to the measure $\mu$. Then $\sigma_n$ is the pattern defined by the points $\{(X_i, Y_i)\}_{i=1}^n$ (see Definition~\ref{def:pattern}).
Thus, every deterministic permuton $\mu$ defines a weak IDU transformation (the converse does not generally hold).

We now extend the definition of pattern densities to permutons. Consider $\rho \in \Sn$. The pattern density $d(\rho,\mu)$ equals the probability that the pattern $\sigma_n$ (sampled as above from $\mu$) equals $\rho$.
Further, given a subset $J\subset [0,1]$ of size $n$, we define the restricted pattern density $d_J(\rho, \mu)$ as follows.
\begin{itemize}
\item Let $X_1, \dots, X_n$ be the elements of $J$ in increasing order.
\item Sample each $Y_i$ independently according to $\mu$ conditioned on $X_i$.
\item Let $\sigma_n$ be the pattern of $\{(X_i, Y_i)\}_{i=1}^n$.
\end{itemize}
Then $d_J(\rho, \mu)$ is the probability that $\sigma_n = \rho$. Importantly,
\[
d(\rho,\mu)=\Exp_{J}\!\left[d_J(\rho,\mu)\right],
\]
where the expectation is taken over a random subset \(J\subseteq[0,1]\) obtained by sampling \(n\) independent points from the uniform distribution on \([0,1]\).

We state important properties of deterministic permutons established in~\cite{HKMRS} in the following theorem.
\begin{theorem}[see \cite{HKMRS}]\label{thm:permuton-properties}$\quad$
\begin{itemize}
\item The set of all permutons is a compact set with respect to the weak topology.
\item The function $\mu \mapsto d(\rho, \mu)$ is continuous with respect to the weak topology on permutons.
\end{itemize}
\end{theorem}

We now generalize the definitions above to random permutons. 
To sample a permutation $\sigma_n$ from a random permuton $R$, we first sample a deterministic permuton $P$ according to $R$, and then sample $\sigma_n$ from $P$ as described above.
We define
$d(\rho, R) = \Exp_{P \sim R}[d(\rho, P)]$ and $d_J(\rho, R) = \Exp_{P \sim R}[d_J(\rho, P)]$.
As in the deterministic case, every random permuton $R$ defines a weak IDU transformation.
Following~\cite{BBFGMP}, we define consistent sequences of random permutations, or in our terminology, consistent weak IDU transformations.
\begin{definition}\label{def:consistency}  
 We say that a weak IDU transformation $\bm{\sigma}$ is \emph{consistent} if for every $k \geq 1$ and $n\geq k$, the distribution of $\sigma_k$ coincides with that of $\patt(\sigma_n, J)$, where $J$ is a random subset of $[n]$ of size $k$ (here the randomness is over both $J$ and $\sigma_n$).
\end{definition}

\begin{theorem}[see Proposition 2.9 in~\cite{BBFGMP}]\label{thm:random-permuton-consistency}
Every random permuton defines a consistent weak IDU transformation. Conversely, every consistent weak IDU transformation is defined by a random permuton.  
\end{theorem}

For brevity, we will refer to random permutons that define strong IDU transformations as \emph{strong IDU permutons}.

\subsection{Pattern density sums}\label{sec:pattern-density-sums}
Our ultimate goal is to analyze $p(\varphi' \to \varphi, \bm{\sigma})$ for various IDU transformations $\bm{\sigma}$. By definition, this quantity is an infimum over $n$, $J$, and $\tau\in \Sat(\varphi')$ of the following sum:
$$
\sum_{\rho\in \Sat(\varphi)\tau^{-1}} d_J(\rho, \sigma_n).
$$
We will refer to sums of this form as (restricted) \emph{pattern density sums}.  
Formally, let $\rho_1, \dots, \rho_t \in \Sk$ and $c_1, \dots, c_t \in \mathbb{R}$; we consider density sums of the form  
$$
\sum_{i=1}^t c_i d_J(\rho_i, \sigma_n) \quad \text{and} \quad \sum_{i=1}^t c_i d_J(\rho_i, R),
$$
where $\bm{\sigma}$ is a weak IDU transformation and $R$ is a random or deterministic permuton.  
We will also consider pattern density sums without fixing a subset $J$
$$
\sum_{i=1}^t c_i d(\rho_i, \sigma_n) \quad \text{and} \quad \sum_{i=1}^t c_i d(\rho_i, R).
$$
Our focus will be on the set of \emph{pattern density inequalities} that a weak IDU transformation or a random permuton satisfies.  
We say that a weak IDU transformation $\bm{\sigma}$ satisfies a pattern density inequality  
$$
\sum_{i=1}^t c_i d_J(\rho_i, \sigma_n) \geq b,
$$
if, for every $n \geq k$ and every $J \subseteq [n]$ of size $k$, the inequality holds.  
Similarly, we say that a random permuton $R$ satisfies a pattern density inequality if the inequality  
$$
\sum_{i=1}^t c_i d_J(\rho_i, R) \geq b
$$
holds almost surely over the choice of a random subset $J \subseteq [0,1]$ of size $k$
(obtained by sampling \(k\) independent points from the uniform distribution on \([0,1]\)).

It will be convenient for us to work with homogeneous inequalities where $b=0$. 
Given an inequality with $b \neq 0$, we can always rewrite it as
$$\sum_{i=1}^t c_i d_J(\rho_i, \sigma_n) \geq b\sum_{\tau\in \Sk} d_J(\tau,\sigma_n),$$
since $\sum_{\tau\in \Sk} d_J(\tau,\sigma_n)$ is always $1$ for $n \geq k$. 
The same trick applies to permutons as well. More generally, for $\rho\in\Sk$ and $n\geq k$,
it always holds that 
$$d(\rho, R) = \sum_{\tau\in \Sn} d(\rho, \tau)\, d(\tau, R),$$
which is a representation of $d(\rho, R)$ as a sum of densities $d(\tau, R)$ 
with fixed nonnegative coefficients $d(\rho, \tau)$. We will use this formula in Section~\ref{sec:flag-algebras}, where we will present the flag algebra framework.

From the previous discussion, we make the following important observation.
\begin{observation}\label{sec:observation-as-good-as}
If a random permuton $R$ satisfies all pattern density inequalities that are satisfied by a weak IDU transformation $\bm{\sigma}$, then for every ordering predicate $\varphi$ and every relaxation $\varphi'$ of $\varphi$, we have
$$
p(\varphi' \to \varphi, R) \geq p(\varphi' \to \varphi, \bm{\sigma}).
$$
\end{observation}
Put simply, this means that $R$ is at least as good for our purposes as $\bm{\sigma}$.
We will also use a shorthand notation:
\begin{align*}
d_J\!\left(\sum_{i=1}^t c_i \rho_i, \sigma_n\right) &\eqdef \sum_{i=1}^t c_i d_J(\rho_i, \sigma_n),\\
d\!\left(\sum_{i=1}^t c_i \rho_i, R\right) &\eqdef \sum_{i=1}^t c_i d(\rho_i, R),
\end{align*}
which will become especially useful in Section~\ref{sec:flag-algebras}, where we interpret the maps $d(\cdot, R)$ as algebra homomorphisms.
\subsection{Notation}\label{sec:notation}
We denote the set of positive integers by $\bbN$. Let $[k] = \{1,\dots, k\}$.
We denote the set of all permutations on $[k]$ by $\Sk$.
We denote the length of a permutation $\rho$ by $|\rho|$.
We say that $x_i$ has rank $r$ among $x_1, \dots, x_n$ if it occupies position $r$ in the sorted order of $x_1, \dots, x_n$.

For a vector $v\in \mathbb{Z}^{r}$ with nonnegative entries, let
\[
M(v) = \binom{\|v\|_1}{v_1, \dots, v_{r}} = \frac{\|v\|_1!}{v_1! v_2! \cdots v_{r}!}.
\]

\section{IDU transformations}
\subsection{\texorpdfstring{$I$, $D$, and $U$}{I, D, and U} permutons and their up-combinations}

\begin{figure}
\centering
\begin{tikzpicture}[line cap=round, line join=round]
  \def\s{3cm}
  \def\gap{2.5cm}

  \begin{scope}
    \draw (0,0) rectangle (\s,\s);
    \draw[thick] (0,0) -- (\s,\s);
    \node[below left] at (0,0) {$0$};
    \node[above left] at (0,\s) {$1$};
    \node[below right] at (\s,0) {$1$};
  \end{scope}

  \begin{scope}[xshift=\s+\gap]
    \draw (0,0) rectangle (\s,\s);
    \draw[thick] (0,\s) -- (\s,0);
    \node[below left] at (0,0) {$0$};
    \node[above left] at (0,\s) {$1$};
    \node[below right] at (\s,0) {$1$};
  \end{scope}

  \begin{scope}[xshift=2*(\s+\gap)]
    \fill[gray!30] (0,0) rectangle (\s,\s);
    \draw (0,0) rectangle (\s,\s);
    \node[below left] at (0,0) {$0$};
    \node[above left] at (0,\s) {$1$};
    \node[below right] at (\s,0) {$1$};
  \end{scope}
\end{tikzpicture}
\caption{The increasing $I$, decreasing $D$, and uniform $U$ permutons.}
\label{fig:permutons}
\end{figure}
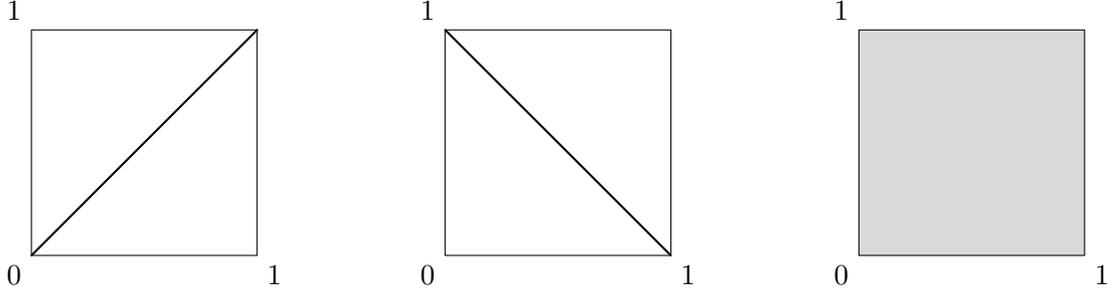

As discussed in the preliminaries, permutations of any length can be sampled from both deterministic and random permutons. Deterministic and random permutons thus define weak IDU transformations. The converse is generally not true: not every weak IDU transformation is defined by a permuton. However, every strong IDU transformation can be represented by a random permuton, as we will see in Section~\ref{sec:strongidu-consistency}.

We begin with several examples of strong IDU transformations, followed by the definition of a composition operation showing that the composition of strong IDU transformations is again strong. This yields a rich family of strong IDU transformations obtained by composing the basic $I$, $D$, and $U$ transformations.

To that end, we first introduce three fundamental permutons -- $I$ (increasing), $D$ (decreasing), and $U$ (uniform or Lebesgue) -- illustrated in Figure~\ref{fig:permutons}.
The permuton $I$ is the uniform measure on the main diagonal of $[0,1]^2$; $D$ is the uniform measure on the anti-diagonal; and $U$ is the uniform (Lebesgue) measure on the entire square $[0,1]^2$.

The increasing IDU transformation associated with $I$ corresponds to the sequence of identity permutations, $\sigma_n = \mathrm{id}_n$, while the decreasing IDU transformation corresponds to the sequence of decreasing permutations defined by $\sigma_n(i) = \mathrm{decr}_n(i) =  n + 1 - i$.
The uniform IDU transformation associated with $U$ corresponds to the sequence of permutations $\sigma_n$ sampled uniformly at random from $\Sn$.

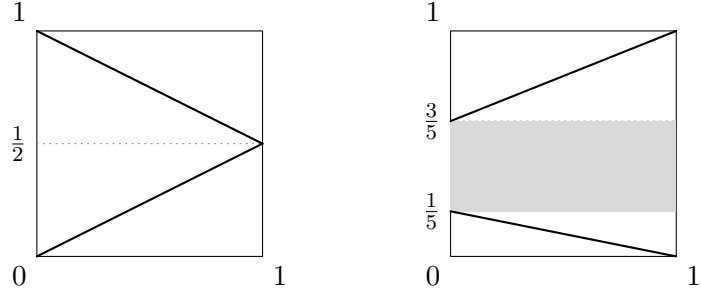
\begin{figure}
\centering
\begin{tikzpicture}[line cap=round, line join=round]
  \def\s{3cm}       
  \def\gap{2.5cm}   

  \begin{scope}
    \draw (0,0) rectangle (\s,\s);
    \draw[gray, dotted] (0,0.5*\s) -- (\s,0.5*\s);
    \draw[thick] (0,0) -- (\s,0.5*\s);
    \draw[thick] (\s,0.5*\s) -- (0,\s);
    \node[below left] at (0,0) {$0$};
    \node[above left] at (0,\s) {$1$};
    \node[below right] at (\s,0) {$1$};
    \node[left] at (0,0.5*\s) {$\tfrac{1}{2}$};
  \end{scope}

  \begin{scope}[xshift=\s+\gap]
    \draw (0,0) rectangle (\s,\s);
    \draw[gray, dotted] (0,0.2*\s) -- (\s,0.2*\s);
    \draw[gray, dotted] (0,0.6*\s) -- (\s,0.6*\s);
    \fill[gray!30] (0,0.2*\s) rectangle (\s,0.6*\s);
    \draw[thick] (\s,0) -- (0,0.2*\s);
    \draw[thick] (0,0.6*\s) -- (\s,\s);
    \node[below left] at (0,0) {$0$};
    \node[above left] at (0,\s) {$1$};
    \node[below right] at (\s,0) {$1$};
    \node[left] at (0,0.2*\s) {$\tfrac{1}{5}$};
    \node[left] at (0,0.6*\s) {$\tfrac{3}{5}$};
  \end{scope}
\end{tikzpicture}

\caption{The figure shows up-combinations of permutons: 
\texorpdfstring{$\frac{1}{2}I \oUp \frac{1}{2} D$ (left) and $\frac{1}{5} D \oUp \frac{2}{5} U \oUp \frac{2}{5} I$ (right).}{}}
\label{fig:permuton-compositions}
\end{figure}

Observe that $I$, $D$, and $U$ are all strong IDU transformations: for every $\rho \in \Sk$ and every $J$, 
$d_J(\rho, I)$ equals $1$ if $\rho = \mathrm{id}_k$ and $0$ otherwise. 
Similarly, $d_J(\rho, D)$ equals $1$ if $\rho = \mathrm{decr}_k$ and $0$ otherwise. 
Finally, $d_J(\rho, U)$ equals $1/k!$ for every $\rho \in \Sk$. 
All these formulas are independent of the choice of $J$. 

\begin{remark}
Let us also present an example of an interesting \textit{weak} IDU transformation that is not a \textit{strong} IDU transformation. Consider a Gaussian random walk $\{X_i\}$ with $X_0 = 0$, $X_{i+1}-X_i \sim \mathcal{N}(0,1)$, and independent increments. For every $n$, define a random permutation $\sigma_n$ as the pattern of the sequence $X_1, \dots, X_n$. Then $\bm{\sigma}$ satisfies nontrivial pattern density sum equalities:
$d_J((\mathsf{1\, 2}), \sigma_n) = d_J((\mathsf{2\, 1}), \sigma_n) = \nicefrac{1}{2}$ and
$d_J((\mathsf{1\, 2\, 3}), \sigma_n) = d_J((\mathsf{3\, 2\, 1}), \sigma_n)
= d_J((\mathsf{1\, 3\, 2}), \sigma_n) + d_J((\mathsf{2\, 3\, 1}), \sigma_n)
= d_J((\mathsf{2\, 1\, 3}), \sigma_n) + d_J((\mathsf{3\, 1\, 2}), \sigma_n)
= \nicefrac{1}{4}$.
Nevertheless, one can easily see that the value of
$d_J((\mathsf{1\, 3\, 2}), \sigma_n)$ depends on $J$;
e.g.,
$d_{\{1,2,n\}}((\mathsf{1\, 3\, 2}), \sigma_n) \to 0$ as $n \to \infty$, whereas
$d_{\{1,n-1,n\}}((\mathsf{1\, 3\, 2}), \sigma_n) \to \nicefrac{1}{4}$ as $n \to \infty$
(here we write permutations using the one-line convention).

As we show in this paper (see Theorem~\ref{thm:main-weak-to-strong}), there exists a strong IDU transformation satisfying the same pattern density sum equalities. In particular, the permutons $\frac12 I \oUp \frac12 D$ and $\frac12 D \oUp \frac12 I$
satisfy the same pattern density sum equalities
(see Definition~\ref{def:up-comb} and Figure~\ref{fig:permuton-compositions}).
\end{remark}

\begin{definition}[up-combination of IDU transformations]\label{def:up-comb}
Consider weak IDU transformations $A_1,\dots, A_r$ and nonnegative coefficients 
$p_1,\dots, p_r$ that sum to $1$.
We define their composition 
$$C = p_1 A_1 \oUp p_2 A_2 \oUp \dots \oUp p_r A_r$$ as follows.  
Given $n$, we independently assign each element $i\in [n]$ to one of $r$ groups $G_1,\dots, G_r$ with $\Pr(i \in G_j) = p_j$.  
For each group $G_j$, sample a random permutation $\sigma^{(j)}$ from transformation $A_j$ with length parameter $|G_j|$.  

Define an ordering on $[n]$:  
if $a \in G_j$ and $b \in G_{j'}$ with $j < j'$, then $a$ precedes $b$ in the ordering.  
If $a, b \in G_j$, their relative order is determined by $\sigma^{(j)}$.  

$C$ outputs the random permutation of length $n$ that is order-isomorphic to the resulting ordering of $[n]$.
We refer to this operation as the \emph{up-combination} of IDU transformations.
\end{definition}

\begin{definition}[up-combination of permutons]
Consider deterministic permutons $A_1,\dots, A_r$ and nonnegative coefficients 
$p_1,\dots, p_r$ that sum to $1$.
We define the up-combination 
$C = p_1 A_1 \oUp p_2 A_2 \oUp \dots \oUp p_r A_r$ as follows.  
Split the unit square $[0,1]^2$ into rectangles $R_1, \dots, R_r$ of width $1$ 
and heights $p_1, p_2, \dots, p_r$ from bottom to top.   Let $s_j = \sum_{i=1}^{j-1}p_i$.
Define the measure on each rectangle 
$R_j = [0,1] \times [s_j, s_{j+1}]$ to be $A_j$ scaled vertically by a factor of $p_j$; formally,
\[
\mu_{R_j}(S) = p_j \cdot A_j\big(\{(x, (y - s_j)/p_j) : (x, y) \in S\}\big).
\]
\end{definition}

It is easy to see that the definitions of up-combinations of IDU transformations 
and permutons are consistent with each other.

\begin{claim}\label{claim:up-combination-is-strong-IDU}
If $A_1, \dots, A_r$ are strong IDU transformations, then their up-combination  
$C = p_1 A_1 \oUp p_2 A_2 \oUp \dots \oUp p_r A_r$ is also a strong IDU transformation.
\end{claim}

\begin{proof}
Fix $k \geq 1$.  
In the up-combination $C$, each element of $[n]$ is independently assigned to group $G_j$ with probability $p_j$, and the elements within each group are ordered according to the permutation sampled from $A_j$. 

Consider a set of indices $Q = \{i_1,\dots, i_k\}\subseteq[n]$ with $i_1 < i_2 < \dots < i_k$, and suppose these elements get assigned to groups $G_{j_1}, G_{j_2}, \dots, G_{j_k}$, respectively.
We analyze the distribution of the pattern $\patt(\sigma_n, Q)$.

Because the indices $j_1,\dots, j_k$ are sampled independently from the same distribution, their joint distribution depends only on the probabilities $p_1,\dots,p_r$ and not on the specific choice of index set $Q$.  
Moreover, the relative order of elements in $G_j\cap Q$ is given by the induced pattern
$\patt(\sigma^{(j)}, G_j\cap Q)$, whose distribution depends only on $|G_j\cap Q|$ and not on the specific elements that form $G_j\cap Q$, since $A_j$ is a strong IDU transformation.

The relative order between groups is deterministic: all elements of group $j$ precede those of group $j'$ when  $j < j'$. So the overall distribution of the induced pattern $\rho \in \Sk$ is the same for every subset $Q$.
Hence $d_Q(\rho, C)$ is independent of $Q$, and $C$ is a strong IDU transformation.
\end{proof}

\begin{definition}
    An \textit{IDU combination} is an up-combination of the $I$, $D$, and $U$ permutons.    An \textit{ID combination} is an up-combination of the $I$ and $D$ permutons. 
\end{definition}
\begin{corollary}
Every ID and IDU combination is a strong IDU permuton.
\end{corollary}
\subsection{Consistency of strong IDU transformations}\label{sec:strongidu-consistency}
\begin{claim}
Every strong IDU transformation $\bm{\sigma}$ is consistent. Therefore, it is defined by a random permuton.
\end{claim}

\begin{proof}
Consistency of $\bm{\sigma}$ means that for every $a < b$, the random permutation $\sigma_a$ has the same distribution as $\patt(\sigma_b, J)$, where $\sigma_b$ is a random permutation and $J$ is a subset of $[b]$ of size $a$ sampled uniformly at random and independently of $\sigma_b$.

By the definition of a strong IDU transformation, the distribution of $\patt(\sigma_c, J)$ is the same for all $c \ge a$ and all $a$-element subsets $J \subseteq [c]$.  
In particular, for $c = b$, we have that $\patt(\sigma_b, J)$ is distributed identically to $\patt(\sigma_a, [a]) = \sigma_a$.  
Hence, $\bm{\sigma}$ is consistent.

Finally, by Theorem~\ref{thm:random-permuton-consistency} (see Section~\ref{sec:prelim}), every consistent weak IDU transformation is defined by a random permuton.  
Therefore, $\bm{\sigma}$ is defined by a random permuton.
\end{proof}

\section{From weak to strong IDU transformations}
In this section, we show that for every weak IDU transformation $\bm{\sigma}$, there exists a strong IDU transformation defined by a random permuton $S$ such that $S$ satisfies all the linear constraints on pattern densities that $\bm{\sigma}$ satisfies (see Section~\ref{sec:pattern-density-sums} for relevant definitions).
Our proof consists of two main steps.
\begin{itemize}
\item First, using that the space of permutons is compact, we observe that there exists a random permuton $W$ that satisfies all the linear constraints on pattern densities that $\bm{\sigma}$ satisfies.
\item Then we begin changing $W$: we let $W_0 = W$, and go over all permutations $\rho_1, \rho_2, \dots$ (enumerated arbitrarily). At iteration $i$, we modify $W_{i-1}$ to obtain $W_i$ 
such that for every $j \leq i$, $d_J(\rho_j, W_i)$  does not depend on the choice of $J$ with $|J| = |\rho_j|$ (almost surely).
Finally, we define $S$ as a limit point of the sequence $W_1, W_2, \dots$.
\end{itemize}

\subsection{From a weak IDU transformation \texorpdfstring{$\bm{\sigma}$}{sigma} to a random permuton \texorpdfstring{$W$}{W}}
Before we proceed with the first step, we need the following definition and lemma.

\begin{definition}\label{def:G-hatG}
For a permutation $\tau \in \Sk$, let $F_{\tau}(x_1,\dots,x_k,y_1,\dots,y_k)$ be the indicator function that the points $(x_1,y_1),\dots,(x_k,y_k)$ form the pattern~$\tau$ (if $x_i = x_j$ or $y_i=y_j$ for some $i\neq j$, then $F_{\tau}$ is $0$).  

For a deterministic permuton~$P$ and measurable subsets $A_1,\dots,A_k \subseteq [0,1]$, define  
\[
G_{P,\tau}(A_1,\dots,A_k)
   = \Exp\left[F_{\tau}(x_1,\dots,x_k,y_1,\dots,y_k)
     \prod_{i=1}^k \mathbf{1}_{A_i}(x_i)\right],
\]
where the expectation is taken over independent samples $(x_i,y_i)\sim P$; $\mathbf{1}_{A_i}$ is the indicator function of~$A_i$.

For a random permuton~$R$, define  
\[
\hat G_{R,\tau}(A_1,\dots,A_k)
   = \Exp_{P \sim R}\left[G_{P,\tau}(A_1,\dots,A_k)\right].
\]
The expectation is taken over the random choice of a deterministic permuton $P$ from distribution $R$.
\end{definition}

\begin{lemma}\label{lem:continuity-G}
For any fixed pattern $\tau$ and segments $Q_1,\dots, Q_k \subseteq [0,1]$,  
the quantity $G_{P, \tau}(Q_1,\dots, Q_k)$ is continuous as a function of a deterministic permuton~$P$.  
Also, for a random permuton $R$,  
$\hat G_{R, \tau}(Q_1,\dots, Q_k)$
is continuous as a function of~$R$.
\end{lemma}

\begin{proof}
First, assume that $P$ is deterministic.  
The function $F(x_1, \dots, x_k, y_1, \dots, y_k)$, which checks whether the points  
$(x_1, y_1), \dots, (x_k, y_k)$ form the pattern~$\tau$,  
is discontinuous only if some pair of points $(x_i, y_i)$ and $(x_j, y_j)$ with $i \neq j$  
share the same $x$- or $y$-coordinate.  
The set of such points has measure zero with respect to the product measure $P^{\otimes k}$,  
since the marginals of $P$ are uniform on~$[0,1]$ and hence each $x_i, x_j$ (as well as $y_i, y_j$)  
are independently and uniformly distributed when sampled from $P$.  
Similarly, the indicator functions $\mathbf{1}_{Q_i}(x_i)$ are discontinuous only on the boundaries of the intervals~$Q_i$,  
which have measure zero with respect to the uniform distribution on~$[0,1]$.  
Furthermore, $G_{P, \tau}(Q_1,\dots, Q_k) \in [0,1]$ is bounded.  

By the Portmanteau theorem, $G_{P, \tau}(Q_1,\dots, Q_k)$ is therefore continuous  
as a function of the deterministic permuton~$P$ in the weak topology.  

Now consider a random permuton~$R$.  
Since $G_{P, \tau}(Q_1,\dots, Q_k)$ is a bounded continuous functional of~$P$,  
its expectation $\hat G_{R, \tau}(Q_1,\dots, Q_k)$  
is continuous as a function of~$R$,  
again by the Portmanteau theorem.  
\end{proof}

\begin{claim}\label{claim:IDU-to-permuton}
Consider a weak IDU transformation $\bm{\sigma}$.
There exists a random permuton $W$ that satisfies every linear pattern-density constraint that $\bm{\sigma}$ satisfies.
\end{claim}

\begin{proof}
Consider the sequence of random permutations $\sigma_{1}, \sigma_2, \sigma_4, \dots, \sigma_{2^t}, \dots$.  
For each $t$, let $S_t$ be the random permuton associated with the random permutation $\sigma_{2^t}$ (see Section~\ref{sec:permutons} for the definition of the random permuton corresponding to a random permutation).  
By compactness of the space of random permutons, there exists a limit point~$W$ of the sequence $S_1, S_2, S_3, \dots, S_t, \dots$.

Fix $k$ and a linear constraint on the sum of densities. As discussed in Section~\ref{sec:pattern-density-sums}, we can write it in a homogeneous form as
\[
\sum_{i=1}^m c_i\, d_J(\tau_i,\sigma_n) \geq 0,
\]
where all $\tau_i \in \Sk$.
Define
\[
H_R(A_1,\dots,A_k) = \sum_{i=1}^m c_i\, \hat G_{R,\tau_i}(A_1,\dots,A_k).
\]
Now define special sets $Q_1, \dots, Q_k$.  
Fix $\delta = 1/2^r$ for some $r$, and consider disjoint intervals $Q_i = ((q_i-1)\delta, q_i\delta)$ for some choice of distinct $q_i \in [2^r]$.
By Lemma~\ref{lem:continuity-G}, for fixed segments $Q_1,\dots, Q_k$ and pattern~$\tau$,  function
$\hat G_{R, \tau}(Q_1,\dots, Q_k)$ is continuous as a function of~$R$,  
and hence $H_{R}(Q_1,\dots, Q_k)$ is also continuous as a linear combination of continuous functions.

Next, we show that for every choice of intervals $Q_1, \dots, Q_k$, we have $H_{S_t}(Q_1,\dots, Q_k) \ge 0$ whenever $t \ge r$.  
Indeed,
\[
\hat G_{S_t, \tau}(Q_1,\dots, Q_k) = \frac{1}{2^{(t-r)k}}\sum_{J}d_J(\tau, \sigma_{2^t}),
\]
where the sum is over all subsets $J=\{j_1,\dots, j_k\}$ with  
$j_i \in \{(q_i-1) 2^{t-r} + 1,\dots, q_i 2^{t-r}\}$.  
Since $\sum_{i=1}^m c_i\, d_J(\tau_i, \sigma_{2^t}) \ge 0$ for every~$J$, it follows that  
$H_{S_t}(Q_1,\dots, Q_k) \geq  0$.

Letting $t \to \infty$ and using the continuity of $H_{R}(Q_1, \dots, Q_k)$ as a function of~$R$, we conclude that 
$H_{W}(Q_1, \dots, Q_k) \ge 0$ for every choice of intervals $Q_1, \dots, Q_k$ as above. 
Now, $H_{W}$ defines an absolutely continuous signed measure on $C = [0,1]^k$, where the measure of a rectangle $A_1 \times \dots \times A_k$ is given by $H_{W}(A_1,\dots, A_k)$. Let $C'\subseteq C$ be the set of all points with distinct coordinates.  
Since the Lebesgue $\sigma$-algebra on $C'$ is generated by sets of the form $Q_1 \times \dots \times Q_k$ with intervals $Q_i$ as above, it follows that $H_{W}$ is nonnegative on $C'$. Therefore, $H_{W}$ is nonnegative on all of $C$, since $C \setminus C'$ has measure zero with respect to the Lebesgue measure, and hence also with respect to the absolutely continuous signed measure $H_{W}$.

By the Lebesgue differentiation theorem, the Radon--Nikodym derivative of $H_W$ with respect to the Lebesgue measure is nonnegative almost surely.  
The derivative at point $(x_1,\dots, x_k)$ equals the expectation of  
$$\sum_{i=1}^m c_i\, d_J({\tau_i}, W) \text{  where }J = \{x_1,\dots, x_k\}.$$  
We conclude that $\sum_{i=1}^m c_i\, d_J({\tau_i}, W) \ge 0$ almost surely,  
and therefore $W$ satisfies all linear constraints on pattern densities that $\bm{\sigma}$ satisfies.
\end{proof}
Recalling Observation~\ref{sec:observation-as-good-as}, we conclude that for every IDU transformation $\bm{\sigma}$, there exists a random permuton $W$ that is at least as good as $\bm{\sigma}$ for our intended applications.

\subsection{From \texorpdfstring{$W$}{W} to a strong IDU transformation}
Now we show that there exists a strong IDU permuton $S$ (a random permuton which defines a strong IDU transformation) that satisfies all the linear constraints on pattern densities that $W$ does. 
Let $\rho_1, \rho_2, \rho_3,\dots$ be an enumeration of all permutations (patterns).
Let $W_0 = W$. We prove that there exists a sequence of random permutons $W_0, W_1, W_2, \dots$ such that $W_i$ satisfies all the linear constraints that $W$ satisfies, and $d_J(\rho_j, W_i) =d(\rho_j, W_i)$ for every $j \leq i$ and all $J$ with $|J|=|\rho_j|$.
Assume that we have already constructed $W_{i-1}$. We now construct $W_i$.

\begin{lemma}\label{lem:fix-one-density}
Let $R$ be a random permuton, $\rho\in\Sk$, and $\varepsilon > 0$.
There exists a random permuton $R_{\varepsilon}^*$ that satisfies all the linear constraints on pattern densities that $R$ satisfies, and further,
$$d_J(\rho, R^{*}_{\varepsilon}) \in [d(\rho, R^{*}_{\varepsilon})-\varepsilon, d(\rho, R^{*}_{\varepsilon})+\varepsilon]$$
almost surely (the randomness is over the choice of a $k$-subset $J$).
\end{lemma} 
\begin{proof}
Consider an infinite sequence of disjoint intervals in $(0,1)$, ordered from left to right, defined by
\[
Q_t = \bigl(1 - 2^{1-t},\, 1 - 2^{-t}\bigr), \quad t = 1, 2, \dots.
\]
In particular, $Q_1 = (0,1/2)$, $Q_2 = (1/2, 3/4)$, $Q_3 = (3/4, 7/8)$, and so on.
The specific choice of these intervals is not essential; any family of nonempty disjoint intervals ordered from left to right would suffice.
They are chosen here purely for concreteness.

For every $k$-subset $H$ of positive integers $\bbN$, define a number $c_H$ as follows. Let $H = \{i_1, \dots, i_k\}$ with $i_1< \dots < i_k$. We sample $x_{i_1},\dots, x_{i_k}$ independently; each $x_i$ is chosen uniformly at random from $Q_i$. Then $c_H$ is the expectation of $d_{\{x_{i_1},\dots, x_{i_k}\}}(\rho, R)$ over the random choice of $x_{i_1},\dots, x_{i_k}$ (sampled as described above).
Assign $H$ the color $\lfloor c_H / \varepsilon \rfloor$. Note that we use at most $1/\varepsilon + 1$ colors; if two sets $H_1$ and $H_2$ have the same color, then $|c_{H_1} - c_{H_2}| \leq \varepsilon$.

Now we use the Ramsey theorem for infinite hypergraphs (see~\cite{Ramsey}).
\begin{theorem}[Ramsey theorem for infinite hypergraphs]
Consider the complete $k$-uniform hypergraph on the set of positive integers $\bbN$. Suppose that each hyperedge is colored with one of $M$ colors. Then there exists an infinite subset of vertices $U$ such that all hyperedges within $U$ have the same color.
\end{theorem}

By the Ramsey theorem, there exists an infinite subset of positive integers $U$ such that all $k$-subsets of $U$ have the same color. Let $u_1, u_2, \dots, u_i, \dots$ be the elements of $U$ in the increasing order. We construct a new weak IDU transformation $\bm{\sigma}$. Given $n$, we sample a permutation $\sigma_n$ as follows. 
First, we sample a deterministic permuton $P$ from $R$. Then, for every $i\in [n]$, we independently sample $x_i$ from $Q_{u_i}$ and $y_i$ according to the conditional distribution of $y$ given $x$ in $P$. We let $\sigma_n\in \Sn$ be the pattern of the sequence $(x_1,y_1), \dots, (x_n,y_n)$.

Note that $d_J(\tau, \sigma_n)$ equals $\Exp_{x_1,\dots,x_n}[d_{\{x_{j_1},\dots, x_{j_t}\}}(\tau, R)]$ for every $t$, $\tau\in\mathbb{S}_t$, and $J = \{j_1,\dots, j_t\}\subseteq [n]$.
Thus, the transformation $\bm{\sigma}$ satisfies all linear constraints on pattern densities that $R$ satisfies.
By our choice of $U$, all $k$-subsets of $U$ have the same color. Therefore,
$$d_J(\rho, \sigma_n) \in [a\varepsilon, (a+1)\varepsilon)$$
for some fixed $a$. By Claim~\ref{claim:IDU-to-permuton}, there exists a random permuton $R^*_{\varepsilon}$ that satisfies all linear constraints on pattern densities satisfied by $\bm{\sigma}$. In particular, it satisfies all linear constraints satisfied by $R$ and, additionally, the two constraints
$d_J(\rho, \sigma_n) \ge a\varepsilon$ and $d_J(\rho, \sigma_n) \le (a+1)\varepsilon$
that hold for $\bm{\sigma}$.
We obtained the desired random permuton $R^{*}_{\varepsilon}$.
\end{proof}

\begin{corollary}\label{cor:basic-step-Wi}
Let $R$ be a random permuton, $\rho\in \Sk$.
There exists a random permuton $R^*$ that satisfies all the linear constraints on pattern densities that $R$ satisfies, and further,
$$d_J(\rho, R^{*}) = d(\rho, R^{*})$$
almost surely over the choice of $J$.
\end{corollary}
\begin{proof}
Consider the sequence of random permutons $R^*_{1/n}$. Since the space of random permutons is compact, the sequence has a limit point. Let $R^*$ be a limit point.
Almost surely we have
$$|d_J(\rho, R^*) - d(\rho, R^*)| \leq \limsup_{n\to \infty} |d_J(\rho, R^*_{1/n}) - d(\rho, R^*_{1/n})|  \leq \lim_{n\to \infty} 1/n= 0
$$ 
\end{proof}

\begin{theorem}\label{thm:main-weak-to-strong}
For every weak IDU transformation $\bm{\sigma}$, there exists a strong IDU permuton $S$ that satisfies all the linear constraints on pattern densities that $\bm{\sigma}$ satisfies.
\end{theorem}
\begin{proof}
As outlined above, we start with $W_0 = W$ from Claim~\ref{claim:IDU-to-permuton} and iteratively construct $W_1, W_2, \dots$. Each time we use Corollary~\ref{cor:basic-step-Wi} to construct $W_i$ from $W_{i-1}$, ensuring that $d_J(\rho_i, W_i) = d(\rho_i, W_i)$ for every $J$. Corollary~\ref{cor:basic-step-Wi} guarantees that $W_i$ satisfies all the linear constraints on pattern densities that $W_{i-1}$ and hence $W$ satisfies.

Finally, we define $S$ as a limit point of the sequence $W_1, W_2, \dots$. By construction, $S$ satisfies all the linear constraints on pattern densities that $W$ (and hence $\bm{\sigma}$) satisfies. Further, if some $W_j$ satisfies a linear constraint, then so do all $W_i$ with $i \geq j$ and by continuity $S$. 
In particular, for every $j$, we have $d_J(\rho_j, W_j) = d(\rho_j, W_j)$. That is, $W_j$ satisfies a linear constraint $d_J(\rho_j, W_j) = c$ with $c = d(\rho_j, W_j)$.
We get that $S$ also satisfies this constraint:
$d_J(\rho_j, S) = c$ (again, $c = d(\rho_j, W_j)$). It follows that $d(\rho_j, S) = \Exp_J[d_J(\rho_j, S)]=c$ and thus
$$d_J(\rho_j, S) = d(\rho_j, S).$$
Since this holds for all permutations $\rho_j$ and all $J$, we conclude that $S$ defines a strong IDU transformation.
\end{proof}

\section{Classification of strong IDU transformations}
\subsection{Core IDU permutons}
We now provide a classification of strong IDU transformations.
We first define core deterministic IDU permutons, which generalize IDU combinations that we considered earlier. Informally, one can think of them as infinite IDU combinations.
\begin{definition}
Consider a finite or infinite family $\cal H$ of disjoint intervals $H_i$ on $[0,1]$. For each interval $H_i$ fix a label, either $I$ or $D$.
We define a core deterministic IDU permuton $B$ as follows:
sample $y$ uniformly at random from $[0,1]$. Now if $y$ is outside of all intervals $H_i$, we sample $x$ uniformly at random from $[0,1]$ (informally, apply the uniform permuton). Otherwise, $y$ lies in some $H_j$. Suppose $H_j = (a,b)$. Represent $y$ as $y = (1-\lambda) a + \lambda b$ with $\lambda \in [0,1]$.
If $H_j$ is labeled $I$, we let $x = \lambda$. If $H_j$ is labeled $D$, we let $x = 1 - \lambda$.
\end{definition}

Observe that every (finite) up-combination is a core deterministic IDU permuton. Consider an up-combination $\alpha_1 I \oUp \beta_1 D \oUp \gamma_1 U \oUp \dots \oUp \alpha_r I \oUp \beta_r D \oUp \gamma_r U$ (where some coefficients may be equal to 0, meaning that the corresponding terms are absent).
Coefficients $\alpha_i, \beta_i, \gamma_i$ define a partition of $[0,1]$ into intervals $A_i, B_i, C_i$. The corresponding family $\mathcal{H}$ in the definition of the core deterministic permuton consists of all segments $A_i$ and $B_i$ (but not $C_i$); intervals $A_i$ are labeled with $I$ and intervals $B_i$ are labeled with $D$.

As a remark, we note that it is possible that the complement of $\bigcup H_i$ -- informally the uniform part -- has positive measure yet does not contain any intervals.
That is the reason why we have not included any segments labeled $U$ in the definition of the core deterministic IDU permuton.

It is convenient to define an alternative sampling procedure for core deterministic IDU permutons.

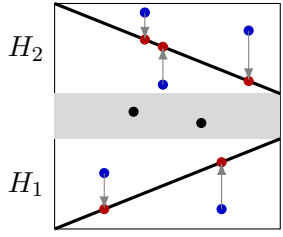
\begin{figure}
  \centering

\begin{tikzpicture}[x=3cm,y=3cm]
  \draw (0,0) rectangle (1,1);

  \fill[gray!30] (0,0.4) rectangle (1,0.6);

  \path (0,0) -- node[left]{$H_1$} (0,0.4);
  \path (0,0.6) -- node[left]{$H_2$} (0,1);

  \draw[very thick] (0,0) -- (1,0.4);

  \filldraw[darkred] (0.22,0.088) circle (0.02);
  \filldraw[darkred] (0.74,0.296) circle (0.02);

  \filldraw[darkblue] (0.22,0.248) circle (0.02);
  \draw[->, >=Latex,gray] (0.22,0.248) -- (0.22,0.088);

  \filldraw[darkblue] (0.74,0.088) circle (0.02);
  \draw[->, >=Latex,gray] (0.74,0.088) -- (0.74,0.296);

  \draw[very thick] (0,1.0) -- (1,0.6);

  \filldraw[darkred] (0.40,0.84) circle (0.02);
  \filldraw[darkred] (0.48,0.808) circle (0.02);
  \filldraw[darkred] (0.86,0.656) circle (0.02);

  \filldraw[darkblue] (0.40,0.96) circle (0.02);
  \draw[->, >=Latex,gray] (0.40,0.96) -- (0.40,0.84);

  \filldraw[darkblue] (0.48,0.64) circle (0.02);
  \draw[->, >=Latex,gray] (0.48,0.64) -- (0.48,0.808);

  \filldraw[darkblue] (0.86,0.88) circle (0.02);
  \draw[->, >=Latex,gray] (0.86,0.88) -- (0.86,0.656);

  \filldraw[black] (0.35,0.52) circle (0.02);
  \filldraw[black] (0.65,0.47) circle (0.02);
\end{tikzpicture}

\caption{The figure illustrates the procedure for obtaining $Y_i$ from $Y_i''$. Blue dots represent points $(X_i, Y_i'')$, and red dots represent points $(X_i, Y_i)$. Black dots in the gray region represent points $(X_i, Y_i) = (X_i, Y_i'')$ that lie outside all segments in $\mathcal{H}$.}
\label{fig:Yprime}
\end{figure}

\begin{claim}[Alternative Sampling Procedure]\label{claim:alt-sampling-coreIDU}
Consider a core deterministic IDU permuton $B$ and a parameter $n$. Consider the following two sampling procedures:
\begin{itemize}
\item \textbf{Standard sampling.} Sample $n$ points $(X_i, Y_i)$ independently from the measure $B$. Let $\rho$ be the pattern of $\{(X_i,Y_i)\}_{i \in [n]}$.
\item \textbf{Alternative sampling.} Sample random variables $Y_i'$ independently and uniformly from $[0,1]$. Define a permutation $\rho'\in \Sn$ as follows: $\rho'(i) < \rho'(j)$ if one of the following holds:
\begin{itemize}
\item $Y_i' < Y_j'$ and there is no interval $H \in \cal H$ containing both $Y_i'$ and $Y_j'$
\item $i < j$ and $Y_i', Y_j' \in H$ for some $H \in \cal H$ labeled $I$
\item $i > j$ and $Y_i', Y_j' \in H$ for some $H \in \cal H$ labeled $D$
\end{itemize}
\end{itemize}
Then the distributions over permutations $\rho$ and $\rho'$ produced by the two procedures are identical.
\end{claim}
\begin{proof}
We couple the two distributions of permutations.
Independently and uniformly sample $Y_i'$ and $X_i$ for $i \in [n]$. Sort all $X_i$:
$$X_{\sigma(1)} < \dots < X_{\sigma(n)}$$
(with probability $1$, there are no ties). 
Let $Y''_{\sigma(i)} = Y_i'$ for $i \in [n]$.
Since $\sigma$ is uniformly distributed in $\Sn$ and independent of the random variables $Y_i'$, the variables $Y_i''$ are also uniformly distributed on $[0,1]$, and all variables $\{X_i, Y_j''\}$ are independent.

We are now ready to define $Y_i$ so that the distribution of $\{(X_i, Y_i)\}_i$ matches that produced by the standard sampling procedure (see Figure~\ref{fig:Yprime}).
\begin{itemize}
  \item If $Y_i''$ lies outside all intervals $H \in \cal H$, let $Y_i = Y_i''$. 
  \item If $Y_i'' \in [a,b] \in \mathcal{H}$ and $[a,b]$ is labeled $I$, set $Y_i = a + (b-a) X_i$.
  \item If $Y_i'' \in [a,b] \in \mathcal{H}$ and $[a,b]$ is labeled $D$, set $Y_i = b - (b-a) X_i$.
\end{itemize}

Since the points $(X_i, Y_i'')$ are independent and each $(X_i, Y_i)$ is determined by $(X_i, Y_i'')$, the points $(X_i, Y_i)$ are also independent. Moreover, each $X_i$ is uniformly distributed on $(0,1)$, which is the $x$-marginal of $B$, as required.

We now verify that the conditional distributions of $Y_i$ given $X_i$ are the same in our procedure and in the distribution $B$.

In both distributions, the conditional probability that $Y_i \in [a,b] \in \mathcal{H}$ equals $b-a$, and when this event occurs, either $Y_i = a + (b-a) X_i$ or $Y_i = b - (b-a) X_i$, depending on the label of $[a,b]$. The probability that $Y_i \in U = [0,1] \setminus \bigcup_{H \in \mathcal{H}} H$ is $\lambda(U)$, and in both cases, $Y_i$ is uniformly distributed on $U$ conditioned on $Y_i \in U$.

It remains to verify that $\rho$ and $\rho'$ are equal. Consider $i < j$. We have $\rho(i) < \rho(j)$ if and only if $Y_{\sigma(i)} < Y_{\sigma(j)}$. We consider several cases.
\begin{itemize}
  \item $Y_i' = Y''_{\sigma(i)}$ and $Y_j' = Y''_{\sigma(j)}$ do not both lie in the same interval $[a,b] \in \mathcal{H}$. Then $Y_{\sigma(i)} < Y_{\sigma(j)}$ if and only if
  $$Y_i' = Y''_{\sigma(i)} < Y''_{\sigma(j)} = Y_j'.$$
  \item $Y_i', Y_j' \in [a,b] \in \mathcal{H}$ and $[a,b]$ is labeled $I$. Then $Y_{\sigma(i)} = a + (b-a) X_{\sigma(i)}$ and $Y_{\sigma(j)} = a + (b-a) X_{\sigma(j)}$. Since the $X_{\sigma(i)}$ are sorted and $b-a > 0$, we have $Y_{\sigma(i)} < Y_{\sigma(j)}$, and thus $\rho(i) < \rho(j)$.
  \item $Y_i', Y_j' \in [a,b] \in \mathcal{H}$ and $[a,b]$ is labeled $D$. Then $Y_{\sigma(i)} = b - (b-a) X_{\sigma(i)}$ and $Y_{\sigma(j)} = b - (b-a) X_{\sigma(j)}$. Since $-(b-a) < 0$, we have $Y_{\sigma(i)} > Y_{\sigma(j)}$, and thus $\rho(i) > \rho(j)$.
\end{itemize}

We conclude that for $i < j$, we have $\rho(i) < \rho(j)$ if and only if either (i) $Y_i', Y_j'$ do not lie in the same interval of $\mathcal{H}$ and $Y_i' < Y_j'$, or (ii) both lie in an interval labeled $I$. This is precisely the condition for $\rho'(i) < \rho'(j)$.
\end{proof}

Now we state the main theorem of this section.
\begin{theorem}\label{thm:main-strongIDU-classification}
Consider a strong IDU permuton $R$. Then $R$ is a probability distribution over core deterministic IDU permutons.
\end{theorem}
\subsection{Proof of Theorem~\ref{thm:main-strongIDU-classification}}
\begin{proof}
First, we apply the same trick as in the proof of Lemma~\ref{lem:fix-one-density}.
Define an infinite sequence of disjoint intervals in $(0,1)$ as follows.
\[
Q_t = \bigl(1 - 2^{1-t},\, 1 - 2^{-t}\bigr), \quad t = 1, 2, \dots.
\]
Consider the following sampling procedure:
\begin{itemize}
\item Sample a deterministic permuton $P$ from $R$.
\item Independently sample $(X_i, Y_i)$ from $P$ conditioned on $X_i \in Q_i$ for every positive integer $i$.
\item For every $i < j$, define $Z_{ij}$ to be $1$ if $Y_i < Y_j$ and $-1$ if $Y_i > Y_j$ (and $0$ if $Y_i = Y_j$; this happens with probability~$0$ since the $y$-marginal of $P$ is uniform on $[0,1]$).
\end{itemize}
We will now show that $Z_{ij}$ is a contractible array. Recall the definition of a contractible array from the theory of exchangeable random variables and arrays (see~\cite{Kallenberg92,Kallenberg05}).
\begin{definition}[Contractible array]
A random array $(Z_{ij})_{1\leq i<j <\infty}$ taking values in $\mathbb{R}$ is called (jointly) \emph{contractible} if, for every strictly increasing map
$c : \bbN \to \bbN$,
\[
(Z_{ij})_{1\leq i<j <\infty}
\stackrel{d}{=}
(Z_{c(i),c(j)})_{1\leq i<j <\infty}.
\]
\end{definition}

We will need Kallenberg's representation theorem for contractible arrays~\cite{Kallenberg92}.

\begin{theorem}[Kallenberg's representation for contractible arrays, restated]
Let $(Z_{ij})_{1\leq i<j <\infty}$ be a contractible array.
There exists a measurable function $f : [0,1]^4 \to\mathbb{R}$ so that the following holds.
Consider independent random variables $S$, $(T_i)_{i\in \bbN}$, and $(U_{ij})_{1\leq i < j <\infty}$
uniformly distributed on $[0,1]$. Then
\[
(Z_{ij})_{1\leq i < j <\infty}\ \stackrel{d}{=}\ \bigl(f(S, T_i, T_j, U_{ij})\bigr)_{1\leq i < j <\infty}.
\]
\end{theorem}

We now verify that the array $(Z_{ij})$ defined above is contractible. Indeed, let $c$ be an increasing function. For every $k$, the distributions of patterns $(X_1, Y_1), \dots, (X_k, Y_k)$ and $(X_{c(1)}, Y_{c(1)}), \dots, (X_{c(k)}, Y_{c(k)})$ are the same, since $R$ defines a strong IDU transformation.
Since $(Z_{ij})_{1\leq i < j\leq k}$ is determined by the pattern of points $(X_1, Y_1), \dots, (X_k, Y_k)$, we get that 
$$(Z_{ij})_{1\leq i < j\leq k}
\ \stackrel{d}{=}\ 
(Z_{c(i), c(j)})_{1\leq i < j\leq k}.$$
Since this holds for every $k$,
$$(Z_{ij})_{i<j} \ \stackrel{d}{=}\ (Z_{c(i), c(j)})_{i<j}.$$
We conclude that $(Z_{ij})$ is contractible.
Applying Kallenberg's representation theorem, we get that there exists a measurable function $f : [0,1]^4 \to \{-1,0,1\}$ such that 
\[
Z_{ij} =f(S, T_i, T_j, U_{ij})
\]
almost surely (where all random variables are defined on an appropriate refinement of the original probability space). Further, we may assume that this equality holds for all values of $S$ almost surely (over the random choice of $T_i$, $T_j$, and $U_{ij}$) by removing a measure $0$ set of $S$ values. 
We will later show that $f$ does not depend on the last argument $U_{ij}$ with probability 1 (over the choice of $S$, $T_i$ and $T_j$).

For fixed $s$, $t_1$, and $t_2\neq t_1$, write $t_1 \leftarrow_s t_2$ if $f(s, t_1, t_2, U) = 1$ with positive probability over the random choice of $U$.
Similarly, write $t_1 \rightarrow_s t_2$ if $f(s, t_1, t_2, U) = -1$ with positive probability over the random choice of $U$.
If both $t_1 \leftarrow_s t_2$ and $t_1 \rightarrow_s t_2$, write $t_1 \leftrightarrow_s t_2$. If only $t_1 \leftarrow_s t_2$, write $t_1 \Leftarrow_s t_2$; if only $t_1 \rightarrow_s t_2$, write $t_1 \Rightarrow_s t_2$. (Importantly, the events $T_1 \leftarrow_S T_2$ and $T_2 \rightarrow_S T_1$ are generally different; likewise, the events $T_1 \Leftarrow_S T_2$ and $T_2 \Rightarrow_S T_1$ need not coincide.)

We now show that $t_1 \leftarrow_s t_2$ satisfies many properties of a linear order. Importantly, however, it may happen with positive probability that $T_1 \leftarrow_S T_2$ and $T_2 \leftarrow_S T_1$.

\begin{lemma}\label{lem:equiv-relations}
The following events hold.
\begin{enumerate}
\item For every $t_1 \neq t_2$, $t_1 \leftarrow_S t_2$ or $t_1 \rightarrow_S t_2$  (always).
\item For every $t_1 \neq t_2$, at most one of the following holds: $t_1 \Leftarrow_S t_2$ or $t_1 \Rightarrow_S t_2$  (always).
\item If $T_1 \leftarrow_S T_2$ and $T_2 \leftarrow_S T_3$, then $T_1 \Leftarrow_S T_3$ (almost surely).
\item If $T_1 \leftarrow_S T_2$ and $T_1 \rightarrow_S T_3$, then $T_2 \Rightarrow_S T_3$ (almost surely).
\item If $T_1 \rightarrow_S T_2$ and $T_1 \leftarrow_S T_3$, then $T_2 \Leftarrow_S T_3$ (almost surely).
\item The probability that $T_1 \leftrightarrow_S T_2$ is 0. Equivalently, exactly one of the conditions $T_1 \Leftarrow_S T_2$ and $T_1 \Rightarrow_S T_2$ holds with probability 1. 
\end{enumerate}
\end{lemma}
\begin{proof}
Item 1 follows directly from the definitions of $\leftarrow_s$ and $\rightarrow_s$. The probabilities that $f(s, t_1, t_2, U) = 1$ and $f(s, t_1, t_2, U) = -1$ add up to 1, and thus at least one of them is positive.

Item 2 is trivial, since the requirements for $T_1 \Leftarrow_S T_2$ and $T_1 \Rightarrow_S T_2$ are mutually exclusive.

For item 3, observe that if $Z_{12} = 1$ and $Z_{23} = 1$, then $Y_1 < Y_2$ and $Y_2< Y_3$ (see the procedure for generating the array $Z$), and thus $Y_1 < Y_3$ and  $Z_{13} = 1$. 
Therefore, if $f(S, T_1, T_2, U_{12}) = 1$ and $f(S, T_2, T_3, U_{23}) = 1$, then $f(S, T_1, T_3, U_{13}) = 1$ (almost surely). 

Almost surely in $s$, $t_1$, $t_2$, and $t_3$, the following holds. If $t_1 \leftarrow_s t_2$ and $t_2 \leftarrow_s t_3$, then with positive probability over $(U_{12}, U_{23})$, $f(s, t_1, t_2, U_{12}) = 1$ and $f(s, t_2, t_3, U_{23}) = 1$, and therefore $f(s, t_1, t_3, U_{13}) = 1$. However, $f(s, t_1, t_3, U_{13})$ does not depend on $(U_{12}, U_{23})$ (since the random variables $U_{ij}$ are independent). Thus, $f(s, t_1, t_3, U_{13}) = 1$ almost surely. We conclude that $t_1 \Leftarrow_s t_3$.

The proofs of items 4 and 5 are completely analogous to that of item 3.

Now we prove item 6.
Assume to the contrary that $T_1  \leftrightarrow_S T_2$ with positive probability.
Then for some set $A$ of positive measure and every $s\in A$,
$T_1  \leftrightarrow_s T_2$ with positive probability. Fix $s\in A$. Now, for some set $B_s$ of positive measure and every $t_1\in B_s$, $t_1  \leftrightarrow_s T_2$ with positive probability. Thus (the probability below is over $T_2, T_3, T$ sampled independently and uniformly at random from $(0,1)$), 
$$\Pr(t_1 \leftrightarrow_s T_2
\text{ and } t_1 \leftrightarrow_s T_3)
=
\Pr(t_1 \leftrightarrow_s T)^2
>0.
$$
This inequality holds for all $s\in A$ and $t_1 \in B_s$.Therefore,  the following event happens with positive probability:
$$T_1 \leftrightarrow_S T_2  \text{ and } T_1 \leftrightarrow_S T_3.$$
Observe that on one hand, $T_1 \leftarrow_S T_2$ and $T_1 \rightarrow_S T_3$
imply that $T_2 \Rightarrow_S T_3$. On the other hand, $T_1 \rightarrow_S T_2$ and $T_1 \leftarrow_S T_3$, imply that $T_2 \Leftarrow_S T_3$. 
But these two options are mutually exclusive. We get a contradiction.
\end{proof}

Lemma~\ref{lem:equiv-relations} proves that $\Leftarrow_s$ and $\Rightarrow_s$ satisfy many properties of a linear order with probability 1.
To avoid unnecessary complications dealing with events of probability 0, we use Petrov's General Removal Lemma~\cite{Petrov} to 
redefine $\leftarrow_s$ and $\rightarrow_s$ so that all items of Lemma~\ref{lem:equiv-relations} always hold for distinct $t_1$, $t_2$, $t_3$ 
(alternatively, instead of directly using the General Removal Lemma, we can use the directed graphon theory).

We say that $t_1 \sim_s t_2$ if one of the following holds:
\begin{itemize}
  \item both $t_1 \Leftarrow_s t_2$ and $t_2 \Leftarrow_s t_1$
  \item both $t_1 \Rightarrow_s t_2$ and $t_2 \Rightarrow_s t_1$ 
  \item $t_1 = t_2$.
\end{itemize}

\begin{lemma} \label{lem}
  Relation $\sim_s$ satisfies the following properties.
\begin{enumerate}
\item The relation $\sim_s$ is an equivalence relation on $[0,1]$.
\item Relation $\Leftarrow_s$ satisfies the axioms of a total order on the equivalence classes of $\sim_s$, except that it is not necessarily reflexive: it is antisymmetric, transitive, and every two distinct elements are comparable.
\item Relation $\Rightarrow_s$ is the reverse order meaning that for $t_1\nsim_s t_2$ if $t_1 \Leftarrow_s t_2$, then $t_2 \Rightarrow_s t_1$ and vice versa.
\item For every equivalence class $C$ of $\sim_s$, either $t_1 \Leftarrow_s t_2$ for every $t_1, t_2\in C$ or $t_1 \Rightarrow_s t_2$ for every $t_1, t_2\in C$.
\end{enumerate}
\end{lemma}
\begin{proof}
First, it is immediate that $\sim_s$ is reflexive and symmetric. 
We now verify the following property: if $t_1 \sim_s t_2$ and $t_2 \Leftarrow_s t_3$, then $t_1 \Leftarrow_s t_3$. If $t_1 = t_2$, the claim is trivial. Otherwise, we either have $t_1 \Leftarrow_s t_2$ or $t_2 \Rightarrow_s t_1$. In both cases, by items 3 and 5 of Lemma~\ref{lem:equiv-relations}, we get that $t_1 \Leftarrow_s t_3$.
Similarly, we prove that
\begin{itemize}
\item if $t_1 \Leftarrow_s t_2$ and $t_2 \sim_s t_3$, then $t_1 \Leftarrow_s t_3$.
\item if $t_1 \sim_s t_2$ and $t_2 \Rightarrow_s t_3$, then $t_1 \Rightarrow_s t_3$.
\item if $t_1 \Rightarrow_s t_2$ and $t_2 \sim_s t_3$, then $t_1 \Rightarrow_s t_3$.
\end{itemize}

\paragraph{Now we are ready to prove item 1.} Assume that $t_1 \sim_s t_2$ and $t_2 \sim_s t_3$. If $t_2=t_3$, the claim is trivial. So we assume that $t_2\neq t_3$. Then either (a) $t_1 \sim_s t_2$, $t_2 \Leftarrow_s t_3$, and $t_3 \Leftarrow_s t_2$ or (b)  $t_1 \sim_s t_2$, $t_2 \Rightarrow_s t_3$, and $t_3 \Rightarrow_s t_2$. In case (a), we have $t_1 \Leftarrow_s t_3$ and $t_3 \Leftarrow_s t_1$, and thus $t_1 \sim_s t_3$. Similarly, in case (b), we have $t_1 \Rightarrow_s t_3$ and $t_3 \Rightarrow_s t_1$, and thus $t_1 \sim_s t_3$.

\paragraph{Now we verify item 2.} As shown above, whether $t_1 \Leftarrow_s t_2$ holds depends only on the equivalence classes of $t_1$ and $t_2$. If $t_1 \Leftarrow_s t_2$ and $t_2 \Leftarrow_s t_1$, then $t_1 \sim_s t_2$. Therefore, the relation $\Leftarrow_s$ is antisymmetric. It is transitive by Lemma~\ref{lem:equiv-relations}, item 3.

It remains to verify totality: for every $t_1 \nsim_s t_2$, either $t_1 \Leftarrow_s t_2$ or $t_2 \Leftarrow_s t_1$. Fix $t_1 \nsim_s t_2$. By Lemma~\ref{lem:equiv-relations}, item 6, either $t_1 \Leftarrow_s t_2$ or $t_1 \Rightarrow_s t_2$. In the latter case, since $t_1 \nsim_s t_2$, $t_2 \Rightarrow_s t_1$ does not hold, and thus, again by Lemma~\ref{lem:equiv-relations}, item 6, $t_2 \Leftarrow_s t_1$. We conclude that either $t_1 \Leftarrow_s t_2$ or $t_2 \Leftarrow_s t_1$, as required.

\paragraph{We prove item 3.} Assume that $t_1 \Leftarrow_s t_2$ and $t_1 \nsim_s t_2$. Then $t_2 \Leftarrow_s t_1$ does not hold (otherwise, we would have $t_1 \sim_s t_2$). By item 6 of Lemma~\ref{lem:equiv-relations}, it follows that $t_2 \Rightarrow_s t_1$, as required. The proof of the implication in the opposite direction is completely analogous.

\paragraph{Finally, we prove item 4.} Consider an equivalence class $C$ of $\sim_s$. Assume that $C$ contains at least two distinct elements $t_1$ and $t_2$.
Then, either $t_1 \Leftarrow_s t_2$ or $t_1 \Rightarrow_s t_2$.
In the former case, for every $t_3$ and $t_4$ in $C$, we have $t_3\sim_s t_1$, $t_4 \sim_s t_2$, and thus $t_3 \Leftarrow_s t_4$.
In the latter case, similarly, for every $t_3$ and $t_4$ in $C$, we have $t_3 \Rightarrow_s t_4$.
\end{proof}

Now we ``rearrange'' $[0,1]$ so that $t_1 \Leftarrow_s t_2$ agrees with the standard ordering $\leq$ on $[0,1]$ when $t_1 \nsim_s t_2$. We construct the desired rearrangement map $\eta_s$ as follows (where $\lambda$ is the standard Lebesgue measure):
\begin{align*}
M_s(t) &=\{t' : (t' \nsim_s t \text{ and } t' \Leftarrow_s t) \text{ or } (t'\sim_s t \text{ and } t' < t)\}\\
\eta_s(t) &= \lambda(M_s(t)).
\end{align*}
It is immediate that $\eta_s$ is measure preserving, since the preimage of a segment $[\eta_s(a), \eta_s(b)]$ is the set $M_s(b)\setminus M_s(a)$ of measure $\lambda(M_s(b)\setminus M_s(a)) = \eta_s(b) - \eta_s(a)$.
Let $\hat T_i = \eta_s(T_i)$. Note that random variables $\hat T_i$ are independent and uniformly distributed on $[0,1]$.

For every equivalence class $C$ of positive measure, consider its image under $\eta_s$. Clearly it is a segment. Let ${\cal H}_s$ be the set of all such segments. Further, label each segment $H\in {\cal H}_s$ with $I$ or $D$: if $t_1 \Leftarrow_s t_2$ for all $t_1, t_2\in \eta_s^{-1}(H)$, we label $H$ with $I$; otherwise, we label $H$ with $D$. 
Write $a \simeq_s b$ if $a$ and $b$ belong to the same segment in ${\cal H}_s$.

We obtain a simple characterization of when $Y_i < Y_j$ (for $i<j$): $Y_i < Y_j$ if and only if $T_i \Leftarrow_S T_j$, which holds if and only if
\begin{itemize}
\item $\hat T_i < \hat T_j$ and $T_i \not\simeq_S T_j$, or
\item $\hat T_i, \hat T_j \in H$ for some $H \in \mathcal{H}_S$ labeled $I$.
\end{itemize}

The above holds almost surely.
For every fixed $s$, consider the core deterministic IDU permuton $B_s$ defined by the family of segments ${\cal H}_s$ with the prescribed labels.
By Claim~\ref{claim:alt-sampling-coreIDU}, the distribution of patterns defined by the permuton equals that given by the sample of  $Y_i' = \hat T_i$.

We conclude that $R$ is a probability distribution over core deterministic IDU permutons $B_s$.
\end{proof}

\section{Approximating arbitrary permutons by IDU and ID combinations}
In this section, we use Theorem~\ref{thm:main-strongIDU-classification} to show that every strong IDU permuton can be approximated by a distribution of IDU or ID combinations.
\begin{theorem}\label{thm:approx-by-IDU}
For every strong IDU permuton $R$, every $\delta \in (0,1)$ and $k$, there exists a strong IDU permuton $R'$, which is a probability distribution over IDU combinations $P'$, such that for every pattern $\rho\in \mathbb{S}_{k'}$ with $k' \leq k$:
$$|d(\rho, R) - d(\rho, R')| < \delta.$$
Further, the number of terms in each IDU combination $P'$ is at most $\frac{2}{\delta}\binom{k}{2} + 1 \leq O(k^2/\delta)$; the number of different IDU combinations in the support of the distribution is at most $k!$ (we will later show that it is at most $2^{k-1}$).
\end{theorem}
\begin{proof}
  Consider a strong IDU permuton $R$ and its $k$-profile -- the vector of pattern densities $(d(\rho, R))_{\rho\in \Sk}$. It equals the expectation of profiles $(d(\rho, P))_{\rho\in \Sk}$ where $P$ is sampled from $R$. Now all profiles lie in a hyperplane in $\mathbb{R}^{\Sk}$ of codimension 1, since the sum of all pattern densities equals 1. That is, all profiles lie in an affine space of dimension $m = k! - 1$. By Carathéodory's theorem, the $k$-profile of $R$ is a convex combination of at most $m+1$ profiles of deterministic permutons $P_1, \dots, P_{m+1}$ from the support of $R$. Denote their weights in the convex combination by $c_1,\dots, c_{m+1}$. 
  As we will see later (see Corollary~\ref{cor:dim-profile-space} in Section~\ref{sec:up--down}), all $k$-profiles actually lie in an affine space of dimension at most $2^{k-1}-1$, and thus we can reduce $m$ to $2^{k-1} -1$.

  We obtained that for every pattern $\rho\in \Sk$,
  \begin{equation}    \label{eq:caratheodory-profile}
  d(\rho, R) = \sum_{t=1}^{m+1} c_t d(\rho, P_t).
  \end{equation}
  Further, since the density $d(\rho, R)$ for $\rho\in \mathbb{S}_{k'}$ with $k' < k$ is a linear combination of densities $d(\tau, R)$ for $\tau\in \Sk$, we get that
  \eqref{eq:caratheodory-profile} holds not only for $\rho\in \Sk$ but also for all $\rho\in \mathbb{S}_{k'}$ with $k' < k$.

  Now we analyze each permuton $P_t$ separately.
  By Theorem~\ref{thm:main-strongIDU-classification}, each $P_t$ is a core deterministic IDU permuton defined by a family of disjoint intervals $\mathcal{H} = {\cal H}_t$ labeled with $I$ and $D$.

  Let $\varepsilon =\delta/\binom{k}{2}$. We remove all segments $H$ in $\cal H$ whose length is at most $\varepsilon$. We keep the remaining segments as well as their labels; let ${\cal H}'$ be the collection of the remaining segments.
${\cal H}'$ defines a core permuton $P_t'$ with at most $1/\varepsilon$ segments. As discussed, every core permuton with at most $1/\varepsilon$ segments can be represented as an IDU combination with at most $2/\varepsilon + 1$ terms.  Now it remains to show that 
  $$|d(\rho, P_t) - d(\rho, P_t')| \leq \delta$$
for every pattern $\rho\in {\mathbb S}_{k'}$.

We use Claim~\ref{claim:alt-sampling-coreIDU} to sample patterns $\rho$ and $\rho'$ from permutons $P_t$ and $P_t'$, respectively. To this end, we set $X_i = i/k'$ for $i \in [k']$ and sample $Y_1, \dots, Y_{k'}$ independently and uniformly from $[0,1]$. Now, for $i<j$, we determine whether $\rho(i) < \rho(j)$ or $\rho(i) > \rho(j)$, and whether $\rho'(i) < \rho'(j)$ or $\rho'(i) > \rho'(j)$, according to the following rules.
\begin{itemize}
\item \textbf{Rule 1.} No segment $H \in \mathcal{H}$ contains both $Y_i$ and $Y_j$. If $Y_i < Y_j$, then $\rho(i) < \rho(j)$ and $\rho'(i) < \rho'(j)$. If $Y_i > Y_j$, then $\rho(i) > \rho(j)$ and $\rho'(i) > \rho'(j)$.
\item \textbf{Rule 2.} There exists a segment $H \in \mathcal{H}'$ labeled $I$ containing both $Y_i$ and $Y_j$. Then $\rho(i) < \rho(j)$ and $\rho'(i) < \rho'(j)$.
\item \textbf{Rule 3.} There exists a segment $H \in \mathcal{H}'$ labeled $D$ containing both $Y_i$ and $Y_j$. Then $\rho(i) > \rho(j)$ and $\rho'(i) > \rho'(j)$.
\item \textbf{Rule 4.} There exists a segment $H \in \mathcal{H} \setminus \mathcal{H}'$ containing both $Y_i$ and $Y_j$. The ordering of $\rho(i)$ and $\rho(j)$ is determined by Rule 2 or Rule 3, depending on the label of $H$; however, the ordering of $\rho'(i)$ and $\rho'(j)$ is given by Rule 1.
\end{itemize}
The only case in which the ordering of $\rho(i)$ and $\rho(j)$ may differ from that of $\rho'(i)$ and $\rho'(j)$ is when Rule 4 applies. Therefore, if Rule 4 is not applied for any pair $i,j$, then $\rho = \rho'$. The probability that Rule 4 is applied for some pair $i,j$ is at most (here, $\lambda$ denotes the Lebesgue measure on $[0,1]$):
\begin{align*} \sum_{i<j}\sum_{H\in {\cal H}\setminus {\cal H}'} \Pr(Y_i, Y_j \in H) &\leq \sum_{i<j}\sum_{H\in {\cal H}\setminus {\cal H}'} \lambda(H)^2 \leq \binom{k}{2} \left(\sum_{H\in {\cal H}\setminus {\cal H}'}\lambda(H)\right) \cdot \max_{H\in {\cal H}\setminus {\cal H}'} \lambda(H)\\ 
  &{} \leq \binom{k}{2} \cdot 1 \cdot \varepsilon = \delta. 
\end{align*}
Therefore, for no pattern $\rho$ of length at most $k$ do the densities in $R$ and $R'$ differ by more than $\delta$, as required.
\end{proof}
\paragraph{ID combinations}
As noted earlier, every strong IDU transformation can be alternatively approximated with a distribution of ID combinations (that is, IDU combinations without $U$ terms). 

Consider a probability distribution over IDU combinations produced by Theorem~\ref{thm:approx-by-IDU}. We approximate  each IDU combination $C$ with an ID combination as follows.
Replace each term $c\, U$ in $C$ with an up-combination $\frac{c}{t} I \oUp \dots \oUp \frac{c}{t} I$ of  $t = \lceil 1/\varepsilon \rceil$ terms. We now observe that the obtained distribution of ID combinations $C'$ approximates $C$ within $\delta$ (as in Theorem~\ref{thm:approx-by-IDU}).
To this end, note that if we apply the segment removal step from Theorem~\ref{thm:approx-by-IDU} to $C'$, we will get back $C$. Thus, the densities of each pattern $\rho\in\mathbb{S}_{k'}$  in $C$ and in $C'$ differ by at most $\varepsilon \binom{k}{2} = \delta$.

We get the following corollary.
\begin{theorem}\label{thm:approx-by-ID}
For every strong IDU permuton $R$, every $\delta > 0$ and $k$, there exists a probability distribution $R'$ over ID combinations $P'$ such that for every pattern $\rho\in \mathbb{S}_{k'}$ with $k' \leq k$:
$$|d(\rho, R) - d(\rho, R')| \leq \delta.$$
Further, the number of terms in each ID combination $P'$ is at most $O(k^2 /\delta)$; the number of different ID combinations in the support of the distribution is at most $2^{k-1}$.
\end{theorem}

\section{A quasisymmetric polynomial characterization of profiles}\label{sec:qsym}
In this section, we will provide a characterization of $k$-profiles of strong IDU transformations using quasisymmetric polynomials. 

\begin{definition}[Quasisymmetric polynomial of level-2; adapted from~\cite{Poirier98}]
Given $N$, consider two sets of variables $x_1, \dots, x_N$ and $y_1, \dots, y_N$. 
We work in the ring of polynomials $\mathbb{R}[x_1, \dots, x_N, y_1, \dots, y_N]$.
A polynomial $f \in \mathbb{R}[x_1,\dots, x_N, y_1,\dots, y_N]$ is called a \emph{quasisymmetric polynomial of level-2} if, for every two sequences of indices $i_1 < i_2 < \dots < i_k$ and $j_1 < j_2 < \dots < j_k$ and every sequence of nonnegative exponents $a_1, b_1, a_2, b_2, \dots, a_k, b_k$ with $a_i + b_i \geq 1$ for all $i$, the coefficients of monomials
$$\prod_{t=1}^k x_{i_t}^{a_t} y_{i_t}^{b_t} \quad \text{ and } \quad \prod_{t=1}^k x_{j_t}^{a_t} y_{j_t}^{b_t}$$
are equal. We denote the set of all quasisymmetric polynomials of level-2 in $2N$ variables by $\operatorname{QSym}^{(2)}_N$.  

We define a linear order on the variables: $x_1, y_1, x_2, y_2, \dots, x_N, y_N$. 
We write $u \prec v$ if variable $u \in \{x_i, y_i\}_i$ appears before variable $v \in \{x_i, y_i\}_i$ in this order.

For brevity, we will denote the vector of variables $x_i$ by $x$; the vector of variables $y_i$ by $y$.
\end{definition}
For a permutation $\pi$, we define the following $\pi$-partition function, which is a variant of Stembridge's enriched $P$-partition~\cite{Stembridge}.
\begin{definition}[$\pi$-partition function]\label{def:pi-partition}
Consider a permutation $\pi \in \Sn$. We say that $f:[n] \to \{x_i, y_i\}_{1\leq i\leq N}$ is a $\pi$-partition function if for every $i < j$:
\begin{itemize}
\item either $f(i) \prec f(j)$, or
\item $\pi(i) < \pi(j)$ and $f(i) = f(j) = x_t$ for some $t$, or
\item $\pi(i) > \pi(j)$ and $f(i) = f(j) = y_t$ for some $t$.
\end{itemize}
\end{definition}
Note that $f$ is a non-decreasing function with respect to the order $\prec$: $f(i) \preceq f(j)$ for every $i < j$, and every strictly increasing function $f$ is a $\pi$-partition function for all $\pi$. The permutation $\pi$ determines at which positions $f$ is allowed to stay constant.

\begin{definition}
For a permutation $\pi \in \Sn$, define a $\pi$-partition generating function $K_{\pi,N}(x,y)$ as follows
$$K_{\pi,N}(x,y) = \sum_{f} \prod_{i=1}^n f(i),$$
where the sum is over all $\pi$-partition functions $f$.
\end{definition}
We note that this definition is different from Stembridge's definition of a generating function.
It is immediate that $K_{\pi,N}(x,y)$ is a quasisymmetric polynomial of level-2.

Now we connect the $\pi$-partition generating functions to pattern densities in ID combinations.
Consider an ID combination $C$ with at most $N$ terms. We can write $C$ as 
\begin{equation}\label{eq:xy-ID-combination}
x_1 I \oUp y_1 D \oUp x_2 I \oUp y_2 D \oUp \dots \oUp x_N I \oUp y_N D,
\end{equation}
where $x_i$ and $y_i$ are nonnegative real coefficients that sum to 1. Some of the coefficients $x_i$ and $y_i$ may be 0, indicating that the corresponding terms are missing.

\begin{theorem}\label{thm:main-qsym-characterization}
Consider an ID combination $C$ as in~\eqref{eq:xy-ID-combination}. The density of a pattern $\pi \in \Sn$ in $C$ equals
\[
d(\pi, C) = K_{\pi^{-1},N}(x,y).
\]
Importantly, the density of $\pi$ equals the generating function for $\pi^{-1}$, not for $\pi$ itself.
\end{theorem}

\begin{proof}
We use the alternative sampling procedure from Claim~\ref{claim:alt-sampling-coreIDU}. We sample $n$ random variables $Y_i$ -- corresponding to $Y_i'$ in Claim~\ref{claim:alt-sampling-coreIDU} -- independently and uniformly at random, and then define the permutation $\rho$ according to this procedure. Our goal is to derive a formula for $\Pr(\rho = \pi) = d(\pi, C)$ for a fixed $\pi\in \Sn$.

Define a function $f$ as follows. For $i\in [n]$, let $i'= \pi^{-1}(i)$. Set \( f(i) = x_t \) if \( Y_{i'} \) lies in the segment corresponding to \( x_t I \), and \( f(i) = y_t \) if \( Y_{i'} \) lies in the segment corresponding to \( y_t D \). We will show that \( \rho = \pi\) if and only if \( f \) is a \( \pi^{-1} \)-partition function. (Note that $f$ is a random function that depends on the random variables $Y_i$.)

Assume first that $\pi = \rho$. We show that $f$ is a $\pi^{-1}$-partition function.
Consider $i < j$ and let $i' = \pi^{-1}(i)$ and $j' = \pi^{-1}(j)$. Since $\rho = \pi$, $\rho(i') = i < j = \rho(j')$. By Claim~\ref{claim:alt-sampling-coreIDU}, this happens if and only if one of the following holds:
\begin{itemize}
    \item $Y_{i'},Y_{j'}$ lie in segments corresponding to different terms in the ID combination $C$ and $Y_{i'} < Y_{j'}$. This condition is equivalent to $f(i) \neq f(j)$  and $Y_{i'} < Y_{j'}$, or, using the definition of $\prec$:
    $$f(i) \prec f(j).$$
    \item Both $Y_{i'},Y_{j'}$ lie in a segment corresponding to a term $x_t I$ in the ID combination $C$ and $i' < j'$. This is equivalent to 
    $$f(i) = f(j) = x_t  \text{ for some } t \text{ and } \pi^{-1}(i) < \pi^{-1}(j).$$
    \item Both $Y_{i'},Y_{j'}$ lie in a segment corresponding to a term $y_t D$ in the ID combination $C$ and $i' > j'$. This is equivalent to 
    $$f(i) = f(j) = y_t \text{ for some } t \text{ and } \pi^{-1}(i) > \pi^{-1}(j).$$
\end{itemize}
We see that $f$ satisfies the definition of a $\pi^{-1}$-partition function.

Conversely, assume that $f$ is a $\pi^{-1}$-partition function. We show that $\pi = \rho$.
Consider $i < j$. Again, let $i' = \pi^{-1}(i)$ and $j' = \pi^{-1}(j)$.
We consider the three possibilities from Definition~\ref{def:pi-partition}.

The first possibility is that $f(i) \prec f(j)$. Then $Y_{i'}$ lies in a segment preceding that of $Y_{j'}$, and thus $\rho(i') < \rho(j')$. Since $\pi(i') = i < j = \pi(j')$, we conclude that $\pi$ and $\rho$ agree on the ordering of $i'$ and $j'$.

Another possibility is that $i' = \pi^{-1}(i) < \pi^{-1}(j) = j'$ and $f(i) = f(j) = x_t$ for some $t$. Since $i' < j'$ and both $Y_{i'}$ and $Y_{j'}$ lie in the segment for $x_t I$, $\rho(i') < \rho(j')$. Also, $\pi(i') = i < j = \pi(j')$. So again $\pi$ and $\rho$ agree on the ordering of $i'$ and $j'$.

The remaining possibility is that $i' = \pi^{-1}(i) > \pi^{-1}(j) = j'$ and $f(i) = f(j) = y_t$ for some $t$. Since $Y_{i'}$ and $Y_{j'}$ lie in a segment for $y_t$, the order of $\rho(i')$ and $\rho(j')$ is the opposite of that of $i'$ and $j'$. Thus, $\rho(i') < \rho(j')$. Again, $\pi$ and $\rho$ agree on the ordering of $i'$ and $j'$.

We conclude that $\pi$ and $\rho$ agree on all pairs $i'$ and $j'$. Since every unordered pair $\{i', j'\}$ arises from some choice of $1 \le i < j \le n$, it follows that $\pi = \rho$, as required.

Finally, we observe that the probability that $f$ equals a fixed $\pi^{-1}$-partition function $f_0$ is 
\begin{align}\label{eq:prob-f-f0}
\Prob{f=f_0} &= \Prob{f(\pi(i')) = f_0(\pi(i')) \text{ for }i'\in [n]} \\
&= \Prob{Y_{i'} \text{ lies in  a segment of length } f_0(\pi(i')) \text{ for all }i'\in[n]}\notag\\
&=\prod_{i'=1}^n \Prob{Y_{i'} \text{ lies in  a segment of length } f_0(\pi(i'))}\notag\\
&= \prod_{i'=1}^n f_0(\pi(i')) = \prod_{i=1}^n f_0(i).\notag
\end{align}
Therefore,
$$d(\pi, C) = \Prob{\rho = \pi} = \sum_{f_0} \Prob{f = f_0} = \sum_{f_0} \prod_{i=1}^n f_0(i) = K_{\pi^{-1}, N}(x,y),$$
where both summations are over $\pi^{-1}$-partition functions $f_0$.
\end{proof}

\begin{definition}
Let $\mathcal{R}_k$ be the set of all $k$-profiles of strong IDU transformations. 
We will refer to $\mathcal{R}_k$ as the \emph{feasible region} of $k$-profiles.
\end{definition}

\begin{theorem}\label{thm:feasible-region-characterization}
Consider the vector $\mathbf{K}_N(x,y) = (K_{\pi^{-1},N}(x,y))_{\pi\in \Sk}$ with one component for each permutation~$\pi$.  
Let $\image(\mathbf{K}_N)$ denote the image of the simplex 
\[
\{(x,y) : x_i, y_i \ge 0, \; \sum_{i=1}^N (x_i + y_i) = 1\}
\]
under~$\mathbf{K}_N$.  
Then the feasible region ${\mathcal{R}}_k$ satisfies
\[
{\mathcal{R}}_k = \overline{\conv}\left(\bigcup_{N=k}^\infty \image(\mathbf{K}_N)\right),
\]
where $\overline{\conv}(X)$ denotes the closed convex hull of~$X$.
\end{theorem}

\begin{proof}
Given any strong IDU permuton $R$ and any precision $\delta>0$, by Theorem~\ref{thm:approx-by-ID}, we can find a random permuton $R'$ supported on at most $2^{k-1}$ different ID combinations (Corollary~\ref{cor:dim-profile-space} ensures that $2^{k-1}$ suffices, though even the bound $k!$ would be enough here) such that each combination has at most $N$ terms (with $N$ depending on $k$ and $\delta$), and the $k$-profile of $R'$ is within $\delta$ (in each coordinate) of the $k$-profile of $R$. In turn, Theorem~\ref{thm:main-qsym-characterization} shows that every ID combination $C$ with at most $N$ terms lies in the image of $\mathbf{K}_N$. Hence, the $k$-profile of $R'$ lies in $\conv(\image(\mathbf{K}_N))$. Letting $\delta \to 0$, we conclude that the $k$-profile of $R$ lies in $\overline{\conv}\left(\bigcup_{N\geq k} \image(\mathbf{K}_N)\right)$, as required.

Conversely, every point in $\image(\mathbf{K}_N)$ corresponds to the $k$-profile of some ID combination, and thus belongs to ${\mathcal R}_k$. Moreover, ${\mathcal R}_k$ is closed because the set of strong IDU permutons is compact. It is also convex, since for any two strong IDU permutons $R_1$ and $R_2$, the permuton $R$ obtained by sampling from $R_1$ with probability $p$ and from $R_2$ with probability $1-p$ also defines a strong IDU transformation.
\end{proof}

Now we prove a lemma that will be useful when we discuss \emph{labeled flags} in Section~\ref{sec:flag-algebras}.
Fix $n$ and $k \le n$. Consider an ID combination $C$ given by~\eqref{eq:xy-ID-combination}:
\begin{equation}
C = x_1 I \oUp y_1 D \oUp x_2 I \oUp y_2 D \oUp \dots \oUp x_N I \oUp y_N D.
\end{equation}
Suppose we are given a set $S$ of $k$ points in the support of $C$; no two points share the same $x$- or $y$-coordinate.
Let $(X_1, Y_1), \dots, (X_k, Y_k)$ be the points from $S$ sorted by their $x$-coordinates. Independently sample the remaining $n-k$ points $\{(X_j,Y_j)\}_{j=k+1}^n$ from $C$. Define the outcome of the sampling to be the pair $(\sigma, A)$, where:
\begin{itemize}
  \item $\sigma$ is the pattern of the points $\{(X_i, Y_i)\}_{i=1}^n$;
  \item $A \subseteq [n]$ is the set of ranks of $X_1,\dots, X_k$ among all $X_1, \dots, X_n$.
\end{itemize}
List the elements of $A$ in increasing order as $A=(j_1, \dots, j_k)$ with $j_1 < \dots < j_k$. Define $\Delta j_i = j_i - j_{i-1} - 1$ for $i \in [k+1]$, where $j_0 = 0$ and $j_{k+1} = n + 1$, and set $\Delta j = (\Delta j_1, \dots, \Delta j_{k+1}) \in \mathbb{R}^{k+1}$.
The coordinates $X_1 < X_2 < \dots < X_k$ divide the interval $[0,1]$ into $k+1$ intervals; let $z_1, \dots, z_{k+1}$ be their lengths (from left to right). 

Finally, we define a sequence $g_1, \dots, g_k$. The ID combination $C$ defines a partition of $[0,1]^2$ into rectangles corresponding to the individual terms in $C$: $[0,1]\times [0,x_1]$, $[0,1]\times [x_1, x_1+y_1]$, $[0,1]\times [x_1+y_1, x_1+y_1 + x_2]$, etc.
If $(X_i, Y_i)$ lies in the rectangle for the term $x_t I$, we let $g_i = x_t$; if $(X_i, Y_i)$ lies in the rectangle for the term $y_t D$, we let $g_i = y_t$.

\begin{lemma}\label{lem:labeled-qsym}
Consider the sampling procedure as described above. Fix $\pi \in \Sn$ and $A = \{j_1, \dots, j_k\}$. 
The probability of sampling $(\pi, A)$ equals
\[
M(\Delta j)\prod_{i=1}^{k+1} z_i^{\Delta j_i} \cdot \sum_f \prod_{\pi^{-1}(i) \notin A} f(i),
\]
where the sum is over all $\pi^{-1}$-partition functions $f$ such that $f(\pi(j_i)) = g_i$ for every $i \in [k]$. See Section~\ref{sec:notation} for the definition of $M(\Delta j)$.
\end{lemma}

\begin{proof}
We restate the desired probability as a conditional probability. We assume that we independently sample all $\{(X_i, Y_i)\}_{i=1}^n$ from $C$. Consider three events:
\begin{itemize}
  \item $\mathcal{E}_{\pi}$ is the event that points $(X_1, Y_1), \dots, (X_n, Y_n)$ have pattern $\pi$;
  \item $\mathcal{E}_S$ is the event that $\{(X_1, Y_1), \dots, (X_k, Y_k)\} = S$ and $X_1 < \dots < X_k$;
  \item $\mathcal{E}_A$ is the event that the rank of each $X_i$ among $X_1,\dots, X_n$ is $j_i$ for $i\in[k]$.
\end{itemize}
Then the desired probability of sampling $(\pi, A)$ equals 
\[
\Prob{\mathcal{E}_{\pi} \text{ and } \mathcal{E}_A\,|\, \mathcal{E}_S }  = \Prob{\mathcal{E}_{\pi} \,|\, \mathcal{E}_{S} \text{ and }\mathcal{E}_{A}} \cdot \Prob{\mathcal{E}_{A} \,|\, \mathcal{E}_{S}}
\]
We first obtain a formula for $\Prob{\mathcal{E}_{A} \,|\, \mathcal{E}_{S}}$. Conditioned on $\mathcal{E}_S$, 
$X_{1}, \dots, X_{k}$ partition $[0,1]$ into $k+1$ segments of lengths $z_1, \dots, z_{k+1}$. All variables
$X_{k+1}, \dots, X_n$ are independent and uniformly distributed on $[0,1]$. Therefore, the event $\mathcal{E}_A$ given $\mathcal{E}_S$ occurs when among $n-k$ random variables $X_{k+1}, \dots, X_n$ exactly $\Delta j_i$ lie in segment $i$ (of length $z_i$) for every $i\in[k+1]$. The probability of this event is
\begin{equation}\label{eq:conditional-part-1}
\binom{n-k}{\Delta j_1,\dots,\Delta j_{k+1}}
\prod_{j=1}^{k+1} z_j^{\Delta j_j}
=
M(\Delta j)\prod_{j=1}^{k+1} z_j^{\Delta j_j}.
\end{equation}

To obtain a formula for $\Prob{\mathcal{E}_{\pi} \,|\, \mathcal{E}_{S} \text{ and }\mathcal{E}_{A}}$, we slightly revise the proof of Theorem~\ref{thm:main-qsym-characterization}. Specifically, we revise formula~\eqref{eq:prob-f-f0} for the probability of $f = f_0$. Given the conditioning, the values of $Y_{i'}$ for $i'\in A$ are already fixed. If $f_0(\pi(j_i)) \neq g_i$ for some $i\in [k]$, then the conditional probability that $Y_{i'}$ lies in the required segment is 0; if $f_0(\pi(j_i)) = g_i$, then this conditional probability is $1$. We obtain the following formula
\begin{equation}\label{eq:conditional-part-2}
  \Prob{f = f_0 \,|\, \mathcal{E}_{S} \text{ and }\mathcal{E}_{A}} = 
  \begin{cases}
    \prod_{\pi^{-1}(i)\notin A} f_0(i), &\text{if } f_0(\pi(j_i)) = g_i \text{ for }i\in [k]\\
    0,  &\text{otherwise.}
  \end{cases}
\end{equation}
Combining~\eqref{eq:conditional-part-1} and \eqref{eq:conditional-part-2}, we obtain the desired formula.
\end{proof}

In this section, we constructed level-2 quasisymmetric polynomials that capture finite ID combinations.
Analogously, we can define level-3 quasisymmetric polynomials to represent IDU combinations. Although we will primarily use ID combinations to prove theoretical results, IDU combinations are often more convenient when working with specific ordering CSPs. We briefly provide the necessary definitions without formal proofs, as the arguments are completely analogous to those for ID combinations.

We consider an up-combination $C$ of the form:
\begin{equation}\label{eq:xyz-IDU-combination}
C = x_1 I \oUp y_1 D \oUp z_1 U \oUp x_2 I \oUp y_2 D \oUp z_2 U \oUp \dots \oUp x_N I \oUp y_N D \oUp z_N U.
\end{equation}

A function $f:[n] \to \{x_i, y_i, z_i\}_{1 \leq i \leq N}$ is a $\pi$-partition function for the IDU combination \eqref{eq:xyz-IDU-combination} if for every $i < j$ and some $t$, either $f(i) = f(j) = z_t$, or it satisfies the same conditions as a $\pi$-partition function for an ID combination:
  $f(i) \prec f(j)$, 
  or $\pi(i) < \pi(j)$ and $f(i) = f(j) = x_t$, 
  or $\pi(i) > \pi(j)$ and $f(i) = f(j) = y_t$.
Next, we define the generating function for \eqref{eq:xyz-IDU-combination}:
\[
K_{\pi,N}(x,y,z) = \sum_{f} \frac{\prod_{i=1}^n f(i)}{\prod_{j=1}^N |f^{-1}(z_j)|!},
\]
where the sum is over all $\pi$-partition functions $f$ for the IDU combination~\eqref{eq:xyz-IDU-combination}.
Now,
\begin{equation}\label{eq:density-for-IDU}
d(\pi, C) = K_{\pi^{-1},N}(x,y,z).
\end{equation}

\section{Equivalence in the density of permutations with matching inverse up--down signatures}\label{sec:up--down}
Recall the definition of the up--down signature of a permutation given in Definition~\ref{def:ud-signature}.
In this section, we first show that if two permutations $\pi_1$ and $\pi_2$ satisfy $\pi_1^{-1} \udsim \pi_2^{-1}$, then they have the same pattern density under every strong IDU transformation.
We then show that this property characterizes strong IDU transformations: if, for a weak IDU transformation $\bm{\sigma}$, it holds that $d(\pi_1, \sigma_n) = d(\pi_2, \sigma_n)$ for all $\pi_1$ and $\pi_2$ with $\pi_1^{-1} \udsim \pi_2^{-1}$ and all $n \ge |\pi_1| = |\pi_2|$, then $\bm{\sigma}$ is a strong IDU transformation.

\begin{claim}\label{claim:ud-sign-for-combinations}
Consider two permutations $\rho_1$ and $\rho_2$ with $\rho_1 \udsim \rho_2$ (see Definition~\ref{def:ud-signature}). Then:
\begin{itemize}
    \item $f$ is a $\rho_1$-partition function if and only if it is a $\rho_2$-partition function;
    \item $K_{\rho_1,N}(x,y) = K_{\rho_2,N}(x,y)$.
\end{itemize}
\end{claim}

\begin{proof}
Assume, for contradiction, that $f$ is a $\rho_1$-partition function but not a $\rho_2$-partition function. Then there exist $i < j$ such that $f$ satisfies one of the three conditions for being a $\rho_1$-partition function but does not satisfy any of the corresponding conditions for $\rho_2$. We consider the $\rho_1$-partition condition that $f$ satisfies. There are three cases.

\noindent \textbf{Case 1.} The condition is $f(i) \prec f(j)$. This condition does not depend on the permutation; therefore, if $f$ satisfies it for $\rho_1$, it also satisfies it for $\rho_2$.

\noindent \textbf{Case 2.} The condition is $\rho(i) < \rho(j)$ and $f(i) = f(j) = x_t$ for some $t$.  
Assume that it holds for $\rho = \rho_1$ but not for $\rho = \rho_2$. Since $f(i) = f(j)$ and $f$ is non-decreasing, we have $f(\ell) = x_t$ for all $\ell \in \{i, \dots, j\}$. Hence, $f(\ell) = f(\ell+1) = x_t$ for all $\ell \in \{i, \dots, j-1\}$. 
Since $f$ is a $\rho_1$-partition function and $f(\ell) = f(\ell+1) = x_t$, we have $\rho_1(\ell) < \rho_1(\ell+1)$ for all $\ell \in \{i, \dots, j-1\}$. Thus, the up--down signature $\udsign(\rho_1)$ has a $u$ at every position $\ell \in \{i, \dots, j-1\}$. 
Since $\rho_1$ and $\rho_2$ have the same up--down signatures, $\udsign(\rho_2)$ also has $u$ at all such positions. We therefore conclude that $\rho_2(i) < \rho_2(j)$, and hence, contrary to our assumption, $f$ also satisfies the $\rho_2$-partition condition for the indices $i$ and $j$.

\noindent \textbf{Case 3.} The condition is $\rho(i) > \rho(j)$ and $f(i) = f(j) = y_t$ for some $t$. The analysis is analogous to that of Case~2.

Since the set $\mathcal F$ of $\rho$-partition functions depends only on the up--down signature of $\rho$, we obtain
\[
K_{\rho_1,N}(x,y) = \sum_{f\in \mathcal{F}} \prod_{i=1}^n f(i) =  K_{\rho_2,N}(x,y).
\]
\end{proof}
Letting $\rho_1 = \pi_1^{-1}$, $\rho_2 = \pi_2^{-1}$ and combining this claim with Theorem~\ref{thm:main-qsym-characterization}, we get the following corollary.
\begin{corollary}\label{cor:profile-us-signature}
For every two permutations $\pi_1$ and $\pi_2$ with $\pi_1^{-1} \udsim \pi_2^{-1}$,
and every ID combination $C$,
$$d(\pi_1, C) = d(\pi_2, C).$$
\end{corollary}

Now we introduce the notion of \emph{inverse-signature invariance}: an IDU transformation is inverse-signature-invariant if the densities of patterns $\pi_1$ and $\pi_2$ are the same for all $\pi_1$ and $\pi_2$ with $\pi_1^{-1} \udsim \pi_2^{-1}$.

\begin{definition}[Inverse-signature invariance]\label{def:inv-sig-invariant}
A random permuton $R$ is \textit{inverse-signature-invariant} if 
$$d(\pi, R) = \Delta_R(\udsign(\pi^{-1}))$$
for some function $\Delta_R$ on up--down signatures and all permutations $\pi$. 

A weak IDU transformation $\{\sigma_n\}$ is \textit{inverse-signature-invariant} if 
$$d(\pi, \sigma_n) = \Delta_{\bm{\sigma}}(\udsign(\pi^{-1}))$$
for some function $\Delta_{\bm{\sigma}}$ on up--down signatures, all permutations $\pi$, and $n \geq |\pi|$.

We will refer to $\Delta_R$ and $\Delta_{\bm{\sigma}}$ as the \emph{inverse-signature profiles} of $R$ and $\bm{\sigma} = \{\sigma_n\}$, respectively.
\end{definition}
Now we extend Claim~\ref{claim:ud-sign-for-combinations} to all strong IDU transformations.

\begin{theorem}\label{thm:strong-IDU-inv-sig-invariant}
Every strong IDU permuton $R$ is inverse-signature-invariant.
\end{theorem}
\begin{proof}
Consider $\pi_1$ and $\pi_2$ with $\pi_1^{-1} \udsim \pi_2^{-1}$.
Every strong IDU permuton $R$ can be approximated by a distribution of ID combinations $C$ within any desired $\delta > 0$ (by Theorem~\ref{thm:approx-by-ID}). By Corollary~\ref{cor:profile-us-signature}, for every such ID combination $C$, we have $d(\pi_1, C) = d(\pi_2, C)$. 
Considering the limit when $\delta$ goes to 0, we get that $d(\pi_1, R) = d(\pi_2, R)$.
\end{proof}

Thus all permutations are divided into equivalence classes according to the up--down signatures of their inverses. The number of such equivalence classes for permutations in $\Sk$ is $2^{k-1}$ (the number of possible up--down signatures of length $k-1$). Taking into account that the sum of densities of all $\rho\in\Sk$ is 1, we get the following corollary.
\begin{corollary} \label{cor:dim-profile-space}
All $k$-profiles of strong IDU transformations lie in an affine space of dimension $2^{k-1} - 1$.
\end{corollary}

Now we provide a characterization of strong IDU transformations using inverse-signature invariance.
\begin{theorem}\label{thm:ud-invariance-implies-strong} $\quad$ \\
1. If a random permuton $R$ is inverse-signature-invariant, then it defines a strong IDU transformation.

\noindent 2. If a weak IDU transformation $\{\sigma_n\}$ is inverse-signature-invariant, then it is a strong IDU transformation.
\end{theorem}
\begin{proof}
1. Consider an inverse-signature-invariant random permuton $R$. $R$ defines a consistent weak IDU transformation $\{\sigma_n\}$ (see Definition~\ref{def:consistency}).
We will need the following auxiliary definition.
\begin{definition}
Consider a permutation $\beta \in \Sn$. We define the operation of inserting a new entry $(x,y)$ into $\beta$ with $x,y\in [n+1]$. To provide more intuition, we offer three equivalent definitions.
Write $\beta$ in one-line notation as $(\beta(1), \dots, \beta(n))$. Insert the value $y$ between positions $x-1$ and $x$ (after the change, $y$ will be at position $x$). After that, increment all previously present values $\beta(i) \geq y$ by 1 (to avoid repetitions). We obtain the one-line representation of $\beta_{(x,y)}$.

Now we provide a more explicit description of the same procedure. 
Consider the graph of $\beta$: $\{(i, \beta(i)) : i \in [n]\}$ (see Figure~\ref{fig:entry-insertion}). 
Modify the graph as follows: move all points to the right of $x$ one unit to the right; 
move all points above $y$ one unit up; then add the point $(x,y)$. We obtain a set of $n+1$ points:
$$
\{(\operatorname{shift}_x(i), \operatorname{shift}_y(\beta(i))): i\in [n]\} \cup \{(x,y)\} \quad \text{where } \operatorname{shift}_t(j) = 
\begin{cases}
  j,& \text{if } j < t\\
  j+ 1, & \text{if } j \geq t.
\end{cases}
$$
Let $\beta_{(x,y)}$ be the permutation on $[n+1]$ with this graph.

Finally, note that $\beta_{(x,y)}$ is the pattern of the set of points $\{(i, \beta(i)): i\in [n]\} \cup \{(x-\nicefrac12, y-\nicefrac12)\}$.
\end{definition}

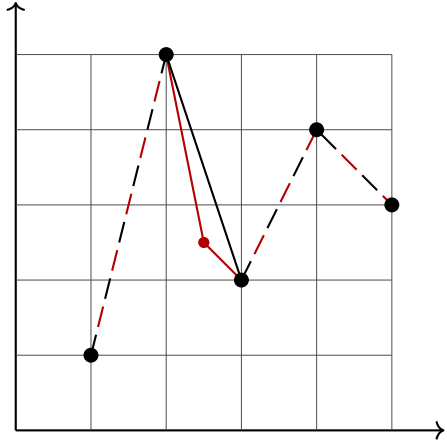
\begin{figure}
  \centering
\begin{tikzpicture}[scale=1]
  \tikzset{
    axis/.style={black, thick, ->},
    gridblue/.style={black!70, very thin},
    dotblack/.style={black, fill=black, circle, inner sep=2pt},
    dotsmall/.style={black, fill=black, circle, inner sep=1pt},
    dotred/.style={darkred, fill=darkred, circle, inner sep=1.5pt},
    dotblue/.style={darkblue, fill=darkblue, circle, inner sep=1.5pt}
  }
  \tikzset{
    alt_r_b dashed/.style={
      dash pattern=on 8pt off 14pt,
      draw=black,
      postaction={
        draw,
        darkred,
        dash pattern=on 8pt off 14pt,
        dash phase=11pt
      }
    }
  }

  \draw[axis] (0,0) -- (5.7,0);
  \draw[axis] (0,0) -- (0,5.7);

  \foreach \x in {1,2,3,4,5} {
    \draw[gridblue] (\x,0) -- (\x,5);
  }
  \foreach \y in {1,2,3,4,5} {
    \draw[gridblue] (0,\y) -- (5,\y);
  }

  \draw[darkred, thick]   (2,5) -- (2.5,2.5)--(3,2); 
  \draw[black, thick]  (2,5) -- (3,2);
  \draw[alt_r_b dashed, thick] (1,1) -- (2,5);
  \draw[alt_r_b dashed, thick] (3,2) -- (4, 4) -- (5,3);
  \node[dotred] at (2.5, 2.5) {};

  \node[dotblack] at (1,1) {};
  \node[dotblack] at (2,5) {};
  \node[dotblack] at (3,2) {};
  \node[dotblack] at (4,4) {};
  \node[dotblack] at (5,3) {};

\end{tikzpicture}

\caption{The figure shows the plot of the permutation \(\beta = (\mathsf{1\ 5\ 2\ 4\ 3})\) (black dots connected by black lines) and the permutation \(\beta_{(3,3)} = (\mathsf{1\ 6\ 3\ 2\ 5\ 4})\) (black dots together with one red dot, connected by red lines). Segments common to both plots are indicated by alternating black--red dashed lines.}
\label{fig:entry-insertion}
\end{figure}

\begin{observation}\label{obs:insertions}
Note that for $\beta\in \Sn$ and $i\in [n]$, 
$$\beta_{(i,\beta(i))} = \beta_{(i+1, \beta(i)+1)} \qquad \text{and}\qquad 
\beta_{(i,\beta(i)+1)} = \beta_{(i+1, \beta(i))}.$$
Both permutations in the first identity are obtained by replacing $(i,\beta(i))$ with a two-point pattern $(\mathsf{1\ 2})$; both permutations in the second identity are obtained by replacing $(i,\beta(i))$ with a two-point pattern $(\mathsf{2\ 1)}$.
\end{observation}

Now we prove the key lemma, from which the theorem will follow.
\begin{lemma}\label{lem:remove-one}
Let $\rho \in \Sn$. Let $J$ and $J'$ be subsets of $[n + 1]$ of size $n$ (that is, each is obtained by removing a single element from $[n+1]$). Then,
$$d_J(\rho, \sigma_{n+1}) = d_{J'}(\rho, \sigma_{n+1}).$$
\end{lemma}
\begin{proof}
We first prove the statement for $J = [n+1] \setminus \{q\}$ and $J' = [n+1] \setminus \{q+1\}$ for some $q\in [n]$.

Observe that $\patt(\sigma_{n+1}, J) = \rho$ if and only if $\sigma_{n+1} = \rho_{(q,y)}$ for some $y\in [n+1]$.
Therefore,
$$d_J(\rho, \sigma_{n+1}) = \sum_{y=1}^{n+1} \Pr(\sigma_{n+1} = \rho_{(q,y)}).$$
Similarly, 
$$d_{J'}(\rho, \sigma_{n+1}) = \sum_{y=1}^{n+1} \Pr(\sigma_{n+1} = \rho_{(q+1,y)}).$$
By our assumption the IDU transformation $\{\sigma_n\}$ is inverse-signature-invariant, and thus 
$\Prob{\sigma_{n+1} = \tau}$ depends only on the up--down signature of $\tau^{-1}$.
Therefore, to prove that $d_J(\rho, \sigma_{n+1}) = d_{J'}(\rho, \sigma_{n+1})$, it suffices to show that for every $y\in [n+1]$, the multisets of up--down signatures 
$$\{\udsign((\rho_{(q,y)})^{-1})\}_{y\in[n+1]}\quad\text{and}\quad \{\udsign((\rho_{(q+1,y)})^{-1})\}_{y\in[n+1]}$$ are equal (including multiplicities).

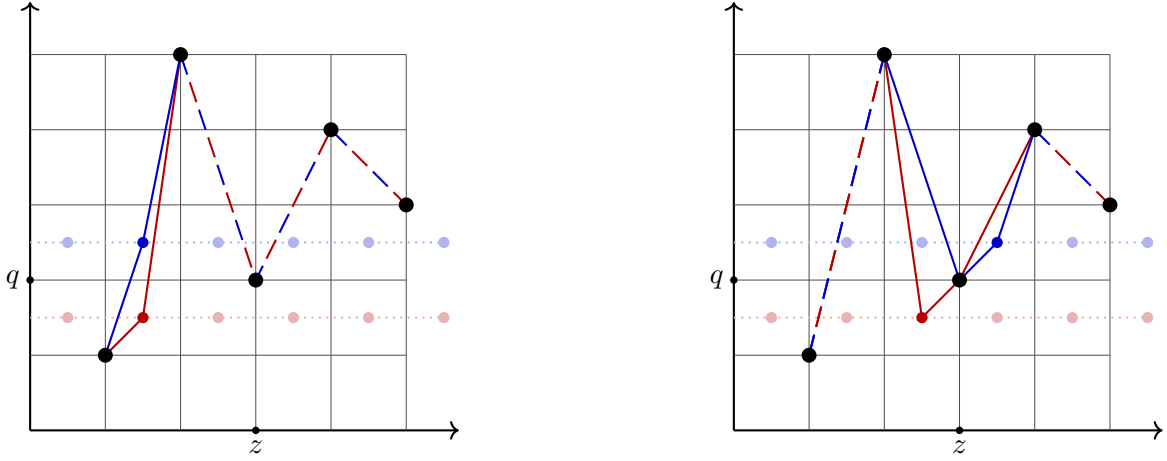
\begin{figure}
  \centering
\begin{tikzpicture}[scale=1]
  \tikzset{
    axis/.style={black, thick, ->},
    gridblue/.style={black!70, very thin},
    dotblack/.style={black, fill=black, circle, inner sep=2pt},
    dotsmall/.style={black, fill=black, circle, inner sep=1pt},
    dotred/.style={darkred, fill=darkred, circle, inner sep=1.5pt},
    dotredinactive/.style={darkred, fill=darkred!30, circle, inner sep=1.5pt},
    dotblue/.style={darkblue, fill=darkblue, circle, inner sep=1.5pt},
    dotblueinactive/.style={darkblue, fill=darkblue!30, circle, inner sep=1.5pt}
  }
  \tikzset{
    alt_r_b dashed/.style={
      dash pattern=on 8pt off 14pt,
      draw=darkblue,
      postaction={
        draw,
        darkred,
        dash pattern=on 8pt off 14pt,
        dash phase=11pt
      }
    }
  }

  \draw[axis] (0,0) -- (5.7,0);
  \draw[axis] (0,0) -- (0,5.7);

  \foreach \x in {1,2,3,4,5} {
    \draw[gridblue] (\x,0) -- (\x,5);
  }
  \foreach \y in {1,2,3,4,5} {
    \draw[gridblue] (0,\y) -- (5,\y);
  }

  \draw[darkred!30, dotted, thick] (0,1.5) -- (5.5,1.5);
  \draw[darkblue!30, dotted, thick] (0,2.5) -- (5.5,2.5);

  \draw[darkred, thick]  (1,1) -- (1.5,1.5) -- (2,5); 
  \draw[darkblue, thick]  (1,1) -- (1.5,2.5) -- (2,5); 
  \draw[alt_r_b dashed, thick] (2,5)--(3,2)--(4,4)--(5,3);

  \foreach \x in {0.5,2.5,3.5,4.5,5.5} {
    \node[dotredinactive] at (\x, 1.5) {};
  }
  \node[dotred] at (1.5, 1.5) {};

  \foreach \x in {0.5,2.5,3.5,4.5,5.5} {
    \node[dotblueinactive] at (\x, 2.5) {};
  }
  \node[dotblue] at (1.5, 2.5) {};

  \node[dotblack] at (1,1) {};
  \node[dotblack] at (2,5) {};
  \node[dotblack] at (3,2) {};
  \node[dotblack] at (4,4) {};
  \node[dotblack] at (5,3) {};

  \node[dotsmall] at (3, 0) {};
  \node[dotsmall] at (0, 2) {};
  \node[below] at (3,0) {$z$};
  \node[left]  at (0,2) {$q$};
\end{tikzpicture}
\hspace{3cm}
\begin{tikzpicture}[scale=1]
  \tikzset{
    axis/.style={black, thick, ->},
    gridblue/.style={black!70, very thin},
    dotblack/.style={black, fill=black, circle, inner sep=2pt},
    dotsmall/.style={black, fill=black, circle, inner sep=1pt},
    dotred/.style={darkred, fill=darkred, circle, inner sep=1.5pt},
    dotredinactive/.style={darkred, fill=darkred!30, circle, inner sep=1.5pt},
    dotblue/.style={darkblue, fill=darkblue, circle, inner sep=1.5pt},
    dotblueinactive/.style={darkblue, fill=darkblue!30, circle, inner sep=1.5pt}
  }
  \tikzset{
    alt_r_b dashed/.style={
      dash pattern=on 8pt off 14pt,
      draw=darkblue,
      postaction={
        draw,
        darkred,
        dash pattern=on 8pt off 14pt,
        dash phase=11pt
      }
    }
  }
  \draw[axis] (0,0) -- (5.7,0);
  \draw[axis] (0,0) -- (0,5.7);

  \foreach \x in {1,2,3,4,5} {
    \draw[gridblue] (\x,0) -- (\x,5);
  }
  \foreach \y in {1,2,3,4,5} {
    \draw[gridblue] (0,\y) -- (5,\y);
  }

  \draw[darkred!30, dotted, thick] (0,1.5) -- (5.5,1.5);
  \draw[darkblue!30, dotted, thick] (0,2.5) -- (5.5,2.5);

  \draw[darkred, thick]
    (2,5) -- (2.5,1.5) -- (3,2) -- (4,4);
  \draw[darkblue, thick]
    (2,5) -- (3,2) -- (3.5, 2.5) -- (4,4);    
  \draw[alt_r_b dashed, thick] (1,1) -- (2,5);
  \draw[alt_r_b dashed, thick] (1,1) -- (2,5);
  \draw[alt_r_b dashed, thick] (4,4) -- (5,3);

  \foreach \x in {0.5,1.5,3.5,4.5,5.5} {
  \node[dotredinactive] at (\x, 1.5) {};
  }
  \node[dotred] at (2.5, 1.5) {};

  \foreach \x in {0.5,1.5,2.5,4.5,5.5} {
    \node[dotblueinactive] at (\x, 2.5) {};
  }
  \node[dotblue] at (3.5, 2.5) {};

  \node[dotblack] at (1,1) {};
  \node[dotblack] at (2,5) {};
  \node[dotblack] at (3,2) {};
  \node[dotblack] at (4,4) {};
  \node[dotblack] at (5,3) {};

  \node[dotsmall] at (3, 0) {};
  \node[dotsmall] at (0, 2) {};
  \node[below] at (3,0) {$z$};
  \node[left]  at (0,2) {$q$};
\end{tikzpicture}

\caption{The figure shows the plot of the permutation $\tau = (\mathsf{1\ 5\ 2\ 4\ 3})$ (black dots). Recall that $\tau = \rho^{-1}$.
The plot of $\tau_{(y,q)}$ is obtained by inserting one red dot into the plot (and shifting the other dots accordingly), 
while the plot of $\tau_{(y,q+1)}$ is obtained by inserting one blue dot. 
In this example, $q = 2$. 
The left panel illustrates that the up--down signatures of $\tau_{(y,q)}$ and $\tau_{(y,q+1)}$ coincide for $y = 2$; 
the ud-signatures at positions $1$ and $2$, where the permutations differ, are $uu$ in both cases. 
The right panel shows that the signatures of $\tau_{(y,q)}$ and $\tau_{(y+1,q+1)}$ coincide for $y = y_0 = 3$.}
\label{fig:perm-plot}
\end{figure}
At this point, it will be convenient to switch from permutations $\rho$ and $\rho_{(x,y)}$ to their inverses (see Figure~\ref{fig:perm-plot}). Let $\tau = \rho^{-1}$. Note that $(\rho_{(x,y)})^{-1} = \tau_{(y,x)}$. Now, we need to show 
$$\{\udsign(\tau_{(y,q)})\}_{y\in[n+1]} = \{\udsign(\tau_{(y,q+1)})\}_{y\in[n+1]}$$
(where both sides are multisets).
For a fixed $y$, permutations $\tau_{(y,q)}$ and $\tau_{(y,q+1)}$ differ only at two positions:
\begin{itemize}
\item $\tau_{(y,q)}: y \mapsto q$\quad and \quad $\tau_{(y,q+1)}: y \mapsto q+1$.
\item for $z = \operatorname{shift}_y(\tau^{-1}(q))$:\quad $\tau_{(y,q)}: z \mapsto q+1$ \quad and \quad $\tau_{(y,q+1)}: z \mapsto q$.
\end{itemize}
Further, the difference between values at these positions is exactly 1. 
Therefore, the only pair of indices $\{a,b\}$ such that $\tau_{(y,q)}(a) < \tau_{(y,q)}(b)$ 
but $\tau_{(y,q+1)}(a) > \tau_{(y,q+1)}(b)$ or the other way around is $\{y, z\}$.

If $y$ and $z$ are not next to each other (that is, $|y-z| > 1$), then both permutations $\tau_{(y,q)}$ and $\tau_{(y,q+1)}$ have the same up--down signatures (even though $\tau_{(y,q)}$ and $\tau_{(y,q+1)}$ are not equal).

The only cases when the signatures may differ are when $z = y-1$ or $z = y+1$. That is, when $y= y_0\eqdef \tau^{-1}(q)$ or $y = y_0+1$. Since $\tau(y_0) = q$, by Observation~\ref{obs:insertions}, we have
$$
\tau_{(y_0,q)} = \tau_{(y_0 + 1,q+1)} \quad \text{and} \quad \tau_{(y_0+1,q)} = \tau_{(y_0,q+1)}
$$
and thus
$$
\udsign(\tau_{(y_0,q)}) = \udsign(\tau_{(y_0 + 1,q+1)}) \quad \text{and} \quad \udsign(\tau_{(y_0+1,q)}) = \udsign(\tau_{(y_0,q+1)}).
$$
We showed that $d_J(\rho, \sigma_{n+1}) = d_{J'}(\rho, \sigma_{n+1})$ for $J = [n+1] \setminus \{q\}$ and $J' = [n+1] \setminus \{q+1\}$ for all $q$. 
This implies that this equality actually holds for all $J, J' \subset [n+1]$ of size $n$.
\end{proof}
We have shown that the distribution of $\patt(\sigma_{n+1}, J)$ does not depend on the particular choice of $J$ of size~$n$. Further, let $\hat J$ be a random $n$-element subset of $[n+1]$. Since $\bm{\sigma}$ is a consistent IDU, 
we get for every subset $J$ of size $n$
$$\patt(\sigma_{n+1}, J) \overset{d}{=} \patt(\sigma_{n+1}, \hat J) \overset{d}{=} \sigma_n.$$

Now consider a subset $J_k$ of size $k \leq n$. For $i \in \{k+1, \dots, n+1\}$, let $J_i$ be the set obtained by adding the $i-k$ smallest elements from $[n+1]\setminus J_k$ to $J_k$. Then
\[
J_k \subset J_{k+1} \subset \dots \subset J_n \subset J_{n+1} = [n+1]
\]
with $|J_i| = i$ for all $i$.

We prove by induction on $i$, going from $n$ down to $k$, that
$$\patt(\sigma_{n+1}, J_i) \overset{d}{=} \sigma_i.$$
We have already established the base case $i=n$. Let $a$ be the single element of $J_{i+1} \setminus J_i$; note that $a$ has rank $a$ in $J_{i+1}$. Then
$$\patt(\sigma_{n+1}, J_i) =
\patt(\patt(\sigma_{n+1}, J_{i+1}), [i+1]\setminus \{a\}]) 
\overset{d}{=}
\patt(\sigma_{i+1}, [i+1]\setminus \{a\}) 
\overset{d}{=} \sigma_i.
$$

Hence, for every $J \subseteq [n+1]$, the distribution of $\patt(\sigma_{n+1}, J)$ depends only on the size of~$J$. Replacing $n+1$ with $n$, we get
\[
d_J(\rho, \sigma_n)
= d(\rho, \sigma_n)
= d(\rho, R),
\]
as required. We conclude that $R$ is a strong IDU transformation.

2. Now consider an inverse-signature-invariant weak IDU transformation $\bm{\sigma}$. We verify that $\bm{\sigma}$ satisfies the condition of consistency and thus is defined by a random permuton $R$. Accordingly, item 1 applies.

Indeed, the consistency condition requires that for every $\rho\in \Sk$, every $n \geq k$,
$$d(\rho, \sigma_n) = d(\rho, \sigma_k).$$
In our case, both are equal to $\Delta_{\bm{\sigma}}(\udsign(\rho^{-1}))$.
\end{proof}

In this section, we identified a set of linear dependencies for pattern densities of strong IDU permutons $R$. Namely, 
\begin{equation}\label{eq:dependencies}
d(\rho_1, R) - d(\rho_2, R) = 0 \qquad\text{for }
\rho_1^{-1} \udsim \rho_2^{-1}.
\end{equation}
Also, pattern densities satisfy arbitrary linear combinations of \eqref{eq:dependencies}.
To complete the picture, we show that there are no other nontrivial linear dependencies among permutation densities $\{d(\sigma, R)\}$ that hold for all strong IDU permutons $R$.
\begin{theorem}\label{thm:no-other-relations}
Let $\rho_1, \dots, \rho_t \in \Sn$ with $\rho_i^{-1} \not\udsim \rho_j^{-1}$ for all $i\neq j$, and let $c_1, \dots, c_t \in \mathbb R$. Assume that for every strong IDU permuton $R$, we have
$$
\sum_{i=1}^t c_i d(\rho_i , R) = 0
$$
Then all $c_i = 0$.
\end{theorem}
\begin{proof}
Without loss of generality, we may assume that all $c_i$ are nonzero; if some are zero, we exclude the corresponding terms. If $t=0$, then we are done. We now get a contradiction with the assumption that $t\geq 1$.

Let $\pi_j = \rho_j^{-1}$. For each permutation $\pi_j$, consider its up--down signature $\udsign(\pi_j)$ and partition it into maximal blocks of consecutive $u$’s and $d$’s, with the first block consisting of $u$’s (possibly empty).
Define the sequence $(a_1, a_2, \dots)$ as follows: 
\begin{itemize}
\item The length of the sequence equals the number of blocks. 
\item For $i \ge 2$, $a_i$ is the length of block~$i$.
\item $a_1$ is the length of block~1 plus~1. 
\end{itemize}
Note that $\sum a_i = n$. The sequence $(a_1, a_2, \dots)$ uniquely determines the up--down signature of~$\pi_j$, and thus the sequences for distinct $\pi_j$ are different.
Among all sequences for different $\pi_j$, take the largest one in lexicographic order and denote it by~$a^*$. Let~$N$ be its length. 
Denote the sequence for $\pi_{j^*}$ by $a^*$. 

Consider an ID combination $C$ as in~(\ref{eq:xy-ID-combination}):
\begin{equation} 
C = x_1 I \oUp y_1 D \oUp x_2 I \oUp y_2 D \oUp \dots \oUp x_N I \oUp y_N D.
\end{equation}
For now, the values of $x_i$ and $y_i$ are not fixed. 

Define a $\pi_{j^*}$-partition function $f$ that maps the first $a_1$ natural numbers to~$x_1$, the next $a_2$ numbers to~$y_1$, the next $a_3$ numbers to~$x_2$, and so on. 
It is immediate that $f$ is indeed a $\pi_{j^*}$-partition function. 
Further, observe that it is not a $\pi_j$-partition function for any $j\neq j^*$: if it were, then the corresponding sequence~$a'$ for $\pi_j$ would be larger than $a^*$ in lexicographic order. 

Note that $\sum_{i=1}^t c_i d(\rho_i, C)$ is a multivariate polynomial in $\{x_i, y_i\}_i$:
$$F(x,y) \eqdef \sum_{i=1}^t c_i d(\rho_i, C) = \sum_{i=1}^t c_i K_{\pi_i, N}(x,y).$$
By our assumption, $\sum_{i=1}^t c_i d(\rho_i, R) = 0$ for all strong IDU permutons $R$. In particular, $F(x,y) = 0$ as long as 
$\sum_i (x_i + y_i) = 1$ and all $x_i, y_i \geq 0$.
However, since $F$ is homogeneous, $F(\lambda x, \lambda y) = \lambda^{n} F(x, y)$, so $F$ remains~0 for all $x$ and $y$ with $x_i, y_i \geq 0$ even if $\sum (x_i + y_i) \ne 1$. Therefore, $F$ is identically zero as a polynomial; that is, all its coefficients are~0. 

The coefficient of $\prod_{i=1}^n f(i)$ in $K_{\pi_{j^*}, N}(x,y)$ equals $1$, and $\prod_{i=1}^n f(i)$ appears only in $K_{\pi_{j^*}, N}(x,y)$ (as $f$ is not a $\pi_j$-partition function for any other $\pi_j$). Hence, its coefficient in $F(x,y)$ equals $c_{j^*} \neq 0$. We obtain a contradiction.
\end{proof}

\section{Computing a nearly optimal strong IDU transformation}\label{sec:computation}
In this section, we discuss how to compute a nearly optimal strong IDU transformation for a given ordering CSP with predicates of arity at most~$k$.
By Theorem~\ref{thm:idu-framework}, our goal is to find a weak IDU transformation that maximizes $p(\Pi' \to \Pi, \bm{\sigma})$.  
As shown in Theorem~\ref{thm:main-weak-to-strong}, it is sufficient to consider only strong IDU transformations.

For a strong IDU transformation~$R$,
\[
p(\varphi' \to \varphi, R)
   = \min_{\tau \in \Sat(\varphi')} \sum_{\rho \in \Sat(\varphi)} d(\rho \tau^{-1}, R).
\]
By Theorem~\ref{thm:approx-by-ID}, every strong IDU transformation~$R$ can be approximated within any desired precision~$\delta > 0$ by a distribution over ID combinations with at most~$N$ terms, where $N = O(k^2 / \delta)$.

Each of the $k$-profile vectors corresponding to these ID combinations lies in the image of~$\mathbf{K}_N$ (see Theorem~\ref{thm:feasible-region-characterization}).  
Thus, the problem of finding an approximately optimal strong IDU transformation 
for $\CSP(\Pi)$ reduces to the following optimization problem:
\begin{quotation}
  \noindent Find $2^{k-1}$ pairs of vectors $(\mathbf{x}^{(i)}, \mathbf{y}^{(i)}) \in \mathbb{R}^{2N}$ and a vector of weights $\mathbf{p} = (p_1, \dots, p_{2^{k-1}})$
  such that $\|\mathbf{x}^{(i)}\|_1 + \|\mathbf{y}^{(i)}\|_1 = 1$ for all~$i$, $\|\mathbf{p}\|_1 = 1$, and all variables are nonnegative,
  in order to maximize
  \[
  \min_{\varphi\in \Pi}\min_{\tau \in \Sat(\varphi')} \sum_{\rho \in \Sat(\varphi)} \sum_{i=1}^{2^{k-1}} p_i K_{\tau \rho^{-1}, N}(\mathbf{x}^{(i)}, \mathbf{y}^{(i)}).
  \]
\end{quotation}

Note that $K_{\psi, N}$ is a homogeneous polynomial of degree~$k$ with positive terms for every $\psi \in \Sk$.  
Let $\mathbf{1}_N$ denote the vector of all ones. Then
\[
\frac{\partial K_{\psi, N}(\mathbf{x}, \mathbf{y})}{\partial x_i}
\leq 
\frac{\partial K_{\psi, N}(\mathbf{1}_N, \mathbf{1}_N)}{\partial x_i}
\leq k \cdot K_{\psi, N}(\mathbf{1}_N, \mathbf{1}_N)
= k \cdot (2N)^k \cdot K_{\psi, N}\!\left(\frac{\mathbf{1}_N}{2N}, \frac{\mathbf{1}_N}{2N}\right).
\]
Thus,
\[
\frac{\partial}{\partial x_i} \sum_{\rho \in \Sat(\varphi)} K_{\tau \rho^{-1}, N}
\leq k \cdot (2N)^k
\sum_{\rho \in \Sat(\varphi)} K_{\tau \rho^{-1}, N}\!\left(\frac{\mathbf{1}_N}{2N}, \frac{\mathbf{1}_N}{2N}\right)
\leq k \cdot (2N)^k,
\]
since $K_{\tau \rho^{-1}, N}(\mathbf{1}_N/(2N), \mathbf{1}_N/(2N))$ corresponds to a valid permuton $\frac{1}{2N} I \oUp \frac{1}{2N} D \oUp \dots\oUp \frac{1}{2N} I \oUp \frac{1}{2N} D$, and the sum of the densities of permutations $\tau \rho^{-1}$ in it is at most~$1$.
The same bound holds for the partial derivative with respect to~$y_i$.
It follows that, if two vectors $(\mathbf{x}, \mathbf{y})$ and $(\mathbf{x}', \mathbf{y}')$ differ by at most 
$\varepsilon = \delta / (2k\cdot N\cdot (2N)^k)$ in each coordinate, then the difference between the values of the objective function at these points is at most~$\delta$.  
Since $N=O(k^2/\delta)$,
$$\varepsilon = \delta / (2k\cdot N\cdot (2N)^k) \geq (\delta / k)^{O(k)}.$$

This optimization problem can be solved to any desired accuracy $O(\delta)$ by brute force, by discretizing both the simplex for each $(\mathbf{x}^{(i)}, \mathbf{y}^{(i)})$ with step~$\varepsilon$ and the simplex of weights~$\mathbf{p}$ with step~$\delta/2^{k-1}$.  
The total number of points is $(1 / \varepsilon)^{O(2^k N)} \cdot (2^{k-1}/\delta)^{2^{k-1}}= (k / \delta)^{O(2^k k^3 / \delta)}$.

The running time of the algorithm depends only on~$k$ and the desired precision~$\delta$.  
The resulting distribution over ID combinations has bit complexity depending only on~$k$ and~$\delta$.  
For every~$n$, we can sample a permutation of length~$n$ from this distribution in time $O(f(k, \delta)\, n)$, that is, in time linear in~$n$.

\begin{theorem}\label{thm:compute-nearly-optimal-IDU}
There exists an algorithm that, given a constraint language~$\Pi$ and its relaxation~$\Pi'$, as well as a desired precision~$\delta > 0$, finds a convex combination~$R'$ of ID combinations such that 
\[
p(\Pi' \to \Pi, R') \ge p(\Pi' \to \Pi, R^*) - \delta,
\]
where~$R^*$ is the optimal strong IDU transformation for~$\Pi$ and~$\Pi'$.  
The running time of the algorithm and the bit complexity of~$R'$ depend only on~$k$ and~$\delta$.

Given~$n$, we can sample a permutation~$\sigma_n$ of length~$n$ from~$R'$ in time $O(g(k, \delta)\, n)$, that is, in time linear in~$n$.
\end{theorem}

\section{Simple sufficient condition for nontrivial approximation}
In this section, we present a simple sufficient condition for 
a CSP problem $\CSP(\varphi)$ to have a nontrivial approximation or, in other words, to not be approximation resistant.
Consider a predicate $\varphi(x_1,\dots, x_k)$ and a relaxation $\varphi'$ of $\varphi$ (which may be $\varphi_{L}$, $\varphi_{R}$, or $\varphi_{\varepsilon}$). We are interested in finding a strong IDU permuton $R$ such that 
$$p(\varphi' \to \varphi, R) > \frac{|\Sat(\varphi)|}{k!}.$$
We begin with the following observation:
the uniform permuton $U$ gives a $\nicefrac{|\Sat(\varphi)|}{k!}$ approximation and thus our goal is to find a strong IDU permuton $R$ that beats $U$. This motivates the following question:
\begin{quote}
For a given $\varphi$, is there a ``perturbation'' $R$ of $U$ with $p(\varphi' \to \varphi, R) > p(\varphi' \to \varphi, U)$ (note the strict inequality)?
\end{quote}
Now, if for some $R$ we get $p(\varphi' \to \varphi, R) > p(\varphi' \to \varphi, U)$ then the inequality also holds for $R'_{\varepsilon} = \varepsilon R \oUp (1-\varepsilon) U$ and $R'_\varepsilon \to U$ as $\varepsilon \to 0$ (w.r.t.\ the standard weak-topology on the space of random permutons). However, we are \emph{not interested} in perturbations like this and instead will require that $R_{\varepsilon}$ is a mixture of deterministic permutons $P_{\varepsilon}$, and each $P_{\varepsilon}$ in the mixture converges to $U$ as $\varepsilon \to 0$.

We note that focusing only on perturbations of $U$ significantly restricts the set of possible options, and thus we can only hope to get a sufficient condition -- but not a criterion -- for when $\varphi$ admits a nontrivial approximation. 
Furthermore, to obtain a simple condition, we will only consider first-order perturbations, as described below.

Using the formula for $p(\varphi'\to \varphi, R)$, we see that $R$ yields a nontrivial approximation for $\varphi$ if and only if for every $\tau\in \Sat(\varphi')$,
$$\sum_{\rho\in\Sat(\varphi) \tau^{-1}} d(\rho, R) > \sum_{\rho\in\Sat(\varphi) \tau^{-1}} d(\rho, U),$$
or equivalently
$$\sum_{\rho\in\Sat(\varphi) \tau^{-1}} (d(\rho, R) - d(\rho,U)) > 0.$$
Previously, we defined the up--down signature of permutations as sequences of letters $u$ and $d$. In this section, it will be convenient to use $+1$ and $-1$ instead of $u$ and $d$, respectively. We will denote a vector signature with $\pm 1$ entries  by $\udsign_{\pm}(\tau)$. Further, for given $\varphi$ and $\tau\in \Sk$, let 
\begin{equation}\label{eq:vectors-v-varphi-tau}
  v(\varphi, \tau) = \sum_{\rho \in \Sat(\varphi)\tau^{-1}} \udsign_{\pm}(\rho^{-1}) \in \mathbb{Z}^{k-1}.
\end{equation}

Now, we define IDU combinations $A(x, \varepsilon)$ and $B(x,\varepsilon)$ as follows:
\begin{align*}
A(x,\varepsilon) &= \frac{x}{1+\varepsilon} U \oUp \frac{\varepsilon}{1+\varepsilon} I \oUp \frac{1-x}{1+\varepsilon} U\\
B(x,\varepsilon) &= \frac{x}{1+\varepsilon} U \oUp \frac{\varepsilon}{1+\varepsilon} D \oUp \frac{1-x}{1+\varepsilon} U.
\end{align*}
We think of both $A(x,\varepsilon)$ and $B(x,\varepsilon)$ as perturbations of $U$, since $A(x,\varepsilon) \to U$ and $B(x,\varepsilon) \to U$ as $\varepsilon \to 0$ for every fixed $x\in (0,1)$. We note that $A(\cdot)$ and $B(\cdot)$ are as powerful as arbitrary perturbations of $U$ for the purpose of the first-order analysis. 
For the sake of analysis, we also define:
$$C(x,\varepsilon) = \frac{x}{1+\varepsilon} U \oUp \frac{\varepsilon}{1+\varepsilon} U \oUp \frac{1-x}{1+\varepsilon} U =U.$$

Consider a permutation $\rho \in \Sk$. Let $s = \udsign_{\pm}(\rho^{-1})$.
Using formula~\eqref{eq:density-for-IDU}, we write approximations for $d(\rho, A(x,\varepsilon))$, $d(\rho, B(x,\varepsilon))$, and $d(\rho, C(x,\varepsilon))$ up to terms of order $O(\varepsilon^3)$. In~\eqref{eq:density-for-IDU}, we consider $\rho^{-1}$-partition functions $f$. Since the coefficients in the formulas for $A$, $B$, and $C$ are not variables, it is more convenient to assume that $f$ maps $[k]$ to the three terms in the up-combination formulas for $A$, $B$, and $C$, rather than to their coefficients. Accordingly, we let $f$ map the first $a$ numbers to the first term in the IDU combination, the next $0$, $1$, or $2$ numbers to the second term (partition functions that map more than two numbers to the second term contribute only $O(\varepsilon^3)$ terms), and the remaining ones to the third term. Since the definition of a partition function allows any consecutive block of numbers to be mapped to the same term of the form $c U$, the only restriction is the following: $f$ can map both $a+1$ and $a+2$ to $\frac{\varepsilon}{1+\varepsilon} I$ (resp.\ $\frac{\varepsilon}{1+\varepsilon} D$) only if $s_{a+1} = 1$ (resp.\ $s_{a+1} = -1$).

 Denote $\gamma = \frac{1}{(1+\varepsilon)^k}$. Below, ``'$\approx$'' will denote equality up to $O(\varepsilon^3)$. We have,
\begin{align*}
d(\rho, A(x,\varepsilon)) &\approx \gamma\left(\sum_{a=0}^k \frac{x^a(1-x)^{k-a}}{a!(k-a)!} + \sum_{a=0}^{k-1} \frac{x^a(1-x)^{k-1 - a}}{a!(k-1 - a)!}\varepsilon + 
\sum_{a:s_{a+1} = 1} \frac{x^a(1-x)^{k-2 - a}}{a!(k-2-a)!}\varepsilon^2
\right)
\\
d(\rho, B(x,\varepsilon)) &\approx \gamma\left(\sum_{a=0}^k \frac{x^a(1-x)^{k-a}}{a!(k-a)!} + \sum_{a=0}^{k-1} \frac{x^a(1-x)^{k-1 - a}}{a!(k-1 - a)!}\varepsilon + 
\sum_{a:s_{a+1} = -1} \frac{x^a(1-x)^{k-2 - a}}{a!(k-2-a)!}\varepsilon^2
\right)\\
d(\rho, C(x,\varepsilon)) &\approx \gamma\left(\sum_{a=0}^k \frac{x^a(1-x)^{k-a}}{a!(k-a)!} + \sum_{a=0}^{k-1} \frac{x^a(1-x)^{k-1 - a}}{a!(k-1 - a)!}\varepsilon + 
\sum_{a=0}^{k-2} \frac{x^a(1-x)^{k-2 - a}}{a!(k-2-a)!}\frac{\varepsilon^2}{2}
\right)
\end{align*}

We obtain simpler expressions for the following differences:
$$
d(\rho, A(x,\varepsilon)) - d(\rho, C(x,\varepsilon))
\approx
-(d(\rho, B(x,\varepsilon)) - d(\rho, C(x,\varepsilon)))
\approx 
\gamma
\sum_{a=0}^{k-2} s_{a+1} \frac{x^a(1-x)^{k-2 - a}}{a!(k-2-a)!}\frac{\varepsilon^2}{2}.
$$
Consider a random permuton $R_\varepsilon$ that is a mixture of permutons $A(x_i,\varepsilon)$ and $B(x_i,\varepsilon)$ with some nonnegative weights $\alpha_i$ and $\beta_i$ that add up to 1, where  $x_i$ are distinct reals from $(0,1)$ and $i\in [N]$ for some $N$. Using the fact that the mixture of $C(x_i, \varepsilon)$ with weights $\alpha_i + \beta_i$ equals $U$, we have
$$
d(\rho, R_\varepsilon) - d(\rho, U) \approx  
\frac{\gamma \varepsilon^2}{2} \sum_{i=1}^N\sum_{a=0}^{k-2}s_{a+1} 
\frac{x_i^a(1-x_i)^{k-2 - a} (\alpha_i -\beta_i)}{a!(k-2-a)!}.
$$
Therefore, 
$$\sum_{\rho \in \Sat(\varphi)\tau^{-1}} \left(d(\rho, R_\varepsilon) - d(\rho, U) \right)
\approx  
\frac{\gamma \varepsilon^2}{2} \sum_{a=0}^{k-2}v_{a+1}(\varphi, \tau)\sum_{i=1}^N
\frac{x_i^a(1-x_i)^{k-2 - a}(\alpha_i -\beta_i)}{a!(k-2-a)!}.
$$
Define a $(k-1) \times N$ matrix $M$ by $M_{bi} = \frac{x_i^{b-1} (1-x_i)^{k-1-b}}{(b-1)!(k-1-b)!}$. We have $M= D_1 V D_2$, where $V$ is a rectangular Vandermonde matrix with entries $\left(\nicefrac{x_i}{(1-x_i)}\right)^{b-1}$, $D_1$ and $D_2$ are diagonal matrices with diagonal entries $\nicefrac{1}{(b-1)!(k-1-b)!}$ and $(1-x_i)^{k-2}$, respectively.
As such, it has rank $\min(N, k-1)$. 
Let $\alpha$ be the vector of $\alpha_i$ and $\beta$ be the vector of $\beta_i$. 
$$\sum_{\rho \in \Sat(\varphi)\tau^{-1}} \left(d(\rho, R_\varepsilon) - d(\rho, U) \right)
\approx  
\frac{\gamma \varepsilon^2}{2} \cdot v(\varphi, \tau)^\transpose M\,(\alpha-\beta).
$$

Our goal now is to choose $N$, distinct $x_i\in (0,1)$, and $\alpha$ and $\beta$ so that
$v(\varphi, \tau)^\transpose\, M\,(\alpha-\beta) > 0$ for every $\tau \in \Sat(\varphi')$.
We observe that, for every vector $y\in \mathbb{R}^{k-1}$, we can choose $N$, $x_i$, $\alpha$, and $\beta$ so that $M\,(\alpha-\beta) = \delta y$ for some $\delta > 0$. Indeed, choose $N = k -1$ and pick arbitrary distinct $x_i \in (0,1)$. Now, let $z = M^{-1} y$. Define 
$\tilde \alpha_i = z_i$ if $z_i > 0$ and $0$, otherwise; 
$\tilde \beta_i = - z_i$ if $z_i < 0$ and $0$ otherwise.
Note that $z = \tilde \alpha - \tilde \beta$. Finally, normalize $\tilde \alpha$ and $\tilde \beta$ by setting $\delta = 1/(\|\tilde \alpha\|_1 + \|\tilde \beta\|_1)$, $\alpha = \delta \tilde \alpha$ and $\beta =\delta \tilde \beta$. Then 
$$M\,(\alpha-\beta) = \delta M\,(\tilde\alpha-\tilde\beta) = \delta M z = 
\delta y.$$

We obtain the following sufficient condition for the existence of the first-order perturbation $R_\varepsilon$ for $\varphi$: there exists a vector $y\in\mathbb{R}^{k-1}$ such that $\langle y, v(\varphi, \tau)\rangle > 0$ for all $\tau\in \Sat(\varphi')$.\footnote{The condition is also necessary as long as we are restricted to the first-order perturbation analysis.} Using LP duality, this condition can be stated as follows.

\begin{theorem}\label{thm:simple-sufficient-admits-approx}
Consider a predicate $\varphi$ of arity $k$ and a relaxation $\varphi'$ of $\varphi$.
Assume that $0$ is not in the convex hull of the vectors $\{v(\varphi, \tau)\}_{\tau\in \Sat(\varphi')}$ 
defined by \eqref{eq:vectors-v-varphi-tau} or, equivalently, for some $y\in\mathbb{R}^{k-1}$,
$$\langle y, v(\varphi, \tau)\rangle > 0
\text{ for all } \tau\in \Sat(\varphi').
$$
Then $p(\varphi'\to \varphi,R) > \nicefrac{|\Sat(\varphi)|}{k!}$
for some strong IDU permuton $R$. If $\varphi' = \varphi_{L}$ or $\varphi' = \varphi_{R}$, then $\varphi$ admits a nontrivial approximation in the satisfiable regime. If $\varphi' = \varphi_{\varepsilon}$, then $\varphi$ admits a nontrivial approximation in both the satisfiable and nearly-satisfiable regimes.
\end{theorem}

\section{Flag algebras for strong IDU transformations}\label{sec:flag-algebras}
\subsection{Overview of flag algebras for permutations}

We begin with an overview of the flag algebra method for permutations.
Flag algebras were introduced by Razborov~\cite{Razborov} and provide a powerful framework for obtaining density inequalities for graphs, permutations, and other combinatorial structures. Following the pioneering work of Balogh, Hu, Lidick{\`y}, Pikhurko, Udvari, and Volec~\cite{BHLPUV15}, which applied flag algebras to permutations, the technique has been widely used to study permutation pattern densities. Consequently, it provides a strong tool for proving density inequalities -- particularly upper bounds -- for IDU transformations.

We refer the reader to~\cite{Razborov, BHLPUV15, crudele2024six} for a detailed introduction to flag algebras.
Here, we provide only the necessary background.
We write permutations in one-line notation; for example, $(\mathsf{1,2,3})$ denotes the identity permutation of length~3, and $(\mathsf{2,1,3})$ is the transposition that swaps the first two elements.

A flag algebra generalizes the linear space of density sums introduced in Section~\ref{sec:pattern-density-sums}.
It extends simple density sums in two important ways.
First, we identify two density sums if they are equal in all permutons.
For instance, since $d((\mathsf{1\,2}), P) + d((\mathsf{2\,1}), P) = 1$ for every permuton~$P$, we write $(\mathsf{1\,2}) + (\mathsf{2\,1}) = 1$.
Second, in addition to summation, we define a multiplication operation on density sums such that for density sums $\Sigma_1$ and $\Sigma_2$:
\[
d(\Sigma_1 \cdot \Sigma_2, P) = 
d(\Sigma_1, P)\cdot d(\Sigma_2, P)  
\]
for every \textit{deterministic} permuton~$P$. Note that this multiplication operation is unrelated to permutation multiplication (composition).

Most importantly, we consider \textit{typed} density sums.
Recall that given a permuton~$P$, we can sample a permutation $\sigma_n$ of length~$n$ by choosing $n$ points $(x_1, y_1), \dots, (x_n, y_n)$ according to~$P$ and taking the pattern defined by these points. 
Now suppose we have fixed a subset~$S$ of~$k$ points in the support of~$P$ with no two points in $S$ sharing the same $x$- or $y$-coordinate.
We call these points \emph{fixed} or \emph{labeled points} and refer to $(P, S)$ as a \emph{tagged permuton}.
To sample a permutation~$\sigma_n$ from $(P, S)$, we sample the remaining $n-k$ points and define $\sigma_n$ as the pattern of all~$n$ points.
We say that an index~$i$ is \emph{labeled} or \emph{tagged} if its corresponding point is one of the $k$ fixed points.
Let~$J$ be the set of all labeled indices.
The pair $(\sigma_n, J)$ is then a \emph{flag}.
Thus, we have described a procedure for sampling flags from a tagged permuton $(P, S)$.

Note that the pattern $\tau = \patt(\sigma_n, J)$ coincides with the one defined by~$S$ and therefore does not depend on the specific random permutation sampled from~$(P, S)$.
We refer to~$\tau$ as the \emph{type} of both $(P, S)$ and $(\sigma_n, J)$.
We define the density $d((\rho, J), (P, S))$ of a flag $(\rho, J)$ in $(P, S)$ as the probability that the sampled flag equals $(\rho, J)$.

As is standard in the literature, when we look at specific examples, we will underline the indices that belong to~$J$.
For example, $(\mathsf{\underline{1}\,3\,\underline{4}\,2})$ denotes the flag $((\mathsf{1\,3\,4\,2}), \{1,3\})$ of type $(\mathsf{1\,2})$.
As we will see shortly, we can define addition and multiplication of flags of the same type.
We now provide the formal definitions.

\begin{definition}\label{def:flag}
Let $\tau \in \Sk$.
A \textit{flag} of type~$\tau$ is a pair $(\pi, J)$, where $\pi \in \mathbb{S}_n$ for some $n \ge k$ and $J \subseteq [n]$ is a set of size~$k$ such that $\patt(\pi, J) = \tau$.
We denote by~$\mathcal{F}^\tau$ the set of all flags of type~$\tau$.

Further, let~$\mathbb{R}\mathcal{F}^\tau$ denote the set of all formal linear combinations of flags of type~$\tau$ with real coefficients.
We refer to elements of~$\mathbb{R}\mathcal{F}^\tau$ as \textit{typed density sums}.
Addition of typed density sums is defined in the standard way, and the density function~$d(\cdot, (P, S))$ extends to them by linearity.

We identify $(\tau, [k])$ with~1 in~$\mathbb{R}\mathcal{F}^\tau$, since $d((\tau, [k]), (P, S))$ is always equal to~1.

In this definition, we allow $k=0$. In this case, $J = \varnothing$ and $S = \varnothing$; all flags have the trivial type $\varnothing$.
\end{definition}

We now define multiplication of typed density sums.

\begin{definition}\label{def:flag-multiplication}
Fix a type $\tau \in \Sk$.
Consider two flags 
 $(\pi_1, J_1)$ and $(\pi_2, J_2)$ of type~$\tau$, where $\pi_1 \in \mathbb{S}_{n_1}$ and $\pi_2 \in \mathbb{S}_{n_2}$. 
Let $n = n_1 + n_2 -k$.
We define their product $(\pi_1, J_1) \cdot (\pi_2, J_2)$ as
\[
(\pi_1, J_1) \cdot (\pi_2, J_2) =
\sum_{\substack{\pi \in \Sn\\ J:|J| = k}} c_{\pi,J}\, (\pi, J),
\]
where coefficients $c_{\pi,J}$ are defined as follows.
Sample a random subset~$A$ of size~$n_1 - k$ from $[n] \setminus J$, and let $B = [n] \setminus (J \cup A)$.
Then $c_{\pi,J}$ is the probability of the following event:
\begin{itemize}
  \item $\patt(\pi, A \cup J) = \pi_1$,
  \item the positions of indices from~$J$ in the sorted list $A \cup J$ equal~$J_1$,
  \item $\patt(\pi, B \cup J) = \pi_2$, and
  \item the positions of indices from~$J$ in the sorted list $B \cup J$ equal~$J_2$.
\end{itemize}
The definition of multiplication is extended to typed density sums by linearity.
\end{definition}

Importantly, for any tagged permuton~$(P, S)$ of type~$\tau$ we have the identity
\begin{equation}\label{eq:flag-multiplication-key-property}
d((\pi_1, J_1) \cdot (\pi_2, J_2), (P, S))
 = d((\pi_1, J_1), (P, S)) \cdot d((\pi_2, J_2), (P, S)).
\end{equation}

Next, we identify typed density sums that take the same value in all tagged permutons of type~$\tau$. To this end, we first describe sums that are identically zero.
While this set is often specified explicitly, we instead define it as the principal ideal generated by a specific element $\beta -1$. 

\begin{definition}
Consider a type $\tau \in \Sk$.
Define
\[
\beta = \sum_{\substack{(\sigma,J)\text{ is of type }\tau\\ \sigma\in \mathbb{S}_{k+1}}} (\sigma,J).
\]
Let $K^{\tau} = \langle \beta - 1 \rangle$ be the principal ideal in~$\mathbb{R}\mathcal{F}^\tau$ generated by~$\beta - 1$.
We define the \textit{flag algebra}~$\mathcal{A}^\tau$ as the quotient algebra
\[
\mathcal{A}^\tau = \mathbb{R}\mathcal{F}^\tau / K^{\tau}.
\]
\end{definition}
We choose $\beta$ as above since $d(\beta, (P, S)) = 1$ for every tagged permuton $(P, S)$ of type $\tau$, and consequently $d(g, (P, S)) = 0$ for every $g \in K^{\tau}$.

The elements of~$\mathcal{A}^\tau$ are equivalence classes of typed density sums.
We note that the function $g \mapsto d(g, (P, S))$ is 
a homomorphism from~$\mathcal{A}^\tau$ to~$\mathbb{R}$.
Moreover, for every flag $f \in \mathcal{F}^\tau$, the density $d(f, (P, S))$ is positive. 
Hence, we refer to this homomorphism as a \emph{positive homomorphism}.
The converse also holds: every positive homomorphism from~$\mathcal{A}^\tau$ to~$\mathbb{R}$ is of the form $d(\cdot, (P, S))$ for some tagged permuton~$(P, S)$ of type~$\tau$.

\subsection{Flag algebras and semidefinite programming}

From a practical standpoint, the most important property of flag algebras is that for every $g \in \mathcal{A}^\tau$, the density of~$g^2 \eqdef g \cdot g$ in any tagged permuton $(P, S)$ of type~$\tau$ is nonnegative:
\[
d(g^2, (P, S)) = d(g, (P, S))^2 \ge 0.
\]

As our ultimate goal is to prove inequalities for untagged density sums, we consider the following operator corresponding to averaging over all random choices of fixed points in a tagged permuton.

\begin{definition}
The operation $\llbracket \cdot \rrbracket : \mathcal{A}^\tau \to \mathcal{A}^{\varnothing}$ is defined by
\[
\llbracket (\pi, J) \rrbracket = \frac{1}{\binom{n}{k}} \pi.
\]
(As is standard, we identify $\pi$ and $(\pi, \varnothing)$ in~$\mathcal{A}^{\varnothing}$.)
\end{definition}

Importantly, $\llbracket \cdot \rrbracket$ is a linear operator, though \emph{not} a homomorphism.
This discussion leads to a semidefinite programming framework for establishing density inequalities.
\begin{fact}
For every $g \in \mathcal{A}^\tau$ and every permuton~$P$,
\[
d(\llbracket g^2 \rrbracket, P) \ge 0.
\]
\end{fact}

Consider a type~$\tau \in \Sk$ and a finite set of flags $f_1, \dots, f_n$ of type~$\tau$.  
Let $F$ be the matrix with entries $F_{ij} = \llbracket f_i \cdot f_j \rrbracket$.  
Then, for every positive semidefinite matrix $Q \in \mathbb{R}^{n \times n}$, we have
\[
Q \circ F = \sum_{i,j} Q_{ij} \llbracket f_i f_j \rrbracket
   = \left\llbracket \sum_{i,j} Q_{ij} f_i f_j \right\rrbracket .
\]

Note that
\[
\sum_{i,j} Q_{ij} f_i f_j
\]
is a sum of squares of elements of $\mathcal{A}^{\tau}$ (to see this, write $Q$ as a sum of rank-one positive semidefinite matrices). Therefore, $Q \circ F \ge 0$ constitutes a valid density inequality.

This yields a practical approach to proving density inequalities: choose an appropriate type~$\tau$ and flags $f_1, \dots, f_n$, construct the matrix~$F$, and find a matrix~$Q$ that produces the desired inequality via semidefinite programming.

In this section, we have discussed pattern and flag densities only for \emph{deterministic} permutons.  
For example, identity~\eqref{eq:flag-multiplication-key-property} holds for every deterministic tagged permuton but does not generally hold for random permutons.  
However, \emph{any linear inequality} on pattern densities that holds for deterministic permutons also holds for random permutons, since the densities of a random permuton are convex combinations of those of deterministic ones.

\subsection{Flag algebra for strong IDU transformations}
We now adapt the flag algebra framework to analyze strong IDU transformations.  
As we have seen, strong IDU transformations are inverse-signature-invariant and we want to incorporate this property into the flag algebra framework.
Since the above condition involves inverses of permutations, it is more convenient to work with transposed permutons for strong IDU transformations, as defined below.

For a deterministic permuton~$P$, define its \emph{transposed} permuton~$P^{\transpose}$ by
\[
P^{\transpose}(A) = P(\{(y,x) : (x,y) \in A\}).
\]
We say that a deterministic permuton~$P$ is a \emph{$T$-permuton} if $P^{\transpose}$ defines a strong IDU transformation.  
Note that for every pattern~$\pi$,
\[
d(\pi, P^{\transpose}) = d(\pi^{-1}, P).
\]
Thus, results for $T$-permutons immediately translate into those for strong IDU transformations.

First, consider the unlabeled flag algebra for $T$-permutons.  
Since permutations with the same up--down signature have the same density in $T$-permutons, we identify such permutations in the flag algebra.  
Formally, let
\[
N = \{\pi_1 - \pi_2 \in \mathcal{A}^{\varnothing} : \udsign(\pi_1) = \udsign(\pi_2)\}.
\]

Now consider the typed flag algebra~$\mathcal{A}^\tau$ for a type~$\tau \in \Sk$.  
In addition to the up--down signature, we consider $\pi|_J$, the restriction of~$\pi$ to~$J$  
(this is not the pattern~$\patt(\pi, J)$ but simply the sequence of values of~$\pi$ at positions in~$J$).  
The elements of $\pi(J)$ partition the set~$[n]$ into $k+1$ intervals (some possibly empty), which we refer to as \emph{blocks}.  
It will be convenient to describe $\pi(J)$ in terms of the sizes of these blocks.

\begin{definition}[Partition vector]\label{def:part-vector}
The \emph{partition vector} $\partv(\pi, J)$ of a flag~$(\pi, J)$ is defined as
\[
 \partv(\pi, J) = (i_1 - 1,\, i_2 - i_1 - 1,\, \dots,\, n - i_k),
\]
where $i_1, \dots, i_k$ are the elements of~$\pi(J)$ in increasing order.
\end{definition}

Note that~$J$, $\partv(\pi, J)$, and~$\tau$ together uniquely determine~$\pi|_J$ (and vice versa).
Consider the product of two flags $(\pi_1, J_1)$ and $(\pi_2, J_2)$ of type~$\tau$.  
Let $c\,(\pi', J')$ be a nonzero term in their product $(\pi_1, J_1) \cdot (\pi_2, J_2)$ in~$\mathbb{R}\mathcal{F}^{\tau}$.  
Then it follows immediately that
\[
\partv(\pi', J') = \partv(\pi_1, J_1) + \partv(\pi_2, J_2).
\]
Finally, define
\[
N^{\tau} = \{ (\pi_1, J) - (\pi_2, J) \in \mathcal{A}^{\tau} :
\udsign(\pi_1) = \udsign(\pi_2)\ \text{and}\ 
\partv(\pi_1,J) = \partv(\pi_2,J) \}.
\]
Let~$\NIdeal$ be the linear span of~$N^{\tau}$.  
We write $a \equiv b \pmod{\NIdeal}$ if $a - b \in \NIdeal$.

\begin{theorem}\label{thm:main-flag-ideal-property} $\quad$\\
\noindent 1. For every $g \in \NIdeal$ and every tagged $T$-permuton $(P, S)$ of type~$\tau$,
$$
d(g, (P, S)) = 0$$ equivalently, for $g_1, g_2 \in \mathcal{A}^\tau$, if $g_1 \equiv g_2 \pmod{\NIdeal}$, then
$$d(g_1, (P, S)) = d(g_2, (P, S)).$$

\noindent 2. If $g\in \NIdeal$, then $\llbracket g \rrbracket = 0 \pmod{\mathbb{R}N}$.

\noindent 3. $\NIdeal$ is an ideal in~$\mathcal{A}^\tau$.
\end{theorem}

We defer the proof of Theorem~\ref{thm:main-flag-ideal-property} to Section~\ref{sec:proofs-flag-algebras}. Now, we discuss the implications of this theorem to the flag algebra framework for $T$-permutons. Recall that our general approach is to choose a type $\tau$ and $n$ flags $f_1, \dots, f_n \in \mathcal{F}^\tau$, construct the matrix $F$ with entries $F_{ij} = \llbracket f_i \cdot f_j \rrbracket$, and find a positive semidefinite matrix~$Q$ such that $Q \circ F$ yields the desired density inequality. 
The theorem shows that we can choose the flags $f_1, \dots, f_n$ from the quotient algebra $\mathcal{A}^\tau / \NIdeal$ or, in other words, choose at most one representative from each equivalence class modulo~$\NIdeal$.
Further, we will provide a simplified multiplication rule for typed density sums modulo~$\NIdeal$ in~Theorem~\ref{thm:flag-multiplication-simplified}.

\section{Analysis of ordering CSPs of arity 3 and 4}
In this section, we apply our framework to ordering CSPs defined by a single predicate of arity 3 or 4.
We will write most derivations modulo $N^{\varnothing}$. Thus, $(\mathsf{udu})$ will denote a flag algebra element $\pi \in \mathcal{A}^{\varnothing}$ with $\udsign(\pi) = udu$. Since all such elements are equivalent modulo $N^{\varnothing}$, the choice of a specific element will not matter. Further, we will write inequalities for pattern densities as 
$$(\mathsf{uu}) + (\mathsf{dd})-\frac{1}{3} \geq 0,$$
meaning that for every $T$-permuton $R$, 
$$d((\mathsf{uu}) + (\mathsf{dd}), R) \geq \frac{1}{3}.$$
Similarly, 
$$\min\big((\mathsf{uu}) + (\mathsf{dd}),\,2(\mathsf{ud})\big) \leq \nicefrac{1}{2}$$
is a shortcut for 
\[
\min\left(d((\mathsf{uu}) + (\mathsf{dd}), R),  d(2(\mathsf{ud}), R)\right) \leq \nicefrac{1}{2}
\]
for all $T$-permutons $R$.

\subsection{Ordering CSPs of arity 3}\label{sec:arity-3-analysis}
In this section, we analyze ordering CSPs of arity~3 using our approach.  
There are 11 non-isomorphic, nontrivial ordering predicates of arity~3, excluding the two trivial ones, ``always true'' and ``always false'' (see also Figure~\ref{fig:summary-3-4-results}).  
This includes the single nontrivial predicate $x_1 < x_2$ of arity~2 (viewed as a predicate of arity~3 by adding a dummy variable $x_3$); it appears in the list below as $\varphi_3$.
Completely satisfiable instances of four of them are polynomially tractable; here are the predicates described by their satisfying permutations:
\[
\begin{aligned}
\varphi_1&=\{(\mathsf{1\,2\,3})\},\\
\varphi_2&=\{(\mathsf{1\,2\,3}),\,(\mathsf{1\,3\,2})\},\\
\varphi_3&=\{(\mathsf{1\,2\,3}),\,(\mathsf{1\,3\,2}),\,(\mathsf{2\,3\,1})\},\\
\varphi_4&=\{(\mathsf{1\,2\,3}),\,(\mathsf{1\,3\,2}),\,(\mathsf{2\,1\,3}),\,(\mathsf{2\,3\,1})\}.
\end{aligned}
\]
Further, $\varphi_1$, $\varphi_2$, and $\varphi_3$ are precedence predicates (see Definition~\ref{def:precedence-CSP}); therefore, nearly satisfiable instances of these predicates admit an approximation ratio of $1 - O(\log n \log\log n)\,\varepsilon$, as discussed in Section~\ref{sec:nearly-satisfiable-csps}.

Among the remaining predicates, four have trivial $\varphi_{L}$ and $\varphi_{R}$ relaxations (the relaxations are equal to the ``always true'' predicate);
therefore, our framework cannot provide any nontrivial approximation guarantees for these four predicates.
All predicates have trivial $\varphi_{\varepsilon}$ relaxations.
However, there are three predicates with nontrivial $\varphi_{L}$ and $\varphi_{R}$ relaxations:
\[
\begin{aligned}
\varphi_5 &= \{(\mathsf{1\,2\,3}),\,(\mathsf{3\,2\,1})\},\\
\varphi_6 &= \{(\mathsf{1\,2\,3}),\,(\mathsf{2\,3\,1})\},\\
\varphi_7 &= \{(\mathsf{1\,2\,3}),\,(\mathsf{1\,3\,2}),\,(\mathsf{3\,1\,2})\}.
\end{aligned}
\]
Note that predicate $\varphi_5$ is \emph{Betweenness}, studied by Chor and Sudan~\cite{CS98}.
We list the $L$- and $R$-relaxations of $\varphi_5$, $\varphi_6$, and $\varphi_7$ below. Note that Betweenness is self-dual and its $L$- and $R$-relaxations are dual to each other. Therefore, both of these relaxations yield the same approximation factor for Betweenness.
\[
\begin{aligned}
\varphi_{5L} &= \{(\mathsf{1\,2\,3}),\,(\mathsf{1\,3\,2}),\,(\mathsf{2\,3\,1}),\,(\mathsf{3\,2\,1})\},\\
\varphi_{5R} &= \{(\mathsf{1\,2\,3}),\,(\mathsf{2\,1\,3}),\,(\mathsf{3\,1\,2}),\,(\mathsf{3\,2\,1})\},\\
\varphi_{6L} &= \varphi_{6R} = \{(\mathsf{1\,2\,3}),\,(\mathsf{1\,3\,2}),\,(\mathsf{2\,3\,1})\},\\
\varphi_{7L} &= \{(\mathsf{1\,2\,3}),\,(\mathsf{1\,3\,2}),\,(\mathsf{2\,1\,3}),\,(\mathsf{3\,1\,2})\}.
\end{aligned}
\]

For each of these relaxations $\varphi'$ (except for $\varphi_{5R}$, which is equivalent to $\varphi_{5L}$), we write the expression for $p(\varphi' \to \varphi, \cdot)$ modulo $\mathbb{R}N$ (that is, we only keep track of the up--down signatures).
\begin{align*}
\varphi_5: &\min\big((\mathsf{uu}) + (\mathsf{dd}),\,(\mathsf{uu}) + (\mathsf{dd}),\,2(\mathsf{ud}),\,2(\mathsf{ud})\big)
   = \min\big((\mathsf{uu}) + (\mathsf{dd}),\,2(\mathsf{ud})\big),\\[4pt]
\varphi_6: &\min\big((\mathsf{ud}) + (\mathsf{du}),\,(\mathsf{uu}) + (\mathsf{du}),\,(\mathsf{uu}) + (\mathsf{ud})\big),\\[4pt]
\varphi_7: &\min\big((\mathsf{uu}) + (\mathsf{du}) + (\mathsf{dd}),\,(\mathsf{uu}) + (\mathsf{ud}) + (\mathsf{dd}),\,2(\mathsf{ud}) + (\mathsf{du}),\,(\mathsf{uu}) + 2(\mathsf{ud})\big).
\end{align*}

We analyze these one by one. For Betweenness, we need to find a $T$-permuton~$P$ that maximizes $\min\big((\mathsf{uu}) + (\mathsf{dd}),\,2(\mathsf{ud})\big)$. We use a $T$-permuton (see Figure~\ref{fig:permuton-compositions})
\[
P = \big(\tfrac{1}{2} I \oUp \tfrac{1}{2} D\big)^\transpose.
\]
(In retrospect, this is equivalent to the transformation used in~\cite{Mak}.)
Using Theorem~\ref{thm:main-qsym-characterization}, it can be immediately verified that
\[
d((\mathsf{uu}) + (\mathsf{dd}), P) =  d(2(\mathsf{ud}), P) = \nicefrac{1}{2}.
\]
Now let us prove an upper bound on the approximability of $\varphi_5$ in our framework. Note that
\begin{equation}
d\!\left(\sum_{\pi\in \mathbb{S}_3} \pi, P\right) = 1.
\label{eq:all-S3-add-to-1}
\end{equation}
Since there is one permutation in $\mathbb{S}_3$ with signature $uu$ (the identity permutation),
one with signature $dd$ (the decreasing permutation), and two each with signatures $ud$ and $du$, we get
\[
(\mathsf{uu}) + (\mathsf{dd}) + 2(\mathsf{ud}) + 2(\mathsf{du}) = 1.
\]
Since all densities are nonnegative, we have
\[
(\mathsf{uu}) + (\mathsf{dd}) + 2(\mathsf{ud}) \le 1.
\]
Therefore,
\[
\min\big((\mathsf{uu}) + (\mathsf{dd}),\,2(\mathsf{ud})\big) \leq \nicefrac{1}{2}.
\]
Thus, the best approximation factor for Betweenness achievable using our approach is $\nicefrac{1}{2}$.  
We can also consider a slightly more general ordering CSP with the constraint language consisting of Betweenness and the less-than predicate $x_1 < x_2$. As the less-than predicate equals its own $L$- and $R$-relaxations, the probability $p(`<\textrm{'} \to `<\textrm{'}, P)$ (where `$<$' is the less-than predicate) is equal to the density of $(\mathsf{1\,2})$ in $P$. Now, $d((\mathsf{1\,2}), P) = \nicefrac{1}{2}$. That is, we can get a $\nicefrac{1}{2}$-approximation for this more general ordering CSP as well.

Now consider the predicate $\varphi_6$. Before we show that it does not admit a nontrivial approximation using our framework, we need to prove that $(\mathsf{uu}) + (\mathsf{dd}) \ge \nicefrac{1}{3}$ (this lower bound is achieved by the uniform permuton). We perform computations in $\mathcal{A}^{(\mathsf{1})}/\mathbb{R}N^{(\mathsf{1})}$. A dot in a signature identifies the position of a labeled index. Here, all flags written as $(\mathsf{xx.xx})$ have partition vector $(2,2)$. For example, $(\mathsf{dd.dd})$ denotes the equivalence class of $(\mathsf{1\ 2\ \underline{3}\ 4\ 5})$; $(\mathsf{ud.du})$ denotes the equivalence class of $(\mathsf{1\ 4\ \underline{3}\ 2\ 5})$.

\begin{align*}
\left\llbracket\left((\mathsf{u.u}) - (\mathsf{d.d})\right)^2\right\rrbracket &= 
\frac{2}{3} \bigl\llbracket(\mathsf{dd.dd}) + (\mathsf{dd.du}) - 2(\mathsf{du.du}) - (\mathsf{du.ud}) \\
&\qquad + (\mathsf{du.uu}) + (\mathsf{ud.dd}) - (\mathsf{ud.du}) - 2(\mathsf{ud.ud}) + (\mathsf{uu.ud}) + (\mathsf{uu.uu})\bigr\rrbracket\\
&= \frac{2}{15} \bigl((\mathsf{dddd}) + (\mathsf{dddu}) - 2(\mathsf{dudu}) - (\mathsf{duud}) \\
&\qquad + (\mathsf{duuu}) + (\mathsf{uddd}) - (\mathsf{uddu}) - 2(\mathsf{udud}) + (\mathsf{uuud}) + (\mathsf{uuuu})\bigr)\\
& = \frac{1}{5}\Bigl((\mathsf{uu}) + (\mathsf{dd})-\frac{1}{3}\Bigr). 
\end{align*}
We conclude that $(\mathsf{uu}) + (\mathsf{dd})-\frac{1}{3} \geq 0$.
We have,
$$(\mathsf{ud}) + (\mathsf{du}) = \frac{1 -((\mathsf{uu}) + (\mathsf{dd}))}{2} \leq \frac{1-1/3}{2} = \frac{1}{3}.$$
We conclude that $\varphi_6$ does not admit an approximation better than $1/3$, which is what random ordering gives.

\begin{figure}
  \centering
\begin{tikzpicture}[x=3.5cm,y=3.5cm,>=stealth]
\draw[->] (-0.1,0) -- (1.1,0) node[below right=2pt] {$(\mathsf{uu})$};
\draw[->] (0,-0.1) -- (0,1.1) node[above left=2pt] {$(\mathsf{dd})$};

  \foreach \k in {1,...,10}{
    \pgfmathsetmacro{\x}{1/\k}
    \pgfmathsetmacro{\y}{1/3 + 2/3*(1/(\k*\k))}
    \coordinate (A\k) at ({0.5*(\x+\y)}, {0.5*(\y-\x)});
    \coordinate (B\k) at ({0.5*(-\x+\y)}, {0.5*(\y+\x)});
  }
  \draw[gray, domain=-1:1, samples=200]
    plot ({0.5*(\x + (1/3 + 2/3*(\x*\x)))},
          {0.5*((1/3 + 2/3*(\x*\x)) - \x)});

  \path[fill=blue!15,draw=blue,line width=0.5pt]
    (A1) -- (A2) -- (A3) -- (A4) -- (A5) -- (A6) -- (A7) -- (A8) -- (A9) -- (A10)
    -- (0.5*1/3,0.5*1/3) 
    -- (B10) -- (B9) -- (B8) -- (B7) -- (B6) -- (B5) -- (B4) -- (B3) -- (B2) -- (B1)
    -- cycle;

  \foreach \k in {1,...,5}{
    \fill[black] (A\k) circle (0.008);
    \fill[black] (B\k) circle (0.008);
  }
  \draw[dashed] (1/3,0) -- (0,1/3);
  \fill[black] (1/6,1/6) circle (0.01);
  \node[above left]  at (0,1) {\small 1};
  \node[below right] at (1,0) {\small 1};
  \node[below left] at (0,0) {\small 1};
  \node[below] at (1/3,0) {\small $\nicefrac{1}{3}$};
  \node[left] at (0,1/3) {\small $\nicefrac{1}{3}$};

\end{tikzpicture}

\caption{The plot depicts the convex hull of achievable densities $(\mathsf{uu})$ and $(\mathsf{dd})$ in $T$-permutons. The boundary forms a polygonal curve composed of infinitely many segments. All feasible points lie above and to the right of the dashed line $(\mathsf{uu}) + (\mathsf{dd}) = \nicefrac{1}{3}$. The point $(\nicefrac{1}{6}, \nicefrac{1}{6})$ represents the uniform permuton. The lower envelope is approximated by the parabola $(\mathsf{uu}) + (\mathsf{dd}) = \nicefrac{1}{3} + \nicefrac{2}{3}((\mathsf{uu}) - (\mathsf{dd}))^2$ (shown in gray).}
\end{figure}
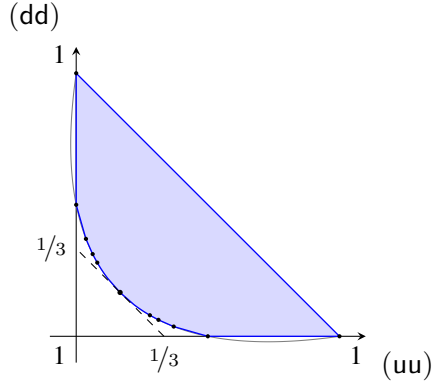
Finally, we consider $\varphi_7$. We prove that it also does not admit any nontrivial approximation using our framework. We have,
$$\left((\mathsf{uu}) + (\mathsf{du}) + (\mathsf{dd})\right) + \left(2(\mathsf{ud}) + (\mathsf{du})\right) = (\mathsf{uu}) + 2(\mathsf{ud}) + 2(\mathsf{du}) + (\mathsf{dd}) = 1.$$
Therefore, either $(\mathsf{uu}) + (\mathsf{du}) + (\mathsf{dd}) \leq \nicefrac{1}{2}$ or $2(\mathsf{ud}) + (\mathsf{du}) \leq \nicefrac{1}{2}$.
We conclude that $\varphi_7$ does not admit an approximation better than the trivial one of $1/2$ using our framework.

\subsection{Two examples of ordering CSPs of arity 4}\label{sec:arity-4-analysis}
In this section, we provide two concrete examples of predicates of arity~4 that admit a nontrivial approximation using our framework. Among numerous examples, we chose two for which the computations are quite simple. In particular, in these examples we use IDU transformations defined by deterministic permutons, while many others require random, non-deterministic ones.

Consider the following ordering CSP: 
\[\varphi(x_1,x_2, x_3, x_4) = (x_1 < x_2 < x_3 < x_4) \vee (x_4 < x_1 < x_2 < x_3)\]
with $\Sat(\varphi) = \{(\mathsf{1\,2\,3\,4}),(\mathsf{2\,3\,4\,1})\}$. It is easy to verify that the $L$, $R$, and $\varepsilon$-relaxations of $\varphi$ are equal:
$$\varphi_{L} = \varphi_{R} = \varphi_{\varepsilon} = (x_1 < x_2 < x_3),$$
and $\Sat(\varphi_{\varepsilon}) = \{(\mathsf{1\,2\,3\,4}),(\mathsf{1\,2\,4\,3}),(\mathsf{1\,3\,4\,2}),(\mathsf{2\,3\,4\,1})\}$.
We compute $p(\varphi_{\varepsilon} \to \varphi, R)$ for a strong random permuton $R$:
\begin{align*}
p(\varphi_{\varepsilon} \to \varphi, R) &= \min(\\
&\qquad d((\mathsf{1\,2\,3\,4})\cdot (\mathsf{1\,2\,3\,4})^{-1} + (\mathsf{2\,3\,4\,1}) \cdot (\mathsf{1\,2\,3\,4})^{-1},  R),\\
&\qquad d((\mathsf{1\,2\,3\,4})\cdot (\mathsf{1\,2\,4\,3})^{-1} + (\mathsf{2\,3\,4\,1}) \cdot (\mathsf{1\,2\,4\,3})^{-1},  R),\\
&\qquad d((\mathsf{1\,2\,3\,4})\cdot (\mathsf{1\,3\,4\,2})^{-1} + (\mathsf{2\,3\,4\,1}) \cdot (\mathsf{1\,3\,4\,2})^{-1},  R),\\
&\qquad d((\mathsf{1\,2\,3\,4})\cdot (\mathsf{2\,3\,4\,1})^{-1} + (\mathsf{2\,3\,4\,1}) \cdot (\mathsf{2\,3\,4\,1})^{-1},  R))\\
&= \min(\\
&\qquad d((\mathsf{1\,2\,3\,4}) + (\mathsf{2\,3\,4\,1}),R),\\
&\qquad d((\mathsf{1\,2\,4\,3}) + (\mathsf{2\,3\,1\,4}),R),\\
&\qquad d((\mathsf{1\,4\,2\,3}) + (\mathsf{2\,1\,3\,4}),R),\\
&\qquad d((\mathsf{4\,1\,2\,3}) + (\mathsf{1\,2\,3\,4}),R))\\
&=\min(\\
&\qquad d((\mathsf{1\,2\,3\,4}) + (\mathsf{4\,1\,2\,3}),R^\transpose),\\
&\qquad d((\mathsf{1\,2\,4\,3}) + (\mathsf{3\,1\,2\,4}),R^\transpose),\\
&\qquad d((\mathsf{1\,3\,4\,2}) + (\mathsf{2\,1\,3\,4}),R^\transpose),\\
&\qquad d((\mathsf{2\,3\,4\,1}) + (\mathsf{1\,2\,3\,4}),R^\transpose))
).
\end{align*}
Write up--down signatures of these permutations:
\begin{align*}
p(\varphi_{\varepsilon} \to \varphi, R)&=
\min(
 d((\mathsf{uuu}) + (\mathsf{duu}),R^\transpose),
 d((\mathsf{uud}) + (\mathsf{duu}),R^\transpose),\\
&\phantom{{}=\min(}
 d((\mathsf{uud}) + (\mathsf{duu}),R^\transpose),
 d((\mathsf{uud}) + (\mathsf{uuu}),R^\transpose))
\\
&=
\min(
 d((\mathsf{uuu}) + (\mathsf{duu}),R^\transpose),
  d((\mathsf{uuu}) + (\mathsf{uud}),R^\transpose),
 d((\mathsf{uud}) + (\mathsf{duu}),R^\transpose)).
\end{align*}

As a warm-up, we note that Theorem~\ref{thm:simple-sufficient-admits-approx} implies that  
\( p(\varphi_{\varepsilon} \to \varphi, R) > \nicefrac{1}{12} \) (the trivial approximation factor) for some IDU permuton $R$. From the signatures above, we deduce that the vectors \( v(\varphi, \tau) \), as in Theorem~\ref{thm:simple-sufficient-admits-approx},  are
\[
(0,2,2) = (1,1,1) + (-1,1,1), \quad
(2,2,0) = (1,1,1) + (1,1,-1), \quad
(0,2,0) = (1,1,-1) + (-1,1,1).
\]
We see that these vectors have strictly positive inner products with \( y = (1,1,1) \), and thus Theorem~\ref{thm:simple-sufficient-admits-approx} applies. However, the theorem does not yield an explicit approximation factor.

Now let $R$ be the deterministic permuton $y D \oUp x I \oUp yD$, where $x+2y = 1$ and $x,y\geq 0$.
We compute the densities of patterns with the following signatures w.r.t.\ $R^\transpose$:
\begin{align*}
(\mathsf{uuu}):&\quad x^4 + 2x^3y + x^2y^2\\
(\mathsf{uud}), (\mathsf{duu}):&\quad x^3y + 2x^2y^2 + xy^3
\end{align*}
Plugging these expressions into the formula for $p(\varphi_{\varepsilon}\to \varphi, R)$ and simplifying the obtained algebraic expressions, we get
$$p(\varphi_{\varepsilon}\to \varphi, R)
= \min\left(\frac{x(1+x)^3}{8},\frac{x+x^2-x^3-x^4}{4}\right).
$$
We now find $x$ that maximizes $f(x) = \frac{x+x^2-x^3-x^4}{4}$. Write 
$$f'(x)= \frac{1 +2x -3x^2  - 4x^3}{4} = \frac{(1+x)(1 +x -4x^2)}{4}.$$
Polynomial $1 +x -4x^2$ has roots $\frac{1\pm \sqrt{17}}{8}$. Thus, the only root of $f'(x)$ on $(0,1)$ is $\frac{1 + \sqrt{17}}{8}$ (the other two roots, $-1$ and $(1-\sqrt{17})/8$ are negative). For this value of $x$, we get  
\begin{align*}
\min\left(\frac{x(1+x)^3}{8},\frac{x+x^2-x^3-x^4}{4}\right) &= 
\min\left(\frac{701+181\sqrt{17}}{4096}, \frac{107 + 51\sqrt{17}}{2048}\right)\\
&= \frac{107 + 51\sqrt{17}}{2048}=0.1549\dots
\end{align*}
\begin{corollary}
There is a polynomial-time approximation algorithm for predicate  $ \varphi(x_1,x_2, x_3, x_4) = (x_1 < x_2 < x_3 < x_4) \vee (x_4 < x_1 < x_2 < x_3)$ that, given a $1 - \varepsilon$ satisfiable instance (with $\varepsilon \geq 0$), returns a solution that satisfies at least an $\alpha - O(\varepsilon\log n \log\log n)$ fraction of the constraints, where
$$\alpha =   \frac{107 + 51\sqrt{17}}{2048}=0.1549\dots$$
\end{corollary}
Note that the random ordering algorithm satisfies a
$$\frac{|\Sat(\varphi)|}{4!} = \frac{2}{24} = \frac{1}{12} = 0.08(3)$$
fraction of the constraints.

Now we briefly describe another example of an ordering CSP of arity 4.
Consider
$$\varphi(x_1, x_2, x_3, x_4) = (x_1 < x_2 < x_3 < x_4) \vee (x_1 < x_4 < x_3 < x_2).$$
For this predicate, the three relaxations $\varphi_{\varepsilon}$, $\varphi_{L}$, and $\varphi_{R}$ are distinct:
\begin{align*}
\varphi_{\varepsilon}(x_1, x_2, x_3, x_4) &= (x_1 < x_2) \wedge (x_1 < x_3) \wedge (x_1 < x_4)\\
\varphi_{L}(x_1, x_2, x_3, x_4) &= \varphi_{\varepsilon}(x_1, x_2, x_3, x_4) \wedge \NFirst(x_3, x_2,x_4)\\
\varphi_{R}(x_1, x_2, x_3, x_4) &= \varphi_{\varepsilon}(x_1, x_2, x_3, x_4) \wedge \NLast(x_3, x_2,x_4)
\end{align*}
We get the following formulas for $p(\cdot \to \varphi, R)$. Recall that although we write $p(\varphi' \to \varphi, R)$ for a strong permuton $R$, when we express this quantity using up--down signatures 
we write the corresponding densities with respect to $R^\transpose$.
\begin{align*}
p(\varphi_{\varepsilon} \to \varphi, R) &= \min(d(2(\mathsf{uud}), R^\transpose), d(2(\mathsf{udu}), R^\transpose), d((\mathsf{uuu}) + (\mathsf{udd}), R^\transpose)) \\
p(\varphi_{L} \to \varphi, R) &= \min(d(2(\mathsf{uud}), R^\transpose), d((\mathsf{uuu}) + (\mathsf{udd}), R^\transpose)) \\
p(\varphi_{R} \to \varphi, R) &= \min(d(2(\mathsf{udu}), R^\transpose), d((\mathsf{uuu}) + (\mathsf{udd}), R^\transpose))
\end{align*}
We now choose permutons $R$ to maximize 
$p(\varphi_{\varepsilon} \to \varphi, R)$ and $p(\varphi_{L} \to \varphi, R)$;
we do not choose $R$ for $\varphi_{R}$ as it yields a worse approximation than that 
for $\varphi_{L}$.

First consider $\varphi_{L}$. 
As in the previous example, Theorem~\ref{thm:simple-sufficient-admits-approx} implies that 
\( p(\varphi_{L} \to \varphi, R) > \nicefrac{1}{12} \) for some strong IDU permuton, since in this case we have the following vectors \( v(\varphi, \tau) \):
\[
(2,2,-2) \quad \text{and} \quad (2,0,0) = (1,1,1) + (1,-1,-1)
\]
and both vectors have strictly positive inner products with \( y = (1,0,0) \).

Let $R = x I \oUp y D$ where $x$ and $y$ are nonnegative and add up to 1.
We get 
\begin{align*}
d((\mathsf{uuu}), R^\transpose) &= x^4 + x^3 y \\
d((\mathsf{uud}), R^\transpose) &= x^3 y + x^2 y^2 \\
d((\mathsf{udd}), R^\transpose) &= x^2 y^2 + x y^3\\
p(\varphi_{L} \to \varphi, R) &= \min(2(x^3 y + x^2 y^2), (x^4 + x^3 y) + (x^2 y^2 + x y^3)).
\end{align*}
Maximizing this expression over all choices of $x,y \geq 0$ with $x+y = 1$, 
we get that the maximum is $8/27$ when $x=2/3$ and $y=1/3$.

Now consider $\varphi_{\varepsilon}$. It is easy to see that Theorem~\ref{thm:simple-sufficient-admits-approx} again guarantees that \( p(\varphi_{\varepsilon} \to \varphi, R) \) is nontrivial for some \( R \), but we omit the straightforward computation. Let $R = x I \oUp y I \oUp z I$ where $x+y+z = 1$ and all of them are nonnegative.
We have,
\begin{align*}
d((\mathsf{uuu}), R^\transpose) &= x^{4} + x^{3} y + x^{3} z + x^{2} y^{2} + x^{2} y z \\
&\phantom{{}={}} + x^{2} z^{2} + x y^{3} + x y^{2} z + x y z^{2} + x z^{3} + y^{4} + y^{3} z + y^{2} z^{2} + y z^{3} + z^{4}\\
d((\mathsf{uud}), R^\transpose) &= x^{3} y + x^{3} z + x^{2} y z + x y^{2} z + y^{3} z \\
d((\mathsf{udu}), R^\transpose) &= x^{2} y^{2} + x^{2} y z + x^{2} z^{2} + x y z^{2} + y^{2} z^{2} \\
d((\mathsf{udd}), R^\transpose) &= x^{2} y z. 
\end{align*}
We plug these expressions into the formula for $p(\varphi_{\varepsilon} \to \varphi, R)$ and solve for the optimal choice of $(x,y,z)$ numerically
(as it appears there is no closed-form formula for them): 
$$(x,y,z) \approx (0.41292217261909, 0.19926838465167, 0.38780944272924),$$
and $p(\varphi_{\varepsilon} \to \varphi, R) \geq 0.1278735827$.

\begin{corollary}
There are polynomial-time approximation algorithms for satisfiable and nearly satisfiable instances of $\CSP(\varphi)$ for the predicate  
$\varphi(x_1, x_2, x_3, x_4) = (x_1 < x_2 < x_3 < x_4) \vee (x_1 < x_4 < x_3 < x_2)$.
Given a satisfiable instance, the first algorithm satisfies at least a fraction $\alpha = 8/27 = 0.2\overline{962}$; given a $(1-\varepsilon)$-satisfiable instance the second algorithm finds a solution satisfying at least a $\beta - O(\varepsilon\log n\log\log n)$ fraction of the constraints, where $\beta = 0.1278735827$. In contrast, the trivial approximation algorithm gives a $1/12 = 0.08\overline{3}$ approximation.
The first approximation algorithm beats the trivial approximation by a factor $3\,\nicefrac{5}{9} = 3.\overline{5}$.
\end{corollary}


\section{Proofs of the main flag algebra results for strong IDU transformations}
\label{sec:proofs-flag-algebras}

\subsection{Multiplication of flags for \texorpdfstring{$T$}{T}-permutons}
We now give an alternative (but equivalent) definition of multiplication of standard permutation flags that is arguably more constructive than Definition~\ref{def:flag-multiplication}.
Although simple, we have not seen this formulation in the literature before.

\begin{definition}
Consider two flags $(\pi_1, J_1)$ and $(\pi_2, J_2)$ of type~$\tau$.
Assume that $\pi_1 \in \mathbb{S}_{n_1}$ and $\pi_2 \in \mathbb{S}_{n_2}$.
We say that a pair $(X, Y)$ is a \emph{$y$-shuffle} of $(\pi_1, J_1)$ and $(\pi_2, J_2)$ if:
\begin{itemize}
  \item $X$ is a sequence of distinct elements from $[n_1 + n_2 - k]$ of length~$n_1$;
  \item $Y$ is a sequence of distinct elements from $[n_1 + n_2 - k]$ of length~$n_2$;
  \item $X$ is order-isomorphic to~$\pi_1$, and $Y$ is order-isomorphic to~$\pi_2$;
  \item $X$ restricted to indices in~$J_1$ is equal to $Y$ restricted to indices in~$J_2$ (as sequences);
  \item $X_i \neq Y_j$ unless $i \in J_1$ and $j \in J_2$.
\end{itemize}
\end{definition}

\begin{example}\label{example:yshuffle}
Consider two flags $(\mathsf{1\,3\,\underline{2}\,5\,\underline{4}})$ and $(\mathsf{1\,\underline{2}\,\underline{3}})$ in $\mathcal{F}^{(\mathsf{1 \ 2})}$. 
They have two $y$-shuffles:
$$\bigl((\mathsf{1\,4\,\underline{3}\,6\,\underline{5}}) , (\mathsf{2\,\underline{3}\,\underline{5}})\bigr)\quad \text{and}\quad
\bigl((\mathsf{2\,4\,\underline{3}\,6\,\underline{5}}) , (\mathsf{1\,\underline{3}\,\underline{5}})\bigr).$$
\end{example}
\begin{definition}\label{def:x-shuffle}
Consider a pair of sequences $(X, Y)$. Assume that all elements in $X$ are distinct and all elements in $Y$ are distinct. Further, assume that the common elements of $X$ and $Y$ appear in the same order. 
Let $J_1$ be the indices of the common elements in~$X$, and let $J_2$ be the indices of the common elements in~$Y$.

Consider a sequence $Z$ formed from the elements of~$X$ and~$Y$, in which the common elements -- those with indices in~$J_1$ in~$X$ and in~$J_2$ in~$Y$ -- appear only once, while the relative orders of elements from both~$X$ and~$Y$ are preserved.  
Let~$J$ denote the set of indices of the common elements in~$Z$.  
We call $(Z, J)$ an \emph{$x$-shuffle} of $(X, Y)$.

If $(X, Y)$ is a $y$-shuffle, then the sequence $Z$ defined in this way is a permutation, and $(Z, J)$ is a flag.
\end{definition}

\begin{example}
Consider a $y$-shuffle from Example~\ref{example:yshuffle},
$\bigl((\mathsf{2\,4\,\underline{3}\,6\,\underline{5}}), (\mathsf{1\,\underline{3}\,\underline{5}})\bigr)$.
It has the following $x$-shuffles:
$(\mathsf{1\,2\,4\,\underline{3}\,6\,\underline{5}})$,
$(\mathsf{2\,1\,4\,\underline{3}\,6\,\underline{5}})$, and
$(\mathsf{2\,4\,1\,\underline{3}\,6\,\underline{5}})$.
\end{example}

\begin{definition}
Let $(\pi_1, J_1)$ and $(\pi_2, J_2)$ be two flags of type~$\tau$.
Define their \emph{simple product} as
\[
(\pi_1, J_1) \star (\pi_2, J_2)
  = \sum_{(X,Y)} \sum_{(Z,J)} (Z,J),
\]
where the first summation is over all $y$-shuffles $(X,Y)$ of $(\pi_1, J_1)$ and $(\pi_2, J_2)$, and the second summation is over all $x$-shuffles $(Z,J)$ of $(X, Y)$.
\end{definition}

\begin{observation}\label{obs:multiplication}
For two flags $(\pi_1, J_1)$ and $(\pi_2, J_2)$ of type~$\tau \in \Sk$, with $\pi_1 \in \mathbb{S}_{n_1}$ and $\pi_2 \in \mathbb{S}_{n_2}$, we have
\[
(\pi_1, J_1) \cdot (\pi_2, J_2)
= \frac{1}{\binom{n_1+n_2 - 2k}{\,n_1 - k\,}}\,(\pi_1, J_1) \star (\pi_2, J_2).
\]
\end{observation}

\begin{claim}
Consider two flags $(\pi_1, J_1),(\pi_2, J_2) \in {\cal F}^{\tau}$. Let $(X,Y)$ and $(X',Y')$ be $y$-shuffles of these flags.
Then
\[
\sum_{Z} (Z, J) \equiv \sum_{Z'} (Z', J) \pmod{\NIdeal},
\]
where the first summation is over all $x$-shuffles $(Z, J)$ of $(X, J_1)$ and $(Y, J_2)$,
and the second summation is over all $x$-shuffles $(Z',J)$ of $(X', J_1)$ and $(Y', J_2)$.
\end{claim}

\begin{proof}
Consider two flags $f_1 = (\pi_1, J_1)$ and $f_2 = (\pi_2, J_2)$. We define a swap operation on the set of $y$-shuffles of $f_1$ and $f_2$.
Given a $y$-shuffle $(X,Y)$ and two elements $X_i$ and $Y_j$ that differ by~1 and are both unlabeled,
$$|X_i - Y_j| = 1\qquad \text{and}\ i \notin J_1,\ j \notin J_2,$$
we define another $y$-shuffle $(X', Y')$ by setting $X'_i = Y_j$ and $Y'_j = X_i$, while keeping all other values unchanged; that is, $X'_t = X_t$ for $t\neq i$ and $Y'_t = Y_t$ for $t\neq j$.

Note that every $y$-shuffle of $f_1$ and $f_2$ can be transformed into any other $y$-shuffle of $f_1$ and $f_2$ by a sequence of swaps as defined above (essentially by applying a bubble-sort procedure).

Therefore, it suffices to prove the statement for $(X,Y)$ and $(X', Y')$ obtained from $(X, Y)$ by such a swap: $X'_i = Y_j$ and $Y'_j = X_i$.
We construct a matching between $x$-shuffles of $(X,Y)$ and of $(X', Y')$ so that matched $x$-shuffles are equivalent modulo $\NIdeal$. In fact, for every $x$-shuffle $(Z,J)$ of $(X, Y)$ and every $x$-shuffle $(Z',J)$ of $(X', Y')$, we have
$$\partv(Z, J) = \partv(Z', J) = \partv(\pi_1, J_1) + \partv(\pi_2, J_2).$$
Thus, it suffices to ensure that the up--down signatures of matched $x$-shuffles $(Z,J)$ and $(Z',J)$ are equal.   

First, we match $x$-shuffles $(Z, J)$ in which the elements $X_i$ and $Y_j$ do not appear next to each other (for the specific $i$ and $j$ as above) with $x$-shuffles $(Z', J)$ in which the elements $X_i'$ and $Y_j'$ also do not appear next to each other. The matching is straightforward -- we consider $x$-shuffles $(Z, J)$ and $(Z', J)$ formed using the same interleaving of $X$ with $Y$ and of $X'$ with $Y'$; that is, $Z_s$ comes from $X$ if and only if $Z'_s$ comes from $X'$, and $Z_s$ comes from $Y$ if and only if $Z'_s$ comes from $Y'$.
We now show that $Z$ and $Z'$ have equal up--down signatures. Consider two adjacent elements $Z_s$ and $Z_{s+1}$. If neither $Z_s$ nor $Z_{s+1}$ equals $X_i$ or $Y_j$, then $Z_s = Z'_s$ and $Z_{s+1} = Z'_{s+1}$, and accordingly $Z$ and $Z'$ have the same letter at position $s$ in their up--down signatures. Now consider the case where one of $Z_s$ or $Z_{s+1}$ equals $X_i$ or $Y_j$
(note that it is impossible that one equals $X_i$ and the other $Y_j$, as we have assumed that $X_i$ and $Y_j$ are not adjacent in $Z$).

Let $Z_t$ be an element adjacent to $X_i$ in $Z$.  
Then $X_i > Z_t$ if and only if $X_i' > Z_t$, since
$|X_i - X_i'| = 1$ and thus $Z_t = Z_t'$ cannot lie between $X_i$ and $X_i'$.  
The same reasoning applies to $Y_j$ and $Y_j'$.  
We conclude that $Z$ and $Z'$ have the same up--down signature.  

Now consider the case where $X_i$ and $Y_j$ appear next to each other in $Z$.  
We match each $x$-shuffle of $(X, Y)$ in which $X_i$ is immediately followed by $Y_j$ with the $x$-shuffle of $(X', Y')$ in which $Y'_j$ is immediately followed by $X'_i$, with all other elements occupying the same positions. These two shuffles are identical and thus have the same up--down signatures.  
Similarly, we match each $x$-shuffle of $(X, Y)$ in which $Y_j$ is immediately followed by $X_i$ with the $x$-shuffle of $(X', Y')$ in which $X'_i$ is immediately followed by $Y'_j$, again with all other elements in the same positions. Again, these two $x$-shuffles are identical and thus have the same up--down signatures.
\end{proof}

We will also need to apply this claim in cases where $(X, Y)$ and $(X', Y')$ are not necessarily $y$-shuffles.  
Assume that in both pairs the set of common elements is the same set $J$, and that the patterns of $X$ and $X'$ are equal, as are the patterns of $Y$ and $Y'$.  
By removing all gaps in $X \cup Y$ and in $X'\cup Y'$ (that is, replacing each $X_i$ and $Y_j$ with its rank in $X\cup Y$, and each $X'_i$ and $Y'_j$ with its rank in $X'\cup Y'$), we reduce the general case of arbitrary $(X, Y)$ and $(X', Y')$ to the case of $y$-shuffles.  
This leads to the following corollary.

\begin{corollary}\label{cor:general-shuffles}
Consider two pairs of sequences $(X, Y)$ and $(X', Y')$ as in Definition~\ref{def:x-shuffle}. Assume that  
the patterns of $X$ and $X'$ are equal, the patterns of $Y$ and $Y'$ are equal,  
and the set of common elements in both pairs is the same $J$.  
Then the multisets of up--down signatures of all $x$-shuffles of $(X, Y)$ and of $(X', Y')$  
are identical, including multiplicities.
\end{corollary}

\begin{theorem}\label{thm:flag-multiplication-simplified}
For any two flags $(\pi_1, J_1), (\pi_2, J_2) \in {\cal F}^{\tau}$
and any $y$-shuffle $(X, Y)$ of $(\pi_1, J_1)$ and $(\pi_2, J_2)$, we have
\[
(\pi_1, J_1) \cdot (\pi_2, J_2) \equiv
c \cdot
\sum_{(Z, J)} (Z, J) \pmod{\NIdeal}
\]
and
\[
\left\llbracket (\pi_1, J_1) \cdot(\pi_2, J_2) \right\rrbracket =
\left\llbracket
c \cdot
\sum_{(Z, J)} (Z, J)
\right\rrbracket,
\]
where
\[
c = \frac{M(\partv(\pi_1, J_1))\, M(\partv(\pi_2, J_2))}{M(\partv(\pi_1, J_1) + \partv(\pi_2, J_2))},
\]
and the summation is taken over all $x$-shuffles $(Z, J)$ of $(X, J_1)$ and $(Y, J_2)$.
\end{theorem}

\begin{proof}
We use the alternative definition of flag multiplication (see Observation~\ref{obs:multiplication}). As we showed, the $x$-shuffle of a $y$-shuffle of a pair of flags does not depend on the specific choice of the $y$-shuffle modulo $\NIdeal$. Therefore, we may use the given shuffle $(X,Y)$. It remains to count the number of $y$-shuffles. Inside each block~$i$ (here we refer to the blocks from Definition~\ref{def:part-vector}), we can choose positions for elements from $X$ and for elements from $Y$ in 
$$\binom{\partv(\pi_1, J_1)_i + \partv(\pi_2, J_2)_i}{\partv(\pi_1, J_1)_i}$$
ways. The total number of $y$-shuffles is the product of these binomial coefficients over all blocks~$i$:
\[
\prod_{i=1}^{k+1} \binom{\partv(\pi_1, J_1)_i + \partv(\pi_2, J_2)_i}{\partv(\pi_1, J_1)_i}.
\]
After simplification, we get the desired value of~$c$.
\end{proof}
\subsection{Proof of Theorem~\ref{thm:main-flag-ideal-property}}

Let us now define another swap operation that we will use in the proof of Theorem~\ref{thm:main-flag-ideal-property}.
\begin{definition}
Consider a flag $(\pi, J) \in \mathcal{F}^\tau$. Choose non-adjacent $i_1, i_2 \notin J$,
and swap values at positions $i_1$ and $i_2$: $\pi'(i_1) = \pi(i_2)$, $\pi'(i_2) = \pi(i_1)$, and $\pi'(i) = \pi(i)$ for all $i \notin \{i_1, i_2\}$.
We refer to this swap as an \emph{$(i_1, i_2)$-swap}.
We say that a swap $(\pi, J) \to (\pi', J)$ is \emph{valid} if
$\udsign(\pi) = \udsign(\pi')$.
\end{definition}

Note that for any swap $(\pi, J) \to (\pi', J)$, we have $\partv(\pi, J) = \partv(\pi', J)$, and for every \emph{valid} swap, $\udsign(\pi) = \udsign(\pi')$.In other words, valid swaps preserve both the up--down signature and the partition vector.

Before proceeding with the proof, we need the following lemma.

\begin{lemma}\label{lem:valid-swaps}
Assume that $(\pi_1, J) - (\pi_2, J) \in N^{\tau}$ (that is,
$\udsign(\pi_1) = \udsign(\pi_2)$ and $\partv(\pi_1, J) = \partv(\pi_2,J)$).
Then $\pi_2$ can be obtained from $\pi_1$ by a sequence of valid swaps.
\end{lemma}

\begin{proof}
We start with $\pi = \pi_1$ and perform at most $n-1$ stages to obtain $\pi_2$. Each stage consists of a sequence of valid swaps. We ensure that after stage~$i$, $\pi$ and $\pi_2$ agree on the positions in $\pi_2^{-1}([i])$. As we only apply valid swaps, we will always have $\udsign(\pi) = \udsign(\pi_1) = \udsign(\pi_2)$ and $\partv(\pi, J) = \partv(\pi_1, J) = \partv(\pi_2, J)$. 

Assume that we have completed $i-1$ stages and that $\pi$ and $\pi_2$ agree on the positions in $F = \pi_2^{-1}([i-1])$. Consider stage $i$. Let $j = \pi_2^{-1}(i)$. Our goal is to ensure that after this stage, $\pi$ and $\pi_2$ still agree on $F$ and that $\pi(j) = \pi_2(j) = i$.

If $j \in J$, then $\pi(j) = \pi_1(j) = \pi_2(j)$, since $\pi|_J = \pi_2|_J$ (both are determined by their partition vectors and $\tau$). Thus, no action is required in this stage. Similarly, if $\pi(j) = i = \pi_2(j)$, then no action is required.

Now assume that $j \notin J$ and $\pi(j) \neq i$. We iteratively find a valid $(j, j')$-swap with $\pi(j') < \pi(j)$ and $j' \notin F$, thereby decreasing the value of $\pi(j)$. We finish this stage when no valid swaps of this form remain.  
Since we never choose $j' \in F$, we do not change the values of $\pi$ on $F$, and thus $\pi$ and $\pi_2$ still agree on $F$, as required. It remains to show that $\pi(j) = i = \pi_2(j)$.  
Assume to the contrary that $\pi(j) > i$, but no swaps of the type described above are available.

Let $C = \pi^{-1}(\{i, \dots, \pi(j) - 1\})$; note that $|C| = \pi(j) - i > 0$.  
For every $j' \in C$, the pair $(j, j')$ is not a valid swap (otherwise, we would have performed it and decreased $\pi(j)$).  
There are four possible reasons for this. In each case below, we assume that none of the previous cases apply.

\begin{itemize}
\item $j' \in J$. 
\item $j'$ is adjacent to $j$. We show that this is not possible. In this case, 
$\pi(j) > \pi(j')$. Since $\pi$ and $\pi_2$ have the same up--down signature, we must have $\pi_2(j) > \pi_2(j')$. However, $\pi_2(j) = i$, hence $\pi_2(j') < i$, and thus $j' \in F$.  
Since $\pi$ and $\pi_2$ agree on $F$, we have $\pi(j') = \pi_2(j') < i$, which contradicts the assumption that $j' \in C$.
\item 
Either $j'-1$ or $j'+1$ \emph{blocks} assigning $\pi(j')$ the value $\pi(j)$, as doing so would alter the up--down signature.
Since the assignment $\pi(j') \leftarrow \pi(j)$ is blocked, there exists $j'' \in \{j'-1, j'+1\}$ such that $\pi(j'')$ lies between $\pi(j')$ and $\pi(j)$.
To summarize, in this case there exists $j'' \in C$ adjacent to $j'$ with $\pi(j'') > \pi(j')$.
\item The assignment $\pi(j) \leftarrow \pi(j')$ is blocked by one of the neighbors $j''$ of $j$, as this assignment would change the up--down signature. We show that this is not possible. In this case, $\pi(j') < \pi(j'') < \pi(j)$.  
Since $\pi$ and $\pi_2$ have the same up--down signature, we conclude that $\pi_2(j'') < \pi_2(j) = i$. Therefore, $j'' \in F$.  
Because $\pi$ and $\pi_2$ agree on $F$, we have $\pi(j'') = \pi_2(j'') < i$, which contradicts $\pi(j'') > \pi(j') \ge i$.
\end{itemize}

We have shown that there are two possibilities: either $j' \in J$, or there exists $j'' \in C$ adjacent to $j'$ with $\pi(j'') > \pi(j')$.  
Define a directed acyclic graph on $C$ by adding an edge from $j'$ to an adjacent $j''$ whenever $\pi(j'') > \pi(j')$.  
Our case analysis implies that all sinks in this graph are in $J$.  

For a given $j'$, consider a path $j_1 \to j_2 \to \dots \to j_t$ in this graph that starts at $j_1= j'$ and terminates at a sink $j_t\in J$.  
We have $\pi(j')=\pi(j_1) < \pi(j_2) < \dots < \pi(j_t) < \pi(j)$.  
Since $\pi$ and $\pi_2$ have the same up--down signatures, it follows that $\pi_2(j') = \pi_2(j_1) < \pi_2(j_2) < \dots < \pi_2(j_t)$.  
Moreover, since $j_t \in J\cap C$, we have 
$$\pi_2(j') < \pi_2(j_t) = \pi(j_t) < \pi(j).$$  
Therefore, $\pi_2(j') < \pi(j)$ for all $j' \in C$.  
Further, if $j' \in C$, then $j' \notin F$ and $j' \neq j$, thus $\pi_2(j') > i$.  
We conclude that for every $j' \in C$, $\pi_2(j') \in \{i+1, \dots, \pi(j) - 1\}$.  
But this is impossible, because 
$$|C| = \pi(j) - i > |\{i+1, \dots, \pi(j) - 1\}|.$$  
We get a contradiction. This concludes the proof.
\end{proof}

Now we are ready to prove Theorem~\ref{thm:main-flag-ideal-property}
\begin{proof}[Proof of Theorem~\ref{thm:main-flag-ideal-property}]
1. Assume first that $P = C$, where $C$ is the transpose of an up--down combination:
\begin{equation}\label{eq:xy-transpose}
C = (x_1 I \oUp y_1 D \oUp x_2 I \oUp y_2 D \oUp \dots \oUp x_N I \oUp y_N D)^\transpose.
\end{equation}
Let $S = (s_1, \dots, s_k)$, where the points are ordered by their $x$-coordinates.
The projections of the points in~$S$ onto the $y$-axis partition the interval~$[0,1]$ into $k+1$ segments of lengths $z_1, \dots, z_{k+1}$.
Let $g_i$ be $x_t$ if the point $s_i$ lies in the region of $[0,1]^2$ corresponding to the term $x_t I$ in $C$, and $y_t$ if $s_i$ lies in the region corresponding to the term $y_t D$.

For $(\pi, J) \in \mathcal{F}^{\tau}$, we derive a formula for $d((\pi, J), (C,S))$.
Write $J=\{j_1, \dots, j_k\}$ with $j_1< \dots < j_k$.
Applying Theorem~\ref{thm:main-qsym-characterization} and Lemma~\ref{lem:labeled-qsym} (and swapping the $x$- and $y$-axes), we obtain that for every $(\pi, J) \in \mathcal{F}^\tau$,
\begin{equation}\label{eq:density-flags}
d((\pi, J), (C, S)) = 
M(\partv(\pi, J)) \prod_{i=1}^{k+1} z_i^{\partv(\pi, J)_i} \cdot \sum_f \prod_{i \notin J} f(i),
\end{equation}
where the sum is over all $\pi$-partition functions~$f$ with $f(j_i) = g_i$.

Now consider two flags $(\pi_1, J)$ and $(\pi_2, J)$ equivalent modulo $\NIdeal$. They have the same signature and the same partition vector, and thus the expressions in~\eqref{eq:density-flags} are identical for both flags. Therefore, their densities are the same in every tagged $T$-permuton~$(C, S)$ with $C$ as in~\eqref{eq:xy-transpose}.
Since ID combinations are dense in the set of all strong IDU permutons, we conclude that the densities of $(\pi_1, J)$ and $(\pi_2, J)$ are equal in every tagged $T$-permuton~$(P, S)$.

\medskip

2. This is immediate:  if $(\pi_1, J) - (\pi_2, J) \in N^{\tau}$, then $\udsign(\pi_1) = \udsign(\pi_2)$, thus $\pi_1 - \pi_2 \in N^{\varnothing}$.

\medskip

3. To verify that ${\mathbb R}N^{\tau}$ is an ideal, we must show that for every
$(\pi, J)$ and $(\pi', J)$ in ${\cal A}^{\tau}$ satisfying
$\udsign(\pi) = \udsign(\pi')$ and $\partv(\pi, J) = \partv(\pi', J)$,
and for every $(\rho, J')$, we have
\begin{equation}\label{eq:ideal-requirement}
(\pi, J) \cdot (\rho, J') \equiv (\pi', J) \cdot (\rho, J') \pmod{\NIdeal}.
\end{equation}
By Lemma~\ref{lem:valid-swaps}, it is sufficient to consider the case where $\pi'$ is obtained from $\pi$ by a valid swap $(i_1, i_2)$. Assume that $i_1 < i_2$.

We express the result of flag multiplication on both sides of \eqref{eq:ideal-requirement} using Theorem~\ref{thm:flag-multiplication-simplified}. Choose a $y$-shuffle $(X, Y)$ for $\pi$ and $\rho$. Then there exists a $y$-shuffle $(X', Y)$ for $\pi'$ and $\rho$ with the same sequence $Y$ for the elements of $\rho$. Namely, $X'$ is obtained from $X$ by swapping $X_{i_1}$ and $X_{i_2}$.

In every $x$-shuffle of $(X, Y)$, the set $J^*$ of positions of common elements is determined by $J$ and $J'$ (namely, the element ranked $i$ in $J^*$ is equal to the sum of the elements ranked $i$ in $J$ and $J'$ minus 1). To specify an $x$-shuffle, we must decide, for each position not occupied by a common element, whether it is occupied by some $X_i$ or some $Y_i$. The same applies to $(X', Y)$.
We proceed as follows. We first fix the positions of $X_{i_1-1}$, $X_{i_1+1}$, $X_{i_2-1}$, and $X_{i_2+1}$ within $Z$ (some of these points may be absent if $i_1 = 1$ or $i_2 = |\pi|$), and then make the choices independently for all positions in three regions $R_1$, $R_2$, and $R_3$ defined below (see Figure~\ref{fig:regions}):

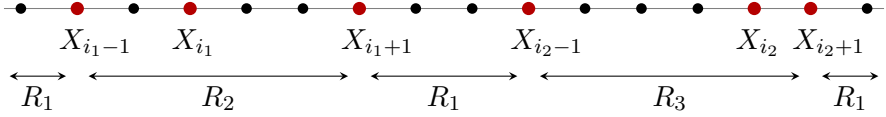
\begin{figure}
\centering
\begin{tikzpicture}[x=0.75cm,y=1cm,>=stealth]
  \draw[thin,gray] (0.7,0)--(16.3,0);
  \foreach \x in {1,...,16}{
    \fill (\x,0) circle (2pt);
  }

  \fill[darkred] (2,0) circle (2.5pt);
  \fill[darkred] (4,0) circle (2.5pt);
  \fill[darkred] (7,0) circle (2.5pt);
  \fill[darkred] (10,0) circle (2.5pt);
  \fill[darkred] (14,0) circle (2.5pt);
  \fill[darkred] (15,0) circle (2.5pt);

  \node[below=4pt] at (2,0) {\ \quad $X_{i_1-1}$};
  \node[below=4pt] at (4,0) {\ $X_{i_1}$};
  \node[below=4pt] at (7,0) {\ \quad $X_{i_1+1}$};
  \node[below=4pt] at (10,0) {\ \quad $X_{i_2-1}$};
  \node[below=4pt] at (14,0) {\ $X_{i_2}$};
  \node[below=4pt] at (15,0) {\ \quad $X_{i_2+1}$};

  \draw[<->] (0.8,-0.9) -- (1.8,-0.9);
  \node[below] at (1.3,-0.9) {$R_1$};

  \draw[<->] (7.2,-0.9) -- (9.8,-0.9);
  \node[below] at (8.5,-0.9) {$R_1$};

  \draw[<->] (15.2,-0.9) -- (16.2,-0.9);
  \node[below] at (15.7,-0.9) {$R_1$};

  \draw[<->] (2.2,-0.9) -- (6.8,-0.9);
  \node[below] at (4.5,-0.9) {$R_2$};

  \draw[<->] (10.2,-0.9) -- (14.8,-0.9);
  \node[below] at (12.5,-0.9) {$R_3$};
\end{tikzpicture}
\caption{The figure shows regions $R_1$, $R_2$, and $R_3$.}
\label{fig:regions}
\end{figure}

\begin{enumerate}
\item The first region, $R_1$, consists of three intervals: positions to the left of $X_{i_1-1}$ (this interval is empty if $i_1 = 1$), positions between $X_{i_1+1}$ and $X_{i_2-1}$, and positions to the right of $X_{i_2+1}$ in $Z$ (this interval is empty if $i_2 = |\pi|$). 
\item The second region, $R_2$, consists of positions between $X_{i_1-1}$ and $X_{i_1+1}$ in $Z$.
\item The third region, $R_3$, consists of positions between $X_{i_2-1}$ and $X_{i_2+1}$ in $Z$.
\end{enumerate}

Importantly, $R_2$ and $R_3$ do not overlap, since $i_2 - i_1 > 1$. Since the positions of $X_{i_1 \pm 1}$ and $X_{i_2 \pm 1}$ are fixed, the sets of elements $X_i$ and $Y_j$ assigned to each interval of $R_1$ and to each of the regions $R_2$ and $R_3$ are determined. However, the interleaving of elements of $X$ and $Y$ within each interval and region is not yet fixed.

We first decide which positions in $R_1$ are occupied by elements of $X$ and which by elements of $Y$. All elements of $J^*$ within $R_1$ in $Z$ (resp.\ $Z'$) are occupied by the common elements of $X$ and $Y$ (resp.\ $X'$ and $Y$). Aside from $J^*$, we choose the same interleaving of elements of $X$ and $Y$ within $R_1$ in $Z$ and of elements of $X'$ and $Y$ within $R_1$ in $Z'$. This completely determines the $x$-shuffles $Z$ and $Z'$ within $R_1$, and therefore their up--down signatures within $R_1$. Since $X$ and $X'$ coincide in $R_1$, and $X_{i_1 \pm 1} = X'_{i_1 \pm 1}$ and $X_{i_2 \pm 1} = X'_{i_2 \pm 1}$ (with the signs chosen consistently; some of these points may be absent if $i_1 = 1$ or $i_2 = |\pi|$), the up--down signatures within $R_1$ are identical.

It remains to prove that the multiset of possible up--down signatures for the region $R_2$ is the same for $(X, Y)$ and $(X', Y)$, and that an analogous statement holds for $R_3$. 

The proofs for $R_2$ and $R_3$ are identical, so we provide only the argument for $R_2$.
We consider auxiliary sequences $A$ and $B$:
\begin{itemize}
\item $A$ consists of $A_1 = X_{i_1-1}$, $A_2 = X_{i_1}$, and  $A_3  = X_{i_1+1}$; element $X_{i_1 - 1}$ is absent if $i_1 = 1$.
\item $B$ consists of $X_{i_1-1}$ (absent if $i_1=1$), followed by all elements of $Y$ assigned to $R_2$, and finally $X_{i_1+1}$. We do not include $X_{i_1-1}$ or $X_{i_1+1}$ twice if they also appear in $Y$.
\end{itemize}
Note that the elements $X_{i_1-1}$ (if it exists) and $X_{i_1+1}$ are the only common elements between sequences $A$ and $B$.

All possible signatures for $R_2$ are precisely the signatures of all possible $x$-shuffles of $A$ and $B$. We similarly define $A'$ and $B' = B$ for $X'$ and $Y$.  
Since $X$ and $X'$ have the same up--down signature, $X_{i_1-1} = X'_{i_1-1}$, and $X_{i_1+1} = X'_{i_1+1}$, the patterns defined by $A$ and $A'$ are equal (the relative order of $A_1 = A'_1$ and $A_3 = A'_3$ is the same in both $A$ and $A'$, and the relative order of $A_3$ (resp.~$A'_3$) with the remaining elements of $A$ (resp.~$A'$) is determined by the up--down signature).  
By Corollary~\ref{cor:general-shuffles}, the signatures of $x$-shuffles of $(A, B)$ and of $(A', B)$ are equal as multisets. This concludes the proof.
\end{proof}

\section*{Acknowledgements}

This project was supported by NSF awards CCF-1955173 and ECCS-2216899. We are grateful to Google for providing access to a \textit{c4-highcpu-8} Google Cloud instance, which we used to compute approximation factors for single-predicate CSPs of arity~4.
We thank the organizers and participants of Dagstuhl Seminar 25211, \emph{The Constraint Satisfaction Problem: Complexity and Approximability}, for providing a stimulating environment in which this project began.
In his lectures at that seminar, Amey Bhangale raised the question of the approximability of completely satisfiable ordering CSPs and rekindled our interest in the topic.
We also thank the anonymous referees for their helpful comments.

\bibliographystyle{abbrv}
\bibliography{bibliography}

@article{Niven68,
  title={A combinatorial problem of finite sequences},
  author={Niven, Ivan},
  journal={Nieuw Arch. Wisk},
  volume={16},
  number={3},
  pages={116--123},
  year={1968}
}

@article{Petrov,
  author       = {Fedor Petrov},
  title        = {Correcting continuous hypergraphs},
  journal      = {Algebra i Analiz},
  year         = {2016},
  volume       = {28},
  number       = {6},
  pages        = {84--90},
  note         = {English transl.: \emph{St. Petersburg Math. J.}, 28(6):783--787, 2017. 
                  An alternative preprint version is available at \url{https://arxiv.org/abs/1309.3795}.},
}

@article{Poirier98,
author = {Stéphane Poirier},
title = {Cycle type and descent set in wreath products},
journal = {Discrete Mathematics},
volume = {180},
number = {1},
pages = {315--343},
year = {1998},
note = {Proceedings of the 7th Conference on Formal Power Series and Algebraic Combinatorics},
}

@article{Ramsey,
  title={On a Problem of Formal Logic},
  author={Ramsey, FP},
  journal={Proceedings of the London Mathematical Society},
  volume={30},
  pages={264--285},
  year={1930}
}

@article{GHMRC,
  title={Beating the random ordering is hard: Every ordering {CSP} is approximation resistant},
  author={Guruswami, Venkatesan and H{\aa}stad, Johan and Manokaran, Rajsekar and Raghavendra, Prasad and Charikar, Moses},
  journal={SIAM Journal on Computing},
  volume={40},
  number={3},
  pages={878--914},
  year={2011},
}

@article{Mak,
title = {Simple linear time approximation algorithm for betweenness},
author = {Yury Makarychev},
journal = {Operations Research Letters},
volume = {40},
number = {6},
pages = {450--452},
year = {2012},
}

@InCollection{MMsurvey,
  author =	{Makarychev, Konstantin and Makarychev, Yury},
  title =	{Approximation Algorithms for {CSPs}},
  booktitle =	{The Constraint Satisfaction Problem: Complexity and Approximability},
  pages =	{287--325},
  series =	{Dagstuhl Follow-Ups},
  year =	{2017},
  volume =	{7},
  editor =	{Krokhin, Andrei and Zivny, Stanislav},
  publisher =	{Schloss Dagstuhl -- Leibniz-Zentrum f{\"u}r Informatik},
 }

@book{Kallenberg05,
  title={Probabilistic symmetries and invariance principles},
  author={Kallenberg, Olav},
  year={2005},
  publisher={Springer}
}

@article{Kallenberg92,
  title={Symmetries on random arrays and set-indexed processes},
  author={Kallenberg, Olav},
  journal={Journal of Theoretical Probability},
  volume={5},
  number={4},
  pages={727--765},
  year={1992},
}

@inproceedings{MMZ15,
  title={Satisfiability of ordering {CSPs} above average is fixed-parameter tractable},
  author={Makarychev, Konstantin and Makarychev, Yury and Zhou, Yuan},
  booktitle={Proceedings of the Symposium on Foundations of Computer Science},
  pages={975--993},
  year={2015},
}

@article{GIMY,
  title={Every ternary permutation constraint satisfaction problem parameterized above average has a kernel with a quadratic number of variables},
  author={Gutin, Gregory and Van Iersel, Leo and Mnich, Matthias and Yeo, Anders},
  journal={Journal of Computer and System Sciences},
  volume={78},
  number={1},
  pages={151--163},
  year={2012},
}

@inproceedings{MakBounded,
  title={Local Search is Better than Random Assignment for Bounded Occurrence Ordering {$k$-CSPs}},
  author={Makarychev, Konstantin},
  booktitle={Proceedings of the International Symposium on Theoretical Aspects of Computer Science},
  pages={139},
  year={2013}
}

@article{BK10,
  title={The complexity of temporal constraint satisfaction problems},
  author={Bodirsky, Manuel and K{\'a}ra, Jan},
  journal={Journal of the ACM (JACM)},
  volume={57},
  number={2},
  pages={1--41},
  year={2010},
}

@article{CS98,
  title={A geometric approach to betweenness},
  author={Chor, Benny and Sudan, Madhu},
  journal={SIAM Journal on Discrete Mathematics},
  volume={11},
  number={4},
  pages={511--523},
  year={1998},
}

@inproceedings{GZ12,
  title={Approximating bounded occurrence ordering {CSPs}},
  author={Guruswami, Venkatesan and Zhou, Yuan},
  booktitle={Proceedings of the International Workshop on Approximation Algorithms for Combinatorial Optimization},
  pages={158--169},
  year={2012},
}

@article{Seymour,
  title={Packing directed circuits fractionally},
  author={Seymour, Paul D.},
  journal={Combinatorica},
  volume={15},
  number={2},
  pages={281--288},
  year={1995},
}

@article{Razborov,
  title={Flag algebras},
  author={Razborov, Alexander A},
  journal={The Journal of Symbolic Logic},
  volume={72},
  number={4},
  pages={1239--1282},
  year={2007},
}

@article{BHLPUV15,
  title={Minimum number of monotone subsequences of length 4 in permutations},
  author={Balogh, J{\'o}zsef and Hu, Ping and Lidick{\`y}, Bernard and Pikhurko, Oleg and Udvari, Bal{\'a}zs and Volec, Jan},
  journal={Combinatorics, Probability and Computing},
  volume={24},
  number={4},
  pages={658--679},
  year={2015},
}

@article{crudele2024six,
  title={Six Permutation Patterns Force Quasirandomness},
  author={Crudele, Gabriel and Dukes, Peter and Noel, Jonathan A},
  journal={Discrete Analysis},
  year={2024}
}

@article{BKfast10,
  title={A fast algorithm and Datalog inexpressibility for temporal reasoning},
  author={Bodirsky, Manuel and K{\'a}ra, Jan},
  journal={ACM Transactions on Computational Logic (TOCL)},
  volume={11},
  number={3},
  pages={1--21},
  year={2010},
}

@inproceedings{Raghavendra,
  title={Optimal algorithms and inapproximability results for every {CSP}?},
  author={Raghavendra, Prasad},
  booktitle={Proceedings of the Symposium on Theory of Computing},
  pages={245--254},
  year={2008}
}

@inproceedings{Bulatov,
  title={A dichotomy theorem for nonuniform {CSPs}},
  author={Bulatov, Andrei A},
  booktitle={Proceedings of the Symposium on Foundations of Computer Science},
  pages={319--330},
  year={2017},  
}

@article{Zhuk,
  title={A proof of the CSP dichotomy conjecture},
  author={Zhuk, Dmitriy},
  journal={Journal of the ACM (JACM)},
  volume={67},
  number={5},
  pages={1--78},
  year={2020},
}

@inproceedings{Schaefer,
  title={The complexity of satisfiability problems},
  author={Schaefer, Thomas J},
  booktitle={Proceedings of the ACM Symposium on Theory of Computing},
  pages={216--226},
  year={1978}
}

@inproceedings{DBHPZ23,
  author       = {Joshua Brakensiek and
                  Neng Huang and
                  Aaron Potechin and
                  Uri Zwick},
  title        = {Separating {Max 2-AND}, {Max} {DI-CUT} and {Max} {CUT}},
  booktitle    = {Proceedings of the Symposium on Foundations of Computer Science},
  pages        = {234--252},
  year         = {2023},}

@inproceedings{BNZ,
  title={Tight approximability of {Max 2-SAT} and relatives, under {UGC}},
  author={Brakensiek, Joshua and Huang, Neng and Zwick, Uri},
  booktitle={Proceedings of the Symposium on Discrete Algorithms},
  pages={1328--1344},
  year={2024},
}

@inproceedings{LLZ,
  title={Improved rounding techniques for the {Max 2-SAT} and {Max DI-CUT} problems},
  author={Lewin, Michael and Livnat, Dror and Zwick, Uri},
  booktitle={Proceedings of the International Conference on Integer Programming and Combinatorial Optimization},
  pages={67--82},
  year={2002},
}

@inproceedings{MM12,
  title={Approximation Algorithm for Non-{B}oolean {Max $k$-CSP}},
  author={Makarychev, Konstantin and Makarychev, Yury},
  booktitle={Proceedings of the International Workshop on Approximation Algorithms for Combinatorial Optimization},
  pages={254--265},
  year={2012},
  organization={Springer}
}

@inproceedings{ACMM,
  title={{$O(\sqrt{\log n})$} approximation algorithms for {Min UnCut, Min 2CNF Deletion,} and directed cut problems},
  author={Agarwal, Amit and Charikar, Moses and Makarychev, Konstantin and Makarychev, Yury},
  booktitle={Proceedings of the Symposium on Theory of Computing},
  pages={573--581},
  year={2005}
}

@inproceedings{CMM1,
  title={Near-optimal algorithms for unique games},
  author={Charikar, Moses and Makarychev, Konstantin and Makarychev, Yury},
  booktitle={Proceedings of the Symposium on Theory of Computing},
  pages={205--214},
  year={2006}
}

@inproceedings{CMM2,
  author={Chlamtac, Eden and Makarychev, Konstantin and Makarychev, Yury},
  booktitle={Proceedings of the  Symposium on Foundations of Computer Science}, 
  title={How to Play Unique Games Using Embeddings}, 
  year={2006},
  pages={687--696},
}

@article{GW,
  title={Improved approximation algorithms for maximum cut and satisfiability problems using semidefinite programming},
  author={Goemans, Michel X and Williamson, David P},
  journal={Journal of the ACM (JACM)},
  volume={42},
  number={6},
  pages={1115--1145},
  year={1995},
}

@inproceedings{ABZ05,
  title={Improved approximation algorithms for {Max NAE-SAT} and {Max SAT}},
  author={Avidor, Adi and Berkovitch, Ido and Zwick, Uri},
  booktitle={Proceedings of the International Workshop on Approximation and Online Algorithms},
  pages={27--40},
  year={2005},
  organization={Springer}
}

@article{HKMRS,
  title={Limits of permutation sequences},
  author={Hoppen, Carlos and Kohayakawa, Yoshiharu and Moreira, Carlos Gustavo and R{\'a}th, Bal{\'a}zs and Sampaio, Rudini Menezes},
  journal={Journal of Combinatorial Theory, Series B},
  volume={103},
  number={1},
  pages={93--113},
  year={2013},
}

@article{BBFGMP,
  title={Universal limits of substitution-closed permutation classes},
  author={Bassino, Fr{\'e}d{\'e}rique and Bouvel, Mathilde and F{\'e}ray, Valentin and Gerin, Lucas and Maazoun, Micka{\"e}l and Pierrot, Adeline},
  journal={Journal of the European Mathematical Society},
  volume={22},
  number={11},
  pages={3565--3639},
  year={2020}
}

@article{Stembridge,
  title={Enriched $P$-partitions},
  author={John R. Stembridge},
  journal={Transactions of the American Mathematical Society},
  volume={349},
  number={2},
  pages={763--788},
  year={1997}
}

@inproceedings{BKM25,
  title={On Approximability of Satisfiable $k$-{CSPs}: {V}},
  author={Bhangale, Amey and Khot, Subhash and Minzer, Dor},
  booktitle={Proceedings of the Symposium on Theory of Computing},
  pages={62--71},
  year={2025}
}

@article{Trevisan,
  title={Approximating satisfiable satisfiability problems},
  author={Trevisan, Luca},
  journal={Algorithmica},
  volume={28},
  number={1},
  pages={145--172},
  year={2000},
}

@inproceedings{KTW14,
  title={A characterization of strong approximation resistance},
  author={Khot, Subhash and Tulsiani, Madhur and Worah, Pratik},
  booktitle={Proceedings of the Symposium on Theory of Computing},
  pages={634--643},
  year={2014}
}

@article{AM09,
  title={Approximation resistant predicates from pairwise independence},
  author={Austrin, Per and Mossel, Elchanan},
  journal={Computational Complexity},
  volume={18},
  number={2},
  pages={249--271},
  year={2009},
}

@article{Chan16,
  title={Approximation resistance from pairwise-independent subgroups},
  author={Chan, Siu On},
  journal={Journal of the ACM (JACM)},
  volume={63},
  number={3},
  pages={1--32},
  year={2016},
}

@inproceedings{Z98,
  title={Finding almost-satisfying assignments},
  author={Zwick, Uri},
  booktitle={Proceedings of the Symposium on Theory of Computing},
  pages={551--560},
  year={1998}
}

@inproceedings{bhangale2021optimal,
  title={Optimal inapproximability of satisfiable $k$-{LIN} over non-abelian groups},
  author={Bhangale, Amey and Khot, Subhash},
  booktitle={Proceedings of the Symposium on Theory of Computing},
  pages={1615--1628},
  year={2021}
}

@article{CMM2009near,
  title={Near-optimal algorithms for maximum constraint satisfaction problems},
  author={Charikar, Moses and Makarychev, Konstantin and Makarychev, Yury},
  journal={ACM Transactions on Algorithms (TALG)},
  volume={5},
  number={3},
  pages={1--14},
  year={2009},
}

@inproceedings{CMM07advantage,
  title={On the advantage over random for maximum acyclic subgraph},
  author={Charikar, Moses and Makarychev, Konstantin and Makarychev, Yury},
  booktitle={Proceedings of the Symposium on Foundations of Computer Science},
  pages={625--633},
  year={2007},
}

@misc{ordering-csp-repo,
  author = {Yury Makarychev},
  title = {Ordering Constraint Satisfaction Problems (arity 4)},
  howpublished = {\url{https://github.com/ymakarychev/ordering-csp}},
  note = {GitHub repository},
  year = {2026}
}

@article{ENSS98,
  title={Approximating minimum feedback sets and multicuts in directed graphs},
  author={Even, Guy and Naor, Joseph and Schieber, Baruch and Sudan, Madhu},
  journal={Algorithmica},
  volume={20},
  number={2},
  pages={151--174},
  year={1998},
}

\appendix
\section{Shuffle-closed predicates}\label{sec:shuffle-closed}
In this section, we show that every shuffle-closed predicate is a Not-First predicate (see Definition~\ref{def:not-first}).
We now state the definition of a shuffle-closed predicate for ordering CSPs.
\begin{definition}\label{def:shuffle-operation}
An ordering predicate $\varphi$ of arity $k$ is \emph{shuffle-closed} if for every two permutations $\alpha, \beta \in \Sat(\varphi)$ and every $t \in [k-1]$, the permutation $\sigma$ defined as follows also belongs to $\Sat(\varphi)$:
\begin{itemize}
    \item $\sigma(i) = \alpha(i)$ if $\alpha(i) \in [t]$;
    \item $\sigma(i) > t$ if $\alpha(i) \notin [t]$;
    \item the relative order of values $\sigma(i)$ for $i \notin \alpha^{-1}([t])$ is the same as in~$\beta$.
\end{itemize}
\end{definition}
Although we do not explicitly require that the property above holds for $t = 0$ and $t = k$, it trivially does, since for $t = 0$ we get $\sigma = \beta$, and for $t = k$ we get $\sigma = \alpha$.

\begin{theorem}
Every shuffle-closed ordering predicate is a Not-First predicate.
\end{theorem}
\begin{proof}
Consider the relaxation $\varphi_{L}$ of $\varphi$ (see Section~\ref{sec:satisfiable-csps}). We will show that $\varphi = \varphi_L$ and thus $\varphi$ is a Not-First predicate, as required.

Fix a permutation $\pi \in \Sat(\varphi_{L})$. We will show that $\pi \in \Sat(\varphi)$.  To this end, we prove by induction on $i$ that there exists a permutation $\pi_i \in \Sat(\varphi)$ that agrees with $\pi$ on $\pi^{-1}([i])$. The statement trivially holds for $i = 0$.  

Assume it holds for some $i < k$; we prove it for $i + 1$.  
Let $j = \pi^{-1}(i+1)$.  
Consider the constraint $\psi = \NFirst_{k-i}(x_j, x_{a_1}, \dots, x_{a_{k-i-1}})$, where $a_1, \dots, a_{k-i-1}$ are the elements of $[k] \setminus \pi^{-1}([i+1])$.  
Note that $\pi$ does not satisfy this constraint, since $\pi(j) = i+1$ is smaller than the elements of $\{\pi(a_r)\}_{r=1}^{k-i-1} = [k] \setminus [i+1]$.  
However, since $\pi\in \Sat(\varphi_L)$, $\pi$ satisfies all Not-First constraints implied by $\varphi$. We conclude that the constraint $\psi$ is not implied by $\varphi$ or, in other words, not a relaxation of $\varphi$.  

Therefore, there exists a permutation $\beta \in \Sat(\varphi)$ that does not satisfy $\psi$. 
We apply the shuffle operation to $\alpha = \pi_i$ and $\beta$ with $t = i$. Denote the obtained permutation by $\pi_{i+1}$.  Since $\pi_i, \beta\in \Sat(\varphi)$ and $\varphi$ is shuffle-closed, 
$\pi_{i+1} \in \Sat(\varphi)$.

By the definition of the shuffle operation, $\pi_{i+1}$ agrees with $\pi_i$, and thus with $\pi$, on $\pi^{-1}([i])$.  
Since $\pi(j),\pi(a_1),\dots, \pi(a_{k-i-1})\notin [i]$, by item 3 of Definition~\ref{def:shuffle-operation}, the relative orderings of 
$$\pi_{i+1}(j),\pi_{i+1}(a_1),\dots, \pi_{i+1}(a_{k-i-1}) \quad\text{and of}\quad  
\beta(j),\beta(a_1),\dots, \beta(a_{k-i-1})$$
are the same. Note that $\beta(j)$ is smaller than all $\beta(a_r)$, since $\beta$ does not satisfy $\psi$. Thus, $\pi_{i+1}(j)$ is also smaller than all $\pi_{i+1}(a_r)$. Now, there are $k -i$ distinct values $\pi_{i+1}(j), \pi_{i+1}(a_1), \dots, \pi_{i+1}(a_{k-i-1})$ and all of them lie in $[k] \setminus [i]$,
since $\pi_{i+1}(\pi_i^{-1}([i])) = [i]$. Therefore, 
$$\{\pi_{i+1}(j), \pi_{i+1}(a_1), \dots, \pi_{i+1}(a_{k-i-1})\} = [k] \setminus [i].$$
Further, $\pi_{i+1}(j)$ is the smallest in the set. We conclude that $\pi_{i+1}(j) = i+1 = \pi(j)$. Thus, $\pi_{i+1}$ agrees with $\pi$ on $\pi^{-1}([i+1])$, as required.

We proved that $\pi_i\in\Sat(\varphi)$ agrees with $\pi$ on $\pi^{-1}([i])$ for every $i$. In particular, $\pi_k$ agrees with $\pi$ everywhere; that is, $\pi_k = \pi$.
We get $\pi = \pi_k\in \Sat(\varphi)$.  
We showed that every $\pi\in\Sat(\varphi_L)$ is in $\Sat(\varphi)$ and thus $\Sat(\varphi) \supseteq \Sat(\varphi_{L})$. Since $\varphi_L$ is a relaxation of $\varphi$, $\Sat(\varphi) \subseteq \Sat(\varphi_{L})$. Thus, $\Sat(\varphi)= \Sat(\varphi_{L})$ and $\varphi = \varphi_{L}$.
\end{proof}

\end{document}